\documentclass[prx,aps,superscriptaddress,footinbib,twocolumn]{revtex4-2}
\usepackage{amsthm,amsmath,amssymb}
\usepackage{mathrsfs}

\newcommand{\ket}[1]{\vert{ #1 }\rangle}

\usepackage{amsthm}
\newtheorem{theorem}{Theorem}
\newtheorem{proposition}{Proposition}
\newtheorem{lemma}{Lemma}
\newtheorem{corollary}{Corollary}

\theoremstyle{definition}
\newtheorem{definition}{Definition}
\theoremstyle{remark}

\usepackage{amssymb}
\usepackage{subfigure}
\usepackage{booktabs}
\usepackage{multirow}
\usepackage{graphicx}
\usepackage{epstopdf}
\newcommand{\figpath}{.}

\usepackage{algorithm}
\usepackage[noend]{algpseudocode}

\usepackage{natbib}

\usepackage[colorlinks=true,citecolor=blue,linkcolor=magenta]{hyperref}

\usepackage{xcolor}

\newcommand{\RED}[1]{{#1}}

\usepackage{comment}

\begin{document}

\title{Low-density parity-check representation of fault-tolerant quantum circuits}

\author{Ying Li}
%\email{yli@gscaep.ac.cn}
\affiliation{Graduate School of China Academy of Engineering Physics, Beijing 100193, China}

\begin{abstract}
In fault-tolerant quantum computing, quantum algorithms are implemented through quantum circuits capable of error correction. These circuits are typically constructed based on specific quantum error correction codes, with consideration given to the characteristics of the underlying physical platforms. Optimising these circuits within the constraints of today's quantum computing technologies, particularly in terms of error rates, qubit counts, and network topologies, holds substantial implications for the feasibility of quantum applications in the near future. This paper presents a toolkit for designing and analysing fault-tolerant quantum circuits. We introduce a framework for representing stabiliser circuits using classical low-density parity-check (LDPC) codes. Each codeword in the representation corresponds to a quantum-mechanical equation regarding the circuit, formalising the correlations utilised in parity checks and delineating logical operations within the circuit. Consequently, the LDPC code provides a means of quantifying fault tolerance and verifying logical operations. We outline the procedure for generating LDPC codes from circuits using the Tanner graph notation, alongside proposing graph-theory tools for constructing fault-tolerant quantum circuits from classical LDPC codes. These findings offer a systematic approach to applying classical error correction techniques in optimising existing fault-tolerant protocols and developing new ones. 

\RED{As an example, we develop a resource-efficient scheme for universal fault-tolerant quantum computing on hypergraph product codes based on the LDPC representation. }
\end{abstract}

\maketitle

\section{Introduction}

Quantum error correction stands as an indispensable strategy in addressing the effects of noise inherent in quantum computers \cite{Terhal2015,Campbell2017,Roffe2019}. The development of quantum error correction codes owes much to classical error correction codes. For instance, Calderbank-Shor-Steane (CSS) codes constitute a diverse category of quantum error correction codes, encompassing the most noteworthy variants, and they are derived from classical linear codes \cite{Calderbank1996,Steane1996a,Steane1996,Bravyi1998,Freedman1998,Bacon2006,Aly2006,Bombin2006,Bombin2007}. Recently, remarkable progress has been made in quantum low-density parity-check (LDPC) codes \cite{Gottesman2014,Breuckmann2021}. In quantum LDPC codes, each parity check, the fundamental operation of error correction, involves only a small number of qubits, with each qubit implicated in just a few parity checks. Such quantum error correction codes not only offer implementation advantages but also allow for a commendable encoding rate, thereby utilising considerably fewer physical qubits compared to the surface code. A major avenue for constructing quantum LDPC codes lies in leveraging classical LDPC codes, which has yielded significant achievements, such as hypergraph product codes \cite{Tillich2014,Kovalev2012,Kovalev2013,Bravyi2013,Hastings2016,Evra2020,Breuckmann2021a,Panteleev2022,Bravyi2024}. 

In the realm of quantum computing, quantum error correction codes are realised through fault-tolerant quantum circuits \cite{Shor1997,Preskill1997,Steane1999}. These circuits are pivotal as they determine which quantum codes are feasible candidates for enabling fault-tolerant quantum computing. An effective circuit must optimise error correction capabilities while minimising the impact of noise in primitive gates, ultimately achieving the highest possible fault-tolerance threshold. Additionally, consideration must be given to the topology of the qubit network to ensure the circuit's feasibility on the given platform. For example, most superconducting systems only support short-range interactions, whereas neutral atoms and network systems excel in facilitating long-range interactions \cite{Acharya2023,Evered2023,Bluvstein2023,Xu2024,Nickerson2013,Bravyi2022}. Apart from error correction, fault-tolerant circuits must also perform certain operations on logical qubits. 

To address these challenges, several methodologies have been developed for composing fault-tolerant quantum circuits. These include transversal gates, which are a standard approach in fault-tolerant quantum computing \cite{Zeng2007,Bravyi2005,Eastin2009,Paetznick2013}, as well as protocols like cat state, error-correcting teleportation and flagged circuit, proposed to minimise errors and control their propagation \cite{Shor1997,Knill2005,Chao2018,Yamasaki2024}. Additionally, code deformation has emerged as another notable strategy \cite{Bombin2009,Bombin2011}. In the context of the surface code, effective methods for operating logical qubits include defect braiding \cite{Raussendorf2006,Raussendorf2007,Fowler2009,Fowler2012}, twist operation \cite{Bombin2010,Hastings2015,Brown2017,Yoder2017} and lattice surgery \cite{Horsman2012,Litinski2019,Vuillot2019}. Furthermore, high-fidelity encoding of magic states can be achieved through post-selection \cite{Li2015,Ye2023}. Recent advancements include protocols for quantum LDPC codes \cite{Gottesman2014,Krishna2021,Cohen2022,Xu2024,Lavasani2019,Breuckmann2024}. However, a unified framework for fault-tolerant quantum circuits remains in demand. In this regard, foliated quantum codes offer a way of composing fault-tolerant circuits through cluster states \cite{Bolt2016,Nickerson2018,Newman2020,Brown2020}, while spacetime codes are quantum codes generated from circuits \cite{Bacon2017,Gottesman2022,Delfosse2023,Fu2024}. Additionally, code deformation can be used to express fault-tolerant circuits \cite{Beverland2024}, and ZX calculus provides a rigorous graphical theory of quantum circuits \cite{Beaudrap2020,Chancellor2023,Bombin2024}. Thus, it is opportune to develop a paradigm of formalising fault-tolerant quantum circuits within the context of classical linear codes. 

\begin{figure}[tbp]
\centering
\includegraphics[width=\linewidth]{\figpath/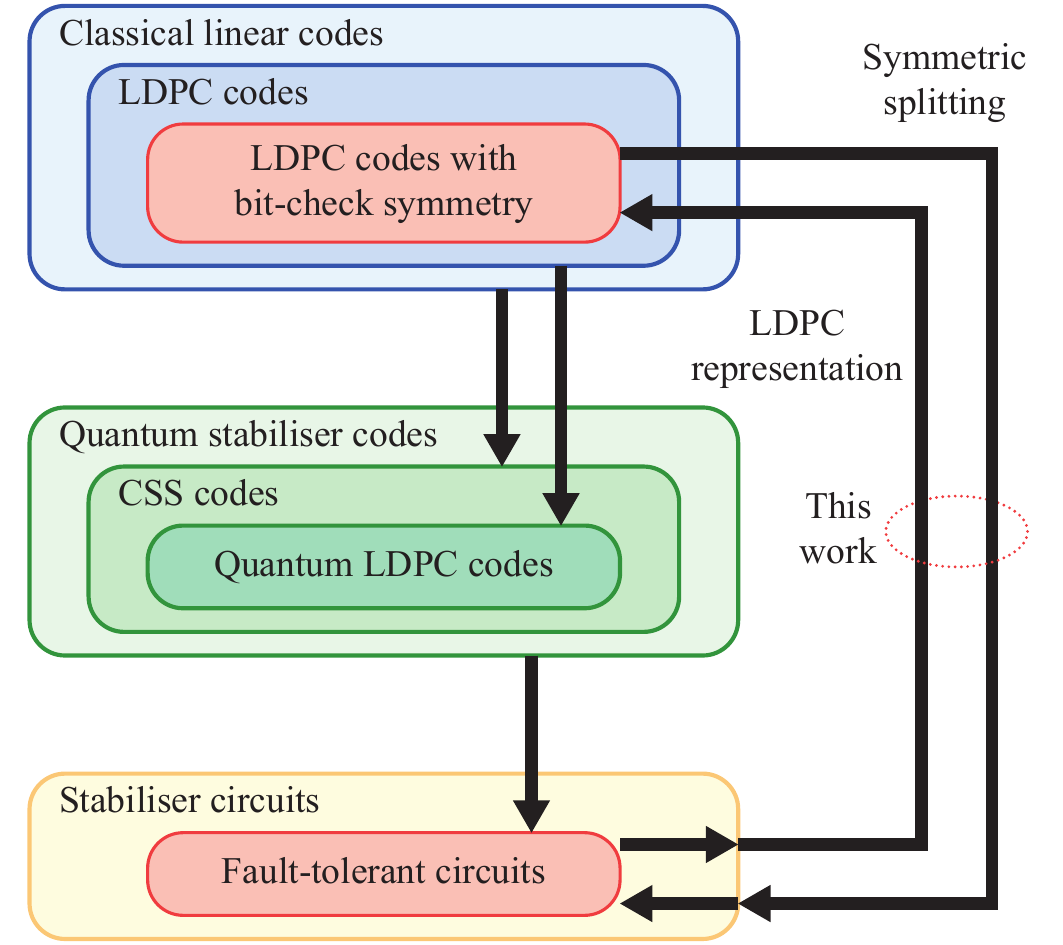}
\caption{
Role of classical error correction codes in fault-tolerant quantum computing. A Calderbank-Shor-Steane (CSS) code is constructed from two compatible classical linear codes. Many quantum low-density parity-check (LDPC) codes originate from classical LDPC codes, through methods such as hypergraph product. Fault-tolerant quantum circuits are usually designed based on specific quantum codes. This work establishes the equivalence between fault-tolerant quantum circuits and classical LDPC codes with bit-check symmetry. General stabiliser circuits can be mapped to LDPC codes with a maximum vertex degree of three on Tanner graphs, forming the LDPC representation. Conversely, general LDPC codes with symmetry can be converted into stabiliser circuits using symmetric splitting, a technique proposed to reduce the maximum vertex degree to three. The stabiliser circuit is fault-tolerant when the LDPC code demonstrates a favourable circuit code distance. 
}
\label{fig:category}
\end{figure}

This work establishes an equivalence between fault-tolerant quantum circuits and a specific category of classical LDPC codes, as depicted in Fig. \ref{fig:category}. These LDPC codes possess bit-check symmetry: Upon the deletion of certain columns from the check matrix, the resulting matrix is symmetric. All stabiliser circuits can be represented by LDPC codes with this symmetry, and conversely, LDPC codes exhibiting such symmetry can be converted into stabiliser circuits. \RED{In this equivalence, codewords in the LDPC code represent correlations in the corresponding circuit. These correlations are related to the two fundamental tasks of a fault-tolerant circuit, i.e. error correction and logical operation. With the code decomposed according to the tasks of codewords, the LDPC code characterises the fault tolerance of the corresponding circuit. To quantify fault tolerance, a measure termed circuit code distance is defined on the code and its decomposition.} If a stabiliser circuit is fault-tolerant, the LDPC code demonstrates a favourable circuit code distance, and vice versa. This circuit-code equivalence offers a universal approach for harnessing algebraic and graphical methods from classical codes in fault-tolerant quantum computing, particularly in the design and analysis of fault-tolerant quantum circuits.

\RED{\section{Summary of results}

Our results can be divided into three parts: The map from stabiliser circuits to linear codes, maps from linear codes to stabiliser circuits, and examples of the application. In the examples, we apply the LDPC representation to analyse existing circuits and develop new circuits. 

\subsection{Map from circuits to codes}

The map from stabiliser circuits to linear codes is presented in Sec. \ref{sec:representation}, which is based on the binary-vector representation of Pauli operators \cite{Gottesman1997,Calderbank1998}. Specifically, we can map an arbitrary stabiliser circuit to a binary check matrix $\mathbf{A}$. We construct the check matrix of a circuit utilising the Tanner graph notation \cite{Tanner1981}, in which check matrices are illustrated as bipartite graphs. 
%First, we give the Tanner graphs representing primitive stabiliser-circuit operations. Then, by combining these primitive graphs, we can generate the Tanner graph for an arbitrary stabiliser circuit. 
Through Tanner graphs, we observe that the code of $\mathbf{A}$ is in the LDPC category. 

The linear code has a physical meaning regarding the time evolution of Pauli operators through the circuit \cite{Gottesman1998}. In quantum mechanics, the state evolves in time while operators are constant in the Schr\"{o}dinger representation, and the state is constant while operators evolve in time in the Heisenberg representation; it turns out that the two representations predict the same time evolution of the operator mean values. Here, we let the state and operators evolve in time simultaneously in a way such that operator mean values are constant. Let $\rho_0$ and $\rho_T$ be the initial and final states of the circuit, respectively. If the circuit consists of only Clifford gates, for all input Pauli operators $\sigma_{in}$, there exists an output Pauli operator $\sigma_{out}$ satisfying the equation 
\begin{eqnarray}
\mathrm{Tr}\sigma_{in}\rho_0
= \nu\mathrm{Tr}\sigma_{out}\rho_T,
\end{eqnarray}
where the factor $\nu = \pm 1$ depends on the gates and input Pauli operator. When the circuit includes the initialisation and measurement, only a subset of input Pauli operators have corresponding output Pauli operators that can satisfy the equation. Such pairs of input and output Pauli operators correspond to codewords $c\in\mathrm{ker}\mathbf{A}$ of the circuit. 

In Sec. \ref{sec:codewords}, we show that each codeword of $\mathbf{A}$ corresponds to an equation between operator mean values in the initial and final states, called the error-free codeword equation; see Theorem \ref{the:errorfree}. Therefore, each codeword represents a correlation in the circuit, specifically, a preserved quantity. Such quantities are important for error correction and logical operation. 

In the quantum error correction, we detect errors through the preserved quantities. Thinking of the case that the initial state is prepared in an eigenstate of $\sigma_{in}$ with the eigenvalue $+1$, i.e. $\mathrm{Tr}\sigma_{in}\rho_0 = 1$, we can detect errors by measuring $\sigma_{out}$. If the measurement outcome is not $\nu$ (i.e. $\nu\mathrm{Tr}\sigma_{out}\rho_T\neq 1$), we know that there are errors in the circuit. Such a mechanism is summarised in Theorem \ref{the:generalised}, in which we introduced the generalised codeword equation regarding circuits with errors. Given the set of preserved quantities for detecting errors, we can list all the corresponding codewords in a matrix $\mathbf{B}$; see Sec. \ref{sec:error_correction}. This matrix, called the error-correction check matrix, represents the error correction protocol in the circuit of $\mathbf{A}$. Since $\mathbf{B}$ is formed of codewords in $\mathrm{ker}\mathbf{A}$, it must satisfy the condition $\mathbf{A}\mathbf{B}^\mathrm{T} = 0$. 

The logical operation is also represented by preserved quantities in the circuit. When $\sigma_{in}$ and $\sigma_{out}$ are two logical operators, the corresponding equation expresses that the logical operator $\sigma_{in}$ is transformed into the logical operator $\sigma_{out}$ by the circuit. Given the set of preserved quantities for logical operators, we can list all the corresponding codewords in a matrix $\mathbf{L}$; see Sec. \ref{sec:logical_operation}. This matrix, called the logical generator matrix, represents the logical operation realised in the circuit of $\mathbf{A}$. Similar to $\mathbf{B}$, the matrix $\mathbf{L}$ must satisfy the condition $\mathbf{A}\mathbf{L}^\mathrm{T} = 0$. 

In summary, to represent a stabiliser circuit, we use a linear code, i.e. its check matrix $\mathbf{A}$. However, to represent the error correction and logical operation in the circuit, we need to decompose the code into two subspaces $\mathrm{rowsp}(\mathbf{B})\oplus \mathrm{rowsp}(\mathbf{L})\subseteq \mathrm{ker}\mathbf{A}$. The three matrices must satisfy the compatibility condition $\mathbf{A}\mathbf{B}^\mathrm{T} = \mathbf{A}\mathbf{L}^\mathrm{T} = 0$. 

%Given a stabiliser circuit, we can map it to a check matrix $\mathbf{A}$. With addition information regarding the protocol of error correction and corresponding logical operations, we can work out matrices $\mathbf{B}$ and $\mathbf{L}$ accordingly. 

For each code and a decomposition of the code, i.e. a 3-tuple $(\mathbf{A},\mathbf{B},\mathbf{L})$, we define the circuit code distance $d(\mathbf{A},\mathbf{B},\mathbf{L})$ in Sec. \ref{sec:distance}. The distance characterises the fault tolerance of the corresponding circuit: if $\mathbf{A}$ is mapped from a circuit, the distance quantifies the minimum number of bit-flip and phase-flip errors responsible for a logical error in the circuit. Because the distance is defined on the three matrices, it can be applied to any three matrices that meet the compatibility condition even if the matrix $\mathbf{A}$ is not mapped from a circuit. 

We can use the LDPC representation in two ways. First, the LDPC representation can be used to analyse a stabiliser circuit. From the matrix $\mathbf{A}$, we can work out all the preserved quantities, i.e. codewords. These codewords determine what correlations can be used for the two fundamental tasks: error correction and logical operation. Once codewords are classified according to the tasks, we can quantify the fault tolerance through $d(\mathbf{A},\mathbf{B},\mathbf{L})$ and verify the logical operation. Second, the LDPC representation can be used to develop circuits from codes. Suppose that we want a stabiliser circuit for a certain logical operation. We can construct the circuit by finding three matrices $\mathbf{A}$, $\mathbf{B}$ and $\mathbf{L}$, which need to satisfy the following requirements: i) They must meet the compatibility condition; ii) The circuit code distance $d(\mathbf{A},\mathbf{B},\mathbf{L})$ must be sufficiently large; and iii) $\mathbf{L}$ must meet the desired logical operation. The matrix $\mathbf{A}$ also needs to possess the bit-check symmetry. There is no other requirement, i.e. the matrix $\mathbf{A}$ may not be any check matrix that is mapped from a circuit. If we can find such three matrices, then we can map $\mathbf{A}$ to a stabiliser circuit. In the circuit, we can detect errors according to $\mathbf{B}$, and it realises the logical operation depicted by $\mathbf{L}$. In this way, we have a fault-tolerant circuit for accomplishing the desired logical operation. 

In addition to constructing circuits from codes, a practical approach to finding new circuits is through analysing existing circuits. The LDPC representation can reveal certain structures responsible for fault tolerance and logical correlations. Then, we can develop new circuits with the insights from the structures. We give such an example in this work. We propose a protocol of universal fault-tolerant quantum computing on hypergraph product codes, which is inspired by the structure in the LDPC representation for circuits composed of transversal gates on CSS codes. 

\subsection{Maps from codes to circuits}

\begin{figure}[tbp]
\centering
\includegraphics[width=\linewidth]{\figpath/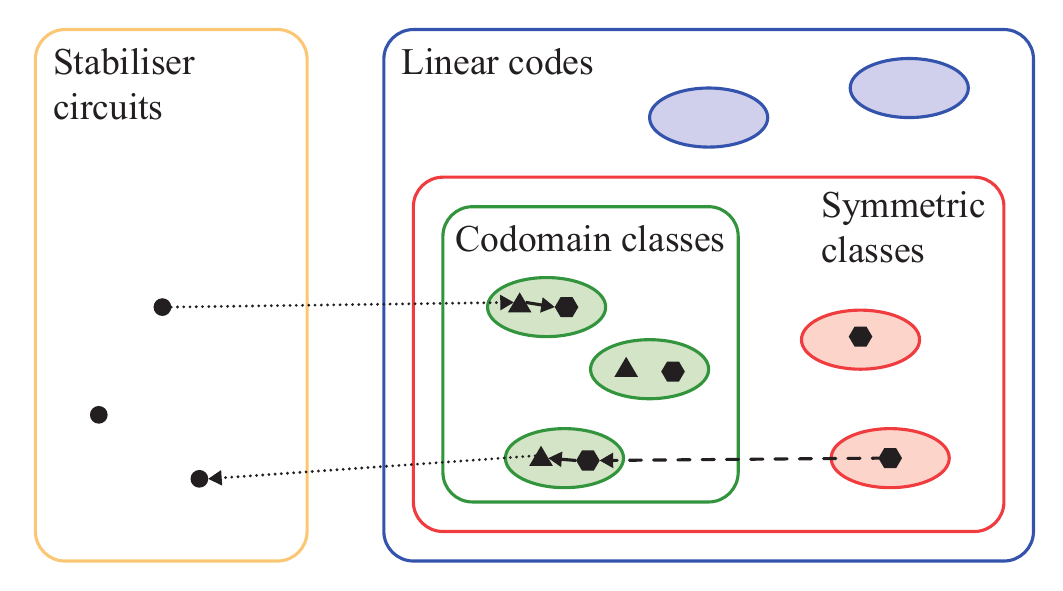}
\caption{
Relations between stabiliser circuits and linear codes. Black circles, triangles and hexagons denote stabiliser circuits, codes mapped from circuits (plain Tanner graphs) and codes with the bit-check symmetry (symmetric Tanner graphs), respectively. Blue, red and green circles are equivalence classes established by the bit splitting, i.e. two codes in the same class can be transformed into each other via the bit splitting and its inverse operation (denoted by solid arrows). In the set of symmetric classes, each class contains at least one code with the symmetry. In the set of codomain classes, each class contains at least one plain Tanner graph. We can map a stabiliser circuit to a plain Tanner graph, which can be further mapped to a symmetric Tanner graph via the bit splitting. We can also map a general code with the symmetry to a stabiliser circuit in three steps: i) transform the code to the symmetric Tanner graph in a codomain class through the symmetric splitting (denoted by the dashed arrow), ii) transform the symmetric Tanner graph to the plain Tanner graph through the inverse bit splitting, and iii) map the plain Tanner graph back to a circuit. 
}
\label{fig:relation}
\end{figure}

Before mapping linear codes to stabiliser circuits, we need to ask what codes can be mapped to circuits. The answer is related to the bit-check symmetry, which is introduced in Sec. \ref{sec:symmetry}. For a linear code with this symmetry, its check matrix is in the block matrix form 
\begin{eqnarray}
\mathbf{A} = \left(\begin{matrix}
\left.\begin{matrix}
\openone \\
0
\end{matrix}~~~\right| & \mathbf{A}_S
\end{matrix}\right)
\end{eqnarray}
up to permutations on the rows and columns; see Definition \ref{def:symmetry}. Here, $\mathbf{A}_S = \mathbf{A}_S^\mathrm{T}$ is a symmetric matrix. For all linear codes with this symmetry, we can map them to stabiliser circuits. If the code has a sufficiently large circuit code distance for a certain decomposition of the code (denoted by matrices $\mathbf{B}$ and $\mathbf{L}$), the stabiliser circuit is fault-tolerant. 

Relations between stabiliser circuits and linear codes is illustrated in Fig. \ref{fig:relation}. Linear codes are related by two operations on Tanner graphs, the bit splitting and symmetric splitting; see Definition \ref{def:bit_splitting} and Sec. \ref{sec:splitting}. An essential observation is that the bit splitting induces an equivalence relation between linear codes, i.e. in each equivalence class, linear codes can be transformed into each other through the bit splitting and its inverse operation. When we map a stabiliser circuit to a linear code, the code may not possess the bit-check symmetry; however, it is always in an equivalence class containing at least one code possessing the symmetry. 

We call the code mapped from the circuit the plain Tanner graph of the circuit, and we call the symmetric code in the same equivalence class the symmetric Tanner graph of the circuit. We call equivalence classes that contain at least one plain Tanner graph codomain classes. Then, the map from circuits to plain Tanner graphs induces a map from circuits to codomain classes. 

Given a linear code possessing the bit-check symmetry (a general symmetric Tanner graph), we can map it to a stabiliser circuit with the following procedure. First, we transform the code into the symmetric Tanner graph in a codomain class, which is achieved by the symmetric splitting. Then, we map the codomain class (i.e. the plain Tanner graph in it) back to a stabiliser circuit. The fault tolerance of the circuit is guaranteed by the properties of bit splitting and symmetric splitting. 

The property of the bit splitting is summarised in Lemma \ref{lem:bit_splitting}. Let $\mathbf{A}$ be the check matrix of the plain Tanner graph, and let $\mathbf{A}'$ be the check matrix of the symmetric Tanner graph in the same equivalence class. Then, we can transform $\mathbf{A}$ to $\mathbf{A}'$ through the bit splitting. This transformation on Tanner graphs induces a linear bijection $\phi:\mathrm{ker}\mathbf{A}\rightarrow\mathrm{ker}\mathbf{A}'$ between two codes. Therefore, the bit splitting induces a transform on 3-tuples of compatible matrices: Let $\mathbf{B}$ and $\mathbf{L}$ be matrices compatible with $\mathbf{A}$, then $(\mathbf{A},\mathbf{B},\mathbf{L})$ is mapped to $(\mathbf{A}',\mathbf{B}',\mathbf{L}')$ by the bit splitting, where $\mathbf{B}' = \phi(\mathbf{B})$ and $\mathbf{L}' = \phi(\mathbf{L})$ are compatible with $\mathbf{A}'$. The key property of the bit splitting is that the circuit code distance is preserved, i.e. $d(\mathbf{A},\mathbf{B},\mathbf{L}) = d(\mathbf{A}',\mathbf{B}',\mathbf{L}')$. 

The property of the symmetric splitting is summarised in Theorem \ref{the:symmetric_splitting}. It is similar to the bit splitting. Let $\mathbf{A}'$ and $\mathbf{A}''$ be check matrices of two symmetric Tanner graphs. Suppose $\mathbf{A}'$ is generated by applying the symmetric splitting on $\mathbf{A}''$. The symmetric splitting also induces a linear bijection $\psi:\mathrm{ker}\mathbf{A}''\rightarrow\mathrm{ker}\mathbf{A}'$ between two codes and a transform on 3-tuples of compatible matrices. Let $\mathbf{B}''$ and $\mathbf{L}''$ be matrices compatible with $\mathbf{A}''$, then $(\mathbf{A}'',\mathbf{B}'',\mathbf{L}'')$ is mapped to $(\mathbf{A}',\mathbf{B}',\mathbf{L}')$ by the symmetric splitting, where $\mathbf{B}' = \psi(\mathbf{B}'')$ and $\mathbf{L}' = \psi(\mathbf{L}'')$ are compatible with $\mathbf{A}'$. The only difference is the impact on the circuit code distance. In the symmetric splitting, the distance has the lower bound  $d(\mathbf{A}',\mathbf{B}',\mathbf{L}') \geq d(\mathbf{A}'',\mathbf{B}'',\mathbf{L}'')/\lfloor g_{max}/2\rfloor$, where $g_{max}$ is the maximum vertex degree in the Tanner graph of $\mathbf{A}''$. When $\mathbf{A}''$ is an LDPC code, $g_{max}$ is small. 

We generate fault-tolerant stabiliser circuits from linear codes in the following way. First, we find a code possessing the bit-check symmetry, which is denoted by $\mathbf{A}''$. We also need a proper decomposition of the code, which is denoted by $\mathbf{B}''$ and $\mathbf{L}''$. The circuit code distance $d(\mathbf{A}'',\mathbf{B}'',\mathbf{L}'')$ must be sufficiently large compared to $\lfloor g_{max}/2\rfloor$, i.e. we always prefer an LDPC code. Additionally, if we have a specific logical operation to implement, we can use the matrix $\mathbf{L}''$ to verify that the generated quantum circuit can achieve the desired logical operation, see Sec. \ref{sec:verification}. Second, we find a symmetric splitting that can transform $\mathbf{A}''$ to the symmetric Tanner graph in a codomain class, which is denoted by $\mathbf{A}'$. The corresponding map on codes is $\psi$. Thirdly, we find the bit splitting relating $\mathbf{A}'$ with the plain Tanner graph $\mathbf{A}$ in the codomain class. The corresponding map on codes is $\phi$. Then, we can map the plain Tanner graph $\mathbf{A}$ back to a stabiliser circuit. We can detect errors in the circuit according to $\mathbf{B} = \phi^{-1}\circ\psi(\mathbf{B}'')$, and the circuit realises the logical operation depicted by $\mathbf{L} = \phi^{-1}\circ\psi(\mathbf{L}'')$. The circuit code distance is guaranteed and has the lower bound $d(\mathbf{A},\mathbf{B},\mathbf{L}) \geq d(\mathbf{A}'',\mathbf{B}'',\mathbf{L}'')/\lfloor g_{max}/2\rfloor$. Therefore, the circuit can tolerate at least $d(\mathbf{A}'',\mathbf{B}'',\mathbf{L}'')/\lfloor g_{max}/2\rfloor$ errors. 

The map from linear codes to stabiliser circuits is not unique. With different strategies in the symmetric splitting, we may generate different stabiliser circuits. All the circuits are fault-tolerant as long as $d(\mathbf{A}'',\mathbf{B}'',\mathbf{L}'')/\lfloor g_{max}/2\rfloor$ is sufficiently large. In Sec. \ref{sec:construction}, we present a strategy of the symmetric splitting to demonstrate the existence of the map, and we also show how to choose $(\mathbf{A}'',\mathbf{B}'',\mathbf{L}'')$ to realise certain operations on logical qubits. 

\subsection{Examples of the application}

To demonstrate the application of the LDPC representation in analysing and designing fault-tolerant circuits, we present two examples. The first example applies the LDPC representation to quantum circuits composed of transversal gates on CSS codes. We will provide a concise general matrix representation for transversal circuits. The second example involves using the LDPC representation to design a fault-tolerant quantum circuit operating on hypergraph product codes. The design of this fault-tolerant quantum circuit is inspired by insights gained from the LDPC representation of transversal circuits. Based on this fault-tolerant circuit, we have developed a scheme for universal fault-tolerant quantum computing on hypergraph product codes. Compared to the scheme presented in a recent work \cite{Xu2024}, our scheme is resource-efficient. 

In the first example, we present the general form of the three matrices $(\mathbf{A},\mathbf{B},\mathbf{L})$ in the LDPC representation of transversal circuits. This general form is applicable to any CSS code and any circuit on logical qubits (logical circuit). Based on these three matrices, we analyse the fault tolerance of the circuit. 

One characteristic of the LDPC representation for transversal circuits is the appearance of a factor $\mathbf{a}$ in matrices $\mathbf{A}$ and $\mathbf{B}$. This factor is the check matrix representing the logical circuit. When we apply the general form of the LDPC representation for transversal circuits to hypergraph product codes, we find that matrices $\mathbf{A}$ and $\mathbf{B}$ are similar to the two check matrices of the $X$ and $Z$ operators in three-dimensional homological product codes \cite{Bravyi2013}. In this three-dimensional product, two dimensions correspond to the check matrices of the two linear codes in the hypergraph product code (spatial dimensions), and the other dimension corresponds to the logical-circuit check matrix $\mathbf{a}$ (time dimension). 

In the second example, we construct a new fault-tolerant circuit by moving the logical-circuit check matrix $\mathbf{a}$ from the time dimension to a spatial dimension. This fault-tolerant circuit can transfer logical qubits between two hypergraph product codes while simultaneously implementing a set of logical operations. Furthermore, we use the LDPC representation to analyse this circuit, proving its fault tolerance and verifying the logical operations. 

Since the construction and analysis of the fault-tolerant circuit use abstract matrix representations without involving specific physical or logical circuits, the results are general. Consequently, we can develop a universal fault-tolerant quantum computing scheme based on these results: 
\begin{itemize}
\item[1.] The logical qubits of the hypergraph product code form a two-dimensional array. We can simultaneously extract several rows or columns from this array and perform operations on them individually. In this way, we can realise all logical operations in which $X$ and $Z$ operators are decoupled, including initialisation, measurement, controlled-NOT gates, swap gates, etc.; 
\item[2.] We can extract a square sub-block from the array and apply the transversal Hadamard gate, thereby realising logical Hadamard gates; 
\item[3.] We can transfer states between different hypergraph product codes. Therefore, we can encode magic states on small surface codes and then transfer them to the hypergraph product code with a large code distance, thereby achieving magic state preparation. 
\end{itemize}
Through these operations, the universal quantum computing can be realised \cite{Fowler2012}. 

In the scheme proposed in Ref. \cite{Xu2024}, the basic operation involve manipulating a logical qubit using a full-size surface code. In comparison, our scheme is more resource-efficient: since it allows simultaneous operations on multiple rows, columns, or sub-blocks, our scheme is more time-efficient; and because operations are performed on rows, columns, or sub-blocks, it does not require many full-size surface codes, making it more qubit-efficient. We want to remark that, in addition to resource efficiency, the choice of a fault-tolerant quantum computing scheme also depends on factors such as logical error rates and compatibility with physical systems, which are not considered in this work. Nonetheless, the advantages in specific metrics are sufficient to demonstrate that the LDPC representation is indeed a promising and powerful tool for developing fault-tolerant circuits. 

\subsection{Comparison to spacetime codes}

An approach to representing stabiliser circuits using codes is mapping the circuits into subsystem codes, known as spacetime codes \cite{Bacon2017,Gottesman2022,Delfosse2023,Fu2024}. Through this mapping, fault-tolerant circuits can be converted into subsystem codes with near-optimal parameters. Additionally, spacetime codes serve as a tool for representing and analysing fault-tolerant circuits. 

In the framework of the LDPC representation, we use classical codes to represent stabiliser circuits. These classical codes include plain Tanner graphs and all linear codes with bit-check symmetry. While the LDPC representation is introduced by mapping stabiliser circuits into plain Tanner graphs, one of the main results in this work is establishing a general relation between stabiliser circuits and all linear codes with bit-check symmetry. Therefore, we can represent stabiliser circuits not only with plain Tanner graphs but also with linear codes with bit-check symmetry, making the LDPC representation a flexible and general method rather than just plain Tanner graphs. 

Due to the flexibility and generality of the LDPC representation, there are significant differences between it and spacetime codes. First, there is a similarity between plain Tanner graphs and spacetime codes. In the process of constructing codes from stabiliser circuits, primitive Tanner graphs introduced in Sec. \ref{sec:primitive} and gauge operators play similar roles, describing the same transformations on Pauli operators. However, general linear codes with bit-check symmetry do not have this similarity with spacetime codes. As an example, Fig. \ref{fig:ZZ}(b) shows a plain Tanner graph of a stabiliser circuit. In this graph, each point in spacetime corresponds to two bit vertices, representing the $X$ and $Z$ operators of the corresponding qubit, respectively. However, through bit splitting and symmetric splitting, we can simplify the plain Tanner graph into a symmetric Tanner graph, shown in Fig. \ref{fig:ZZsimpllification}. In this simplified graph, the numbers of $X$ and $Z$ bit vertices are different, therefore, the graph does not correspond to any subsystem code. 

Generally speaking, it is during the process of converting a linear code with bit-check symmetry into a specific quantum circuit that we assign concrete physical meaning to each bit vertex, i.e. we assign a Pauli operator in spacetime to that vertex. By selecting different paths in the process and assigning different Pauli operators to the bit vertices, a linear code can be converted into different quantum circuits. Therefore, there is no direct correspondence between these linear codes and spacetime codes. 

Compared to spacetime codes, the LDPC representation uses two specific techniques: check matrices and Tanner graphs. The use of these two techniques has led to a series of results, including: i) discovering the relation between stabiliser circuits and bit-check symmetry using the Tanner graph representation; ii) introducing two graph operations, bit splitting and symmetric splitting, allowing us to transform a general linear code with bit-check symmetry into a stabiliser circuit while ensuring the fault tolerance of the circuit; iii) discovering a product structure in transversal circuits using the check matrix representation; and iv) designing fault-tolerant quantum circuits on hypergraph product codes, and verifying the fault tolerance and logical operations using the check matrix representation. Therefore, using the techniques of check matrices and Tanner graphs makes the LDPC representation a promising and powerful tool. 
}

\section{Low-density parity-check representation of stabiliser circuits}
\label{sec:representation}

This section presents an approach for mapping stabiliser circuits to classical LDPC codes. Firstly, we introduce some notations and expound upon the basic idea of the LDPC representation. Subsequently, we provide a set of check matrices (Tanner graphs) representing primitive operations that yield general stabiliser circuits. Finally, we outline the procedure for constructing LDPC codes representing general stabiliser circuits. 

A circuit is said to be a stabiliser circuit if it consists of operations generated by the following operations \cite{Gottesman1998}: 
\begin{itemize}
\item[$\bullet$] Initialisation of a qubit in the state $\ket{0}$; 
\item[$\bullet$] Clifford gates; 
\item[$\bullet$] Measurement in the $Z$ basis. 
\end{itemize}
A projective measurement is equivalent to a measurement followed by an initialisation and an $X$ Pauli gate depending on the measurement outcome. Without loss of generality, we assume that following a measurement, the subsequent operation on the same qubit is always an initialisation. Circuits with this feature can effectively minimise errors stemming from measurements, making them prevalent in quantum error correction. 

Clifford gates transform Pauli operators into Pauli operators through conjugation, up to a phase factor. For $n$ qubits, the set of Pauli operators is $P_n = \{I,X,Y,Z\}^{\otimes n}$, where $I$, $X$, $Y$ and $Z$ are single-qubit Pauli operators. When a Pauli operator acts non-trivially on only one qubit, it is denoted by $\alpha_q = I^{\otimes(q-1)}\otimes \alpha\otimes I^{\otimes(n-q)}$, where $\alpha = X,Y,Z$. The Pauli group of $n$ qubits is $G_n = P_n\times\{\pm 1,\pm i\}$. We can define a map $\eta : G_n \rightarrow \{\pm 1,\pm i\}$ to denote the phase factor of Pauli group elements: When $g = \pm S,\pm iS$ and $S\in P_n$, the function takes values $\eta(g) = \pm 1,\pm i$, respectively. The Clifford group is the normaliser of the Pauli group. 

We can represent Pauli operators using binary vectors \cite{Gottesman1997,Calderbank1998}. Let's define a map $\sigma(\bullet,\bullet) : \mathbb{F}_2^n\times\mathbb{F}_2^n \rightarrow P_n$. For two vectors $\mathbf{x}=(x_1,x_2,\ldots,x_n)\in \mathbb{F}_2^n$ and $\mathbf{z}=(z_1,z_2,\ldots,z_n)\in \mathbb{F}_2^n$, the map reads 
\begin{eqnarray}
\sigma(\mathbf{x},\mathbf{z}) = i^{\vert \mathbf{x}\odot\mathbf{z}\vert}X^{x_1}Z^{z_1}\otimes X^{x_2}Z^{z_2}\otimes \cdots\otimes X^{x_n}Z^{z_n}.
\end{eqnarray}
Here, $\mathbf{x}\odot\mathbf{z} = (x_1z_1,x_2z_2,\ldots,x_nz_n)$, and $\vert\bullet\vert$ is the Hamming weight. The map is bijective, and $\sigma^{-1}$ denotes the inverse map. 

In the binary-vector representation of Pauli operators, conjugations by Clifford gates become binary linear maps. Let $[U]\bullet = U\bullet U^\dag$ denote the conjugation by the Clifford operator $U$. For all $(\mathbf{x}^{(0)},\mathbf{z}^{(0)})\in \mathbb{F}_2^n\times\mathbb{F}_2^n$, the superoperator $[U]$ transforms $\sigma(\mathbf{x}^{(0)},\mathbf{z}^{(0)})$ into an element in the Pauli group. In other words, there exists $(\mathbf{x}^{(1)},\mathbf{z}^{(1)})\in \mathbb{F}_2^n\times\mathbb{F}_2^n$ such that 
\begin{eqnarray}
[U]\sigma(\mathbf{x}^{(0)},\mathbf{z}^{(0)}) = \eta\left([U]\sigma(\mathbf{x}^{(0)},\mathbf{z}^{(0)})\right)\sigma(\mathbf{x}^{(1)},\mathbf{z}^{(1)}).
\label{eq:Usigma}
\end{eqnarray}
Notice that $\eta\left([U]\sigma(\mathbf{x}^{(0)},\mathbf{z}^{(0)})\right) = \pm 1$ is always a sign factor. According to Eq. (\ref{eq:Usigma}), we can define a map $M_U : \mathbb{F}_2^n\times\mathbb{F}_2^n \rightarrow \mathbb{F}_2^n\times\mathbb{F}_2^n$ to represent $[U]$, which reads
\begin{eqnarray}
M_U(\mathbf{x},\mathbf{z}) = \sigma^{-1}\left(\eta\left([U]\sigma(\mathbf{x},\mathbf{z})\right)[U]\sigma(\mathbf{x},\mathbf{z})\right).
\end{eqnarray}
This map is consistent with Eq.~(\ref{eq:Usigma}) in the sense $M_U(\mathbf{x}^{(0)},\mathbf{z}^{(0)}) = (\mathbf{x}^{(1)},\mathbf{z}^{(1)})$. We can find that $M_U$ is linear and bijective. 

Furthermore, we can represent Clifford gates using check matrices. Let $\mathbf{M}_U$ be the matrix of $M_U$, and $\mathbf{b}_t = (\mathbf{x}^{(t)},\mathbf{z}^{(t)})$. Then, linear maps are in the form $\mathbf{M}_U\mathbf{b}_0^\mathrm{T} = \mathbf{b}_1^\mathrm{T}$. They can be rewritten as equations 
\begin{eqnarray}
\mathbf{A}_U\left(\begin{matrix}
\mathbf{b}_0 & \mathbf{b}_1
\end{matrix}\right)^\mathrm{T} = 0.
\label{eq:Ab}
\end{eqnarray}
Here, 
\begin{eqnarray}
\mathbf{A}_U = \left(\begin{matrix}
\mathbf{M}_U & \openone
\end{matrix}\right)
\end{eqnarray}
is the check matrix representing the Clifford gate $U$. We can also represent the initialisation and measurement with check matrices, which will be given latter. 

\begin{proposition}
Let $\mathbf{A}_U$ be the check matrix of the Clifford gate $U$. Codewords of $\mathbf{A}_U$ describe the transformation of Pauli operators under $U$, i.e. for a vector $(\mathbf{b}_0,\mathbf{b}_1) = (\mathbf{x}^{(0)},\mathbf{z}^{(0)},\mathbf{x}^{(1)},\mathbf{z}^{(1)})$, Eq. (\ref{eq:Ab}) holds if and only if Eq. (\ref{eq:Usigma}) holds. 
\label{prop:gate}
\end{proposition}

\subsection{Check matrices of primitive operations}
\label{sec:primitive}

The controlled-NOT gate $\Lambda$, Hadamard gate $H$ and $\pi/4$ phase gate $S$ can generate all Clifford gates. Their linear maps are 
\begin{eqnarray}
\mathbf{M}_\Lambda = \left(\begin{matrix}
1 & 0 & 0 & 0 \\
1 & 1 & 0 & 0 \\
0 & 0 & 1 & 1 \\
0 & 0 & 0 & 1
\end{matrix}\right),
\end{eqnarray}
\begin{eqnarray}
\mathbf{M}_H = \left(\begin{matrix}
0 & 1 \\
1 & 0
\end{matrix}\right)
\end{eqnarray}
and 
\begin{eqnarray}
\mathbf{M}_S = \left(\begin{matrix}
1 & 0 \\
1 & 1
\end{matrix}\right),
\end{eqnarray}
respectively. For the completeness, the linear map of the single-qubit identity gate $I$ is 
\begin{eqnarray}
\mathbf{M}_I = \left(\begin{matrix}
1 & 0 \\
0 & 1
\end{matrix}\right).
\end{eqnarray}

\begin{figure}[tbp]
\centering
\includegraphics[width=\linewidth]{\figpath/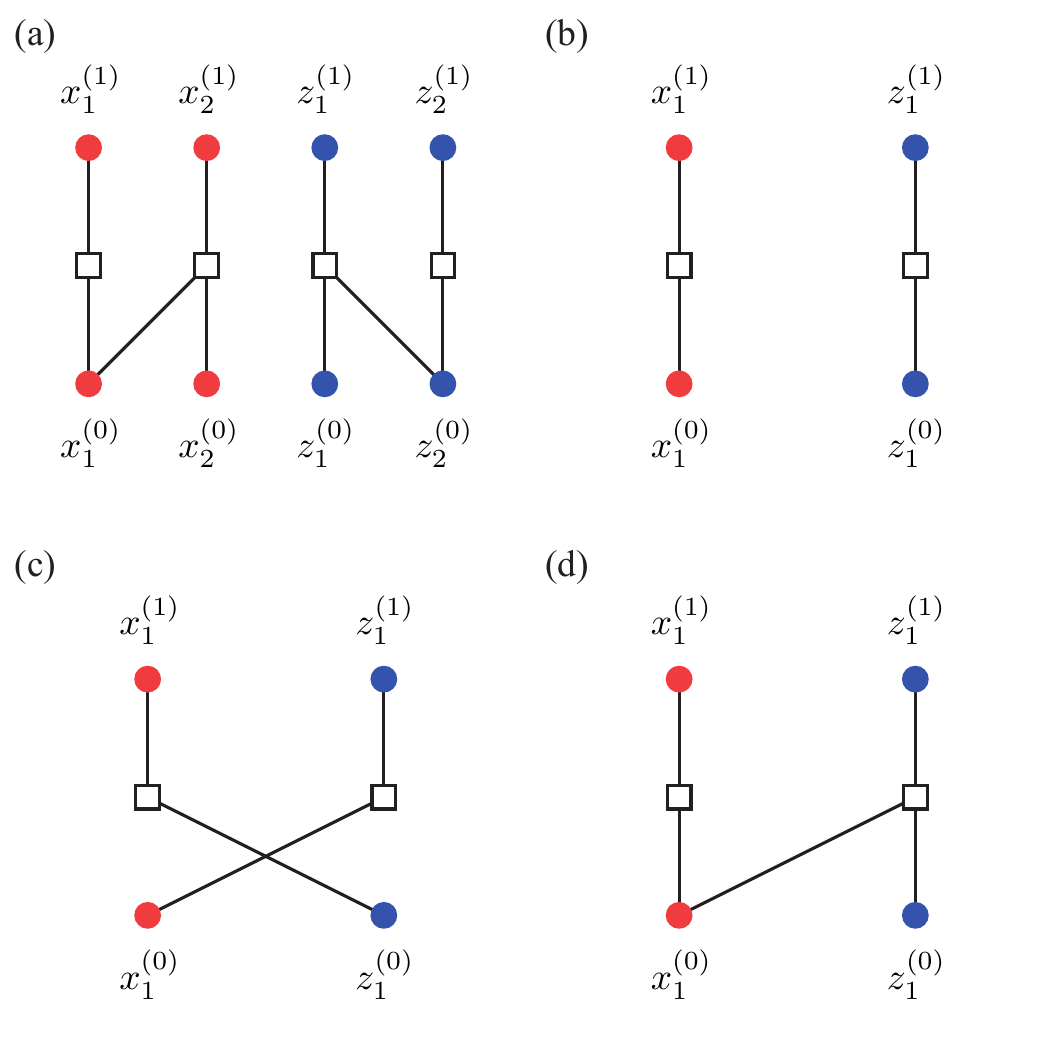}
\caption{
Tanner graphs of the (a) controlled-NOT gate, (b) identity gate, (c) Hadamard gate and (d) phase gate. Red and blue circles are bits $x_q^{(t)}$ and $z_q^{(t)}$, respectively. Squares are checks. The controlled-NOT gate is applied on qubit-1 and qubit-2, in which qubit-1 and qubit-2 are the control and target qubits, respectively. Single-qubit gates are applied on qubit-1. 
}
\label{fig:gates}
\end{figure}

A way of illustrating check matrices is using Tanner graphs \cite{Tanner1981}. A Tanner graph is a bipartite graph with a set of bit vertices and a set of check vertices. Each bit (check) corresponds to a column (row) of the check matrix. There is an edge incident on a bit and a check if and only if the corresponding matrix entry takes one. Using Tanner graphs, we can illustrate check matrices $\mathbf{A}_\Lambda$, $\mathbf{A}_I$, $\mathbf{A}_H$ and $\mathbf{A}_S$ of the four gates as shown in Fig. \ref{fig:gates}. 

\begin{figure}[tbp]
\centering
\includegraphics[width=\linewidth]{\figpath/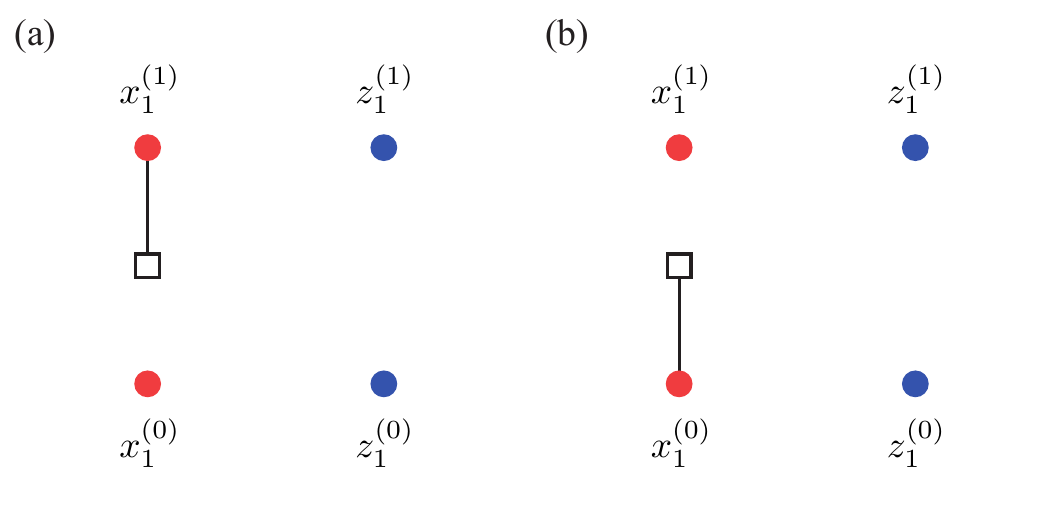}
\caption{
Tanner graphs of the (a) initialisation and (b) measurement on qubit-1. 
}
\label{fig:IandM}
\end{figure}

We can also represent the initialisation and measurement with check matrices. Let’s consider the initialisation first. In the initialisation, the input and output states of the qubit are independent. Accordingly, input and output bits are decoupled. Suppose qubit-1 is initialised in the $Z$ basis. After the initialisation, if we measure any Pauli operator involving $X_1$ as a factor, the measurement outcome is completely random. Therefore, $X_1$ is useless in quantum error correction, and only Pauli operators without $X_1$ are interesting. Accordingly, the output bit $x_1^{(1)}$ always takes the value of $0$. Due to the above reasons, the check matrix of the initialisation reads 
\begin{eqnarray}
\mathbf{A}_{ini} = \left(\begin{matrix}
0 & 0 & 1 & 0
\end{matrix}\right).
\end{eqnarray}
Similarly, the check matrix of the measurement reads 
\begin{eqnarray}
\mathbf{A}_{mea} = \left(\begin{matrix}
1 & 0 & 0 & 0
\end{matrix}\right).
\end{eqnarray}
See Fig. \ref{fig:IandM} for their Tanner graphs. By adopting check matrices for initialisation and measurement in this manner, we will find that LDPC codes representing stabiliser circuits possess intuitive and definite physical meanings, which will be discussed in Sec. \ref{sec:codewords}. 

\RED{Using a similar method, we can construct Tanner graphs to represent single-qubit measurements of $X$ and $Y$ operators, as well as multi-qubit measurements; see Appendix \ref{app:measurements}. Consequently, we can represent quantum circuits composed of a sequence of non-commuting measurements. These types of quantum circuits are used for subsystem and Floquet codes \cite{Poulin2005,Hastings2021}. }

\subsection{Check matrices of stabiliser circuits}
\label{sec:CircuitGraph}

We can construct the check matrix (draw the Tanner graph) of a stabiliser circuit in the following way. See Fig. \ref{fig:ZZ} for an example. 
\begin{itemize}
\item[1.] Consider a circuit with $n$ qubits and a depth of $T$ (the circuit consists of $T$ layers of parallel operations). Draw $2n(T+1)$ bit vertices. 
\item[] The set of bits $V_{B,all} = V_{B,0}\cup V_{B,1}\cup\cdots\cup V_{B,T}$ is the union of $T+1$ subsets, and each subset $V_{B,t} = \{\hat{x}_q^{(t)},\hat{z}_q^{(t)}\vert q = 1,2,\ldots,n\}$ has $2n$ bits. Here, we use notations $\hat{x}_q^{(t)}$ and $\hat{z}_q^{(t)}$ with the hat to denote bit vertices for clarity, and notations $x_q^{(t)},z_q^{(t)}\in \mathbb{F}_2$ without the hat are values of the bits. Bits $V_{B,t-1}$ and $V_{B,t}$ are the input and output of the layer-$t$ operations, respectively. Bits $V_{B,0}$ and $V_{B,T}$ are the input and output of the circuit, respectively. 
\item[2.] Add check vertices and edges for each primitive operation, according to Figs. \ref{fig:gates} and \ref{fig:IandM}. 
\item[3.] Remove isolated bits. 
\end{itemize}
From now on, we use $\mathbf{A}$ to denote the check matrix of a stabiliser circuit, and we use $V_B\subseteq V_{B,all}$ ($V_C$) to denote the bit (check) vertex set of the final Tanner graph. 

\begin{figure}[tbp]
\centering
\includegraphics[width=\linewidth]{\figpath/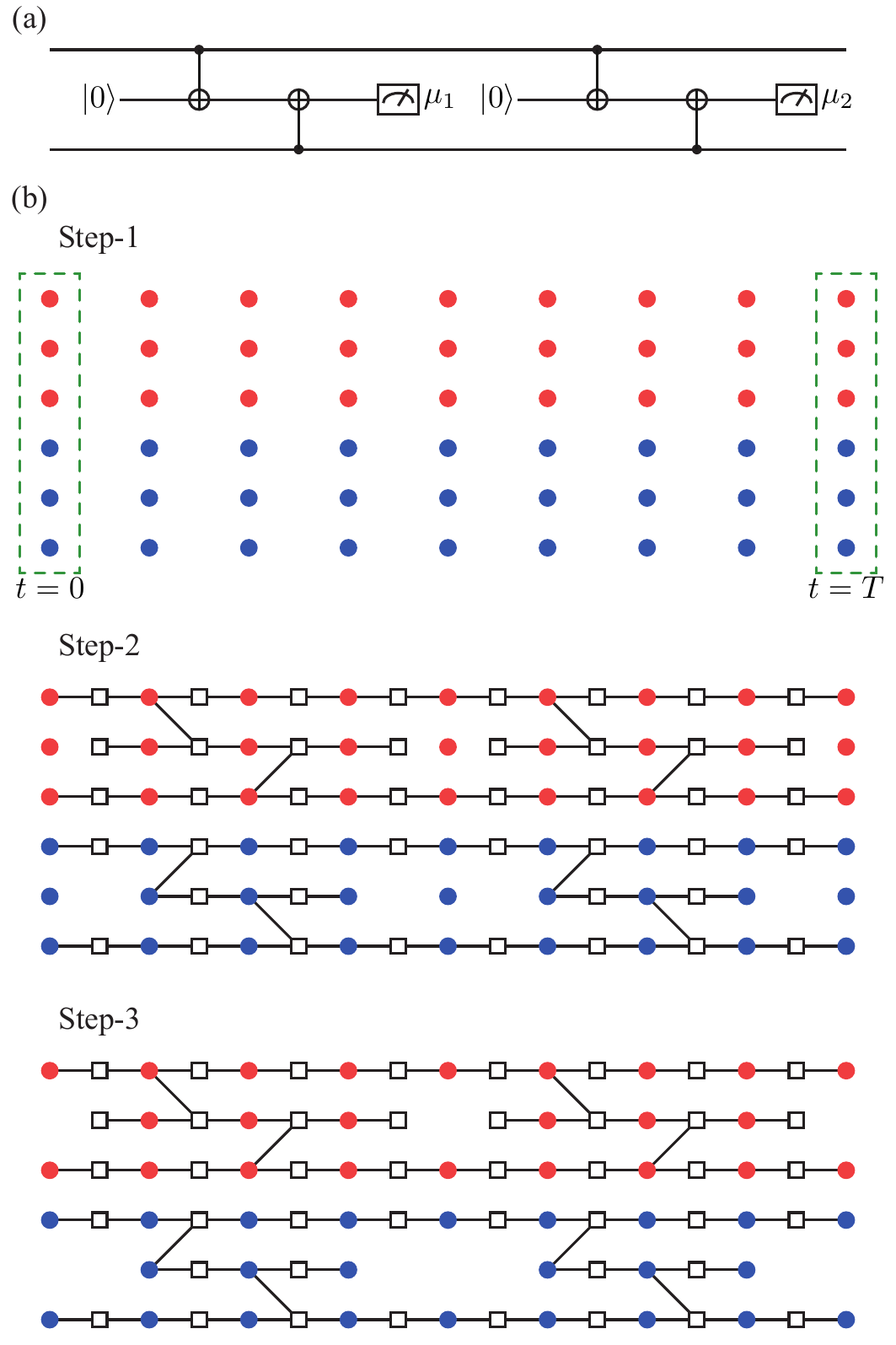}
\caption{
(a) The circuit of the repeated $Z_1Z_2$ measurement. Measurement outcomes are $\mu_1,\mu_2 = \pm 1$. This circuit is the parity-check circuit of the [[2,1,1]] code. In the [[2,1,1]] code, the stabiliser generating set is $\{Z_1Z_2\}$, and logical Pauli operators are $X_1X_2$ and $Z_1$. In the circuit, the measurement is repeated two times. 
(b) Drawing the Tanner graph. The circuit has $n = 3$ qubits and $T = 8$ layers. In step-1, an $2n\times(T+1)$ array of bits is created. In step-2, checks and edges are added. In step-3, isolated bits are removed. 
}
\label{fig:ZZ}
\end{figure}

We call the Tanner graph constructed according to the above procedure the plain Tanner graph of a stabiliser circuit. If the stabiliser circuit is generated by the controlled-NOT gate and single-qubit operations, the Tanner graph has a maximum vertex degree of three. Therefore, the corresponding linear code is an LDPC code. 

When stabiliser circuits only differ in Pauli gates, they are represented by the same LDPC code. For example, the check matrix of a single-qubit Pauli gate is the same as the identity gate. The reason is that Pauli gates only change the sign of Pauli operators. In quantum error correction, we usually neglect such a difference between two stabiliser circuits. 

\RED{{\bf Application to general circuits.} Besides stabiliser circuits, the LDPC representation can also be used to study non-stabiliser circuits. The key point here is that the LDPC representation describes transformations on quantum states, and the input quantum state can be any state. Therefore, we can generate a stabiliser sub-circuit by removing non-Clifford gates from a general quantum circuit (see Fig. \ref{fig:non-Clifford}) and then apply the LDPC representation. Usually, this sub-circuit determines the capability of error correction. Therefore, the LDPC representation can be applied to non-stabiliser circuits to study fault tolerance. 

\begin{figure}[tbp]
\centering
\includegraphics[width=\linewidth]{\figpath/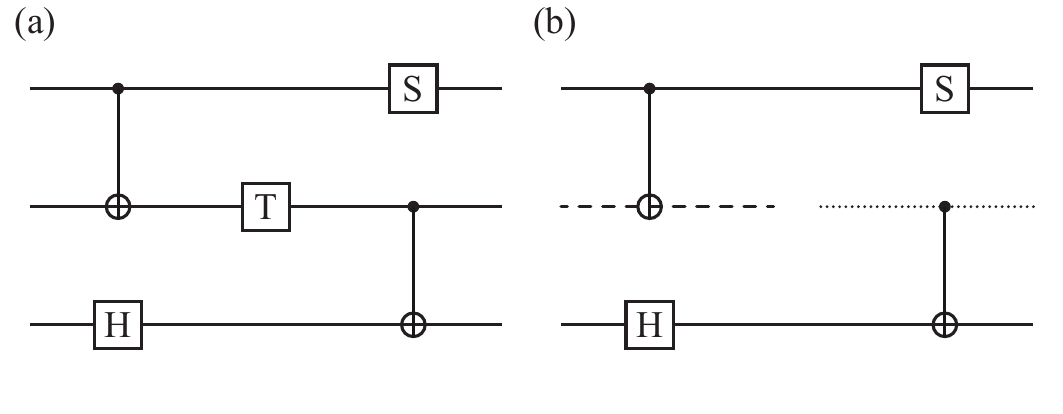}
\caption{
(a) A non-stabiliser circuit. (b) A stabiliser sub-circuit generated by deleting the non-Clifford gate $T$. The dotted and dashed lines represent two different qubits. 
}
\label{fig:non-Clifford}
\end{figure}

In fault-tolerant quantum computing, a practical approach to achieve universal quantum gates is through the preparation and distillation of magic states. The above method can be used to study the preparation of magic states. One way to prepare the magic state is to initialise a physical qubit in the magic state (a non-stabiliser operation) and then encode the state of this physical qubit into a logical qubit using a stabiliser circuit (see Ref. \cite{Li2015} for example). In this process, the encoding fidelity is mainly determined by the stabiliser circuit. The LDPC representation can be used to study this encoding process. Similarly, magic state distillation is also realised through stabiliser circuits \cite{Fowler2012}. Therefore, the LDPC representation can be broadly applied in universal fault-tolerant quantum computing. }

\section{Codewords of stabiliser circuits}
\label{sec:codewords}

We already have a representation of stabiliser circuits in the form of LDPC codes. In this section, we discuss the physical meaning of the representation. We classify codewords of a stabiliser circuit into several categories according to their physical meaning. The categories include checker, detector, emitter and propagator. They correspond to the parity check on measurement outcomes, measurement, eigenstate preparation and transformation on Pauli operators, respectively. A subset of propagators called genuine propagators is essential for logical quantum gates because they represent coherent correlations. 

For a stabiliser circuit, codewords describe the correlations established by the circuit. For a circuit with only Clifford gates, codewords describe the transformation on Pauli operators (Proposition \ref{prop:gate}). For a general stabiliser circuit, we can express the correlations with a set of equations termed codeword equations. 

Before giving the equations, let's consider two example circuits. The first example is the controlled-NOT gate. See Fig. \ref{fig:gates}(a). Bits $\hat{x}_1^{(0)},\hat{x}_2^{(0)},\hat{z}_1^{(0)},\hat{z}_2^{(0)}$ represent the Pauli operator at $t = 0$ (input to the gate), and bits $\hat{x}_1^{(1)},\hat{x}_2^{(1)},\hat{z}_1^{(1)},\hat{z}_2^{(1)}$ represent the Pauli operator at $t = 1$ (output of the gate). A codeword of the controlled-NOT gate is 
\begin{eqnarray}
&& \left(x_1^{(0)},x_2^{(0)},z_1^{(0)},z_2^{(0)},x_1^{(1)},x_2^{(1)},z_1^{(1)},z_2^{(1)}\right) \notag \\
&=& (1,0,0,0,1,1,0,0).
\end{eqnarray}
Here, $\left(x_1^{(0)},x_2^{(0)},z_1^{(0)},z_2^{(0)}\right) = (1,0,0,0)$ means that the input operator is $\sigma_{in} = X_1$, and $\left(x_1^{(1)},x_2^{(1)},z_1^{(1)},z_2^{(1)}\right) = (1,1,0,0)$ means that the output operator is $\sigma_{out} = X_1X_2$. The physical meaning of this codeword is that the controlled-NOT gate maps $X_1$ to $X_1X_2$. 

Similarly, for a general stabiliser circuit, the layer-$0$ bits represent the input Pauli operator $\sigma_{in}$, and the layer-$T$ bits represent the output Pauli operator $\sigma_{out}$. The physical meaning of a codeword is that the stabiliser circuit maps $\sigma_{in}$ to $\sigma_{out}$, up to a sign. 

\begin{figure}[tbp]
\centering
\includegraphics[width=\linewidth]{\figpath/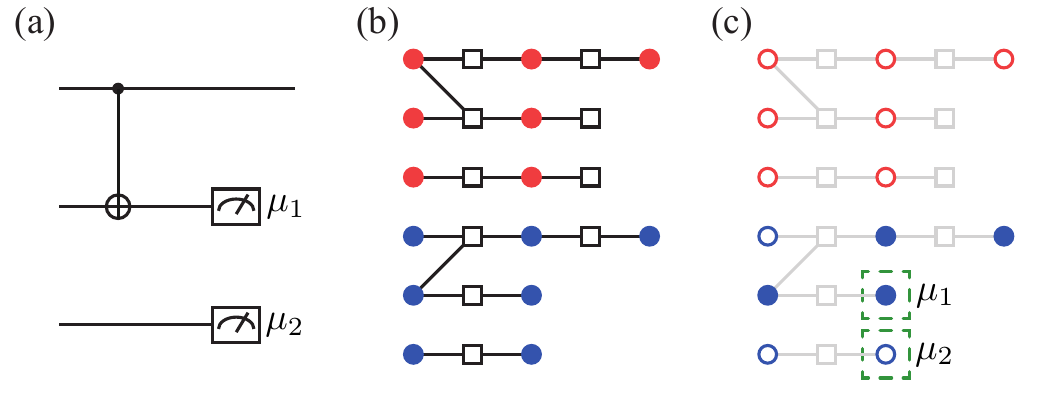}
\caption{
(a) An example stabiliser circuit with two measurements. (b) Tanner graph of the circuit. (c) A codeword. Open and closed circles denote bits taking the values of zero and one, respectively. Each bit in the green dashed box corresponds to a measurement. 
}
\label{fig:CONT_M}
\end{figure}

The second example is a circuit with measurements as shown in Fig. \ref{fig:CONT_M}(a). In this example, there is a sign factor depending on measurement outcomes. The Tanner graph of the example circuit is given in Fig. \ref{fig:CONT_M}(b), and a codeword is illustrated in Fig. \ref{fig:CONT_M}(c). According to the codeword, the input operator is $\sigma_{in} = Z_2$, and the output operator is $\sigma_{out} = Z_1$. Now, let’s suppose that the initial state is an eigenstate of $\sigma_{in}$ with the eigenvalue $+1$. In this case, the final state is an eigenstate of $\sigma_{out}$, however, the eigenvalue depends on the measurement outcome: When the outcome is $\mu_1 = \pm 1$, the eigenvalue is $\mu_1$. Therefore, the circuit maps $\sigma_{in}$ to $\mu_1\sigma_{out}$. 

In Fig. \ref{fig:CONT_M}(a), the circuit has two measurement outcomes. Only one of them is relevant to the sign factor. Notice that each measurement corresponds to a bit vertex on the Tanner graph. In the codeword, the bit of $\mu_1$ takes the value of one, then it is relevant; the bit of $\mu_2$ takes the value of zero, then it is irrelevant. In this way, only one measurement is relevant. If there are multiple relevant measurements, the sign factor is determined by the product of their outcomes. 

\subsection{Codeword equations}

To formally express the physical meaning of codewords, we need to introduce {\it layer projections} $\mathbf{P}_t : \mathbb{F}_2^{\vert V_B\vert} \rightarrow \mathbb{F}_2^{\vert V_{B,t}\vert}$. Let $v\in V_B$ be a bit on the Tanner graph and $u\in V_{B,t}$ be a bit in the layer-$t$ (notice that $u$ may not in $V_B$ if it is a removed bit). Then, matrix elements of $\mathbf{P}_t$ are 
\begin{eqnarray}
\mathbf{P}_{t;u,v} = \left\{\begin{matrix}
\delta_{u,v} & v\in V_{B,t}; \\
0, & \text{otherwise}.
\end{matrix}\right.
\end{eqnarray}

To understand the layer projection, consider a circuit with only Clifford gates. In this case, all bits are kept, i.e. $V_B = V_{B,all}$. Then, a codeword is in the form $\mathbf{c} = \{\mathbf{b}_0,\mathbf{b}_1,\ldots,\mathbf{b}_T\}\in \mathbb{F}_2^{\vert V_B\vert}$, where $\mathbf{b}_t\in \mathbb{F}_2^{\vert V_{B,t}\vert}$. The projection reads $\mathbf{P}_t\mathbf{c}^\mathrm{T} = \mathbf{b}_t^\mathrm{T}$. It is similar in the case of a general circuit. If there are removed bits, $\mathbf{c}$ does not have corresponding entries. Then, these removed bits always take the value of zero in $\mathbf{b}_t$, and other bits in $\mathbf{b}_t$ take the same values as in $\mathbf{c}$. 

With the layer projection, we can understand a codeword in the following picture. A codeword $\mathbf{c}$ is an evolution trajectory of Pauli operators in the stabiliser circuit. The input operator is $\sigma_{in}(\mathbf{c}) = \sigma(\mathbf{P}_0\mathbf{c}^\mathrm{T})$. At the time $t$, the operator becomes $\sigma(\mathbf{P}_t\mathbf{c}^\mathrm{T})$. Eventually, the output operator is $\sigma_{out}(\mathbf{c}) = \sigma(\mathbf{P}_T\mathbf{c}^\mathrm{T})$. 

Now, let's deal with measurements. Suppose a stabiliser circuit includes $n_M$ measurements (all in the $Z$ basis). We can use a subset of bits $V_M\subseteq V_B$ ($n_M = \vert V_M\vert$) to denote measurements: $\hat{z}_q^{(t-1)}\in V_M$ if a measurement is performed on qubit-$q$ at the time $t$. Let $\mu_v = \pm 1$ be the outcome of the measurement $v\in V_M$, and let $\boldsymbol{\mu}$ be an $n_M$-tuple of measurement outcomes. 

Relevant measurements determine the sign factor of the output operator. Given a codeword $\mathbf{c}$, bits $v\in V_M$ taking the value of one are relevant. Therefore, the sign factor is determined by 
\begin{eqnarray}
\mu_R(\mathbf{c},\boldsymbol{\mu}) = \prod_{v\in V_M\vert\mathbf{c}_v = 1} \mu_v.
\end{eqnarray}

Let $\rho_0$ be the initial state of the stabiliser circuit. The final state $\rho_T(\boldsymbol{\mu}) = \mathcal{M}_{\boldsymbol{\mu}}\rho_0$ depends on measurement outcomes. Here, $\rho_T(\boldsymbol{\mu})$ is unnormalised, and $\mathrm{Tr}\rho_T(\boldsymbol{\mu})$ is the probability of $\boldsymbol{\mu}$. Superoperators $\mathcal{M}_{\boldsymbol{\mu}}$ are completely positive maps, and $\sum_{\boldsymbol{\mu}} \mathcal{M}_{\boldsymbol{\mu}}$ is completely positive and trace-preserving. 

\begin{theorem}
Each codeword of a stabiliser circuit corresponds to an equation involving the initial state $\rho_0$, final state $\rho_T$, Pauli operators $\sigma_{in}$ and $\sigma_{out}$, and measurement outcomes $\boldsymbol{\mu}$. Let $\mathbf{c}\in C$ be a codeword of the stabiliser circuit. For all initial states $\rho_0$, the following equation holds: 
\begin{eqnarray}
\mathrm{Tr}\sigma_{in}(\mathbf{c})\rho_0
= \nu(\mathbf{c})\sum_{\boldsymbol{\mu}} \mu_R(\mathbf{c},\boldsymbol{\mu})\mathrm{Tr}\sigma_{out}(\mathbf{c})\rho_T(\boldsymbol{\mu}),
\label{eq:errorfree}
\end{eqnarray}
where $\nu(\mathbf{c}) = \pm 1$ is a sign factor due to Clifford gates in the circuit. 
\label{the:errorfree}
\end{theorem}

Here, $C = \mathrm{ker}\mathbf{A}$ is the space of codewords. The proof of the theorem and expression of $\nu(\mathbf{c})$ are given in Appendix \ref{app:errorfree}. Notice that given the stabiliser circuit, $\nu(\mathbf{c})$ can be evaluated in polynomial time on a classical computer \cite{Gottesman1998}. 

We call Eq. (\ref{eq:errorfree}) the error-free codeword equation because an error-free circuit is assumed. In Sec. \ref{sec:errors}, a generalised equation for circuits with errors will be introduced. 

\subsection{Checkers, detectors, emitters and propagators}

\begin{figure}[tbp]
\centering
\includegraphics[width=\linewidth]{\figpath/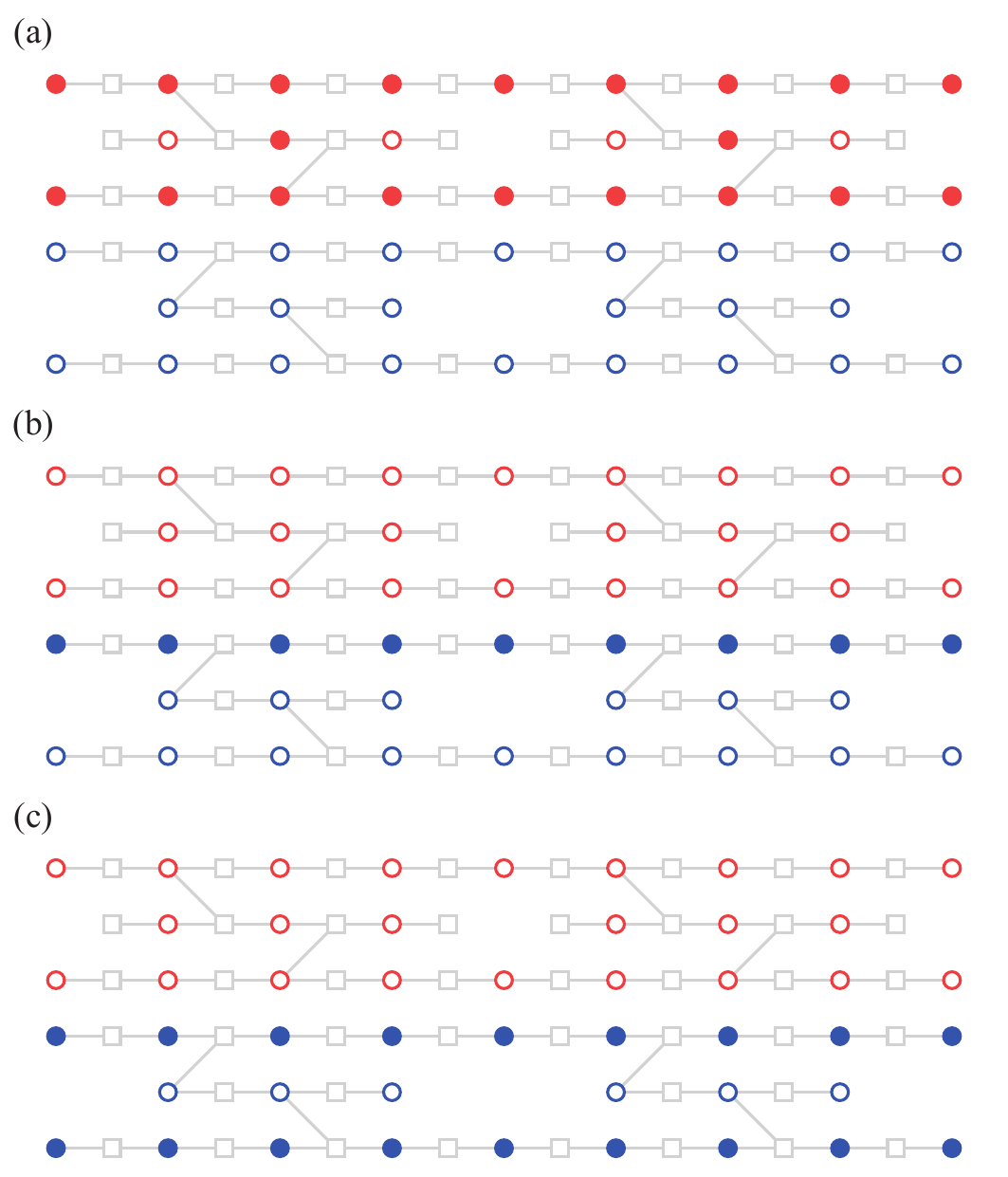}
\caption{
Propagator codewords of the circuit in Fig. 3(a). Open and closed circles denote bits taking the values of zero and one, respectively. 
}
\label{fig:codewordsP}
\end{figure}

\begin{figure}[tbp]
\centering
\includegraphics[width=\linewidth]{\figpath/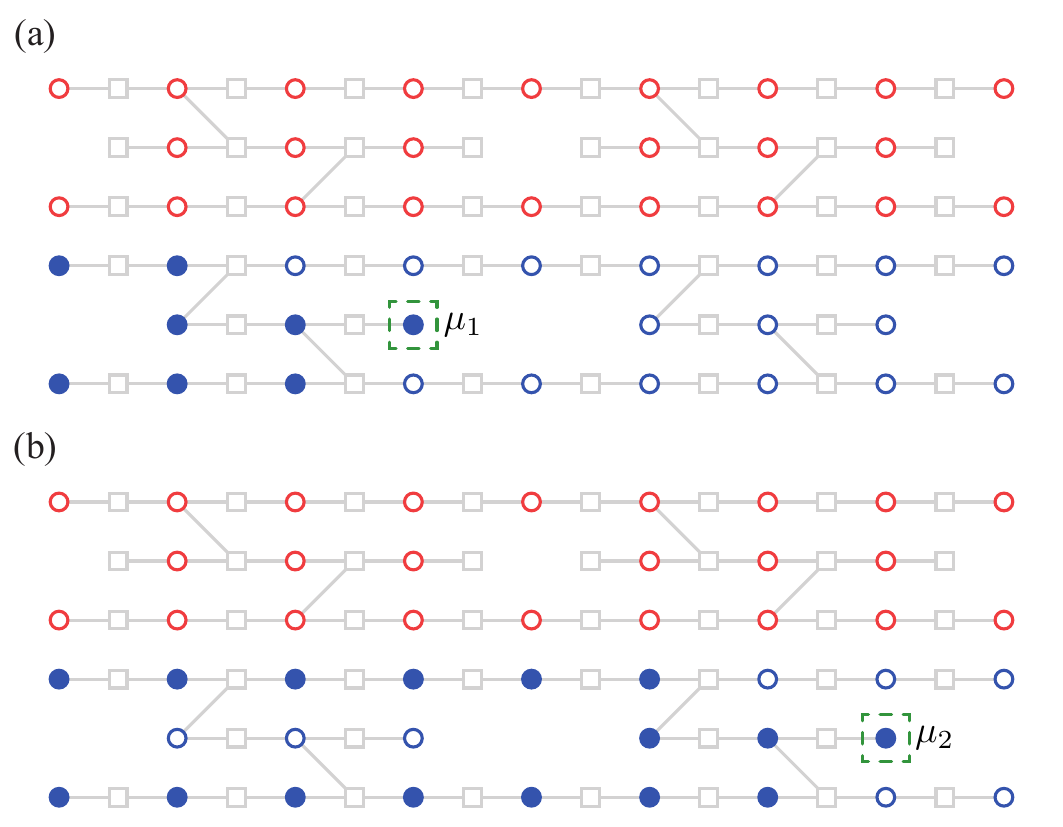}
\caption{
Detector codewords of the circuit in Fig. 3(a). Open and closed circles denote bits taking the values of zero and one, respectively. 
}
\label{fig:codewordsD}
\end{figure}

\begin{figure}[tbp]
\centering
\includegraphics[width=\linewidth]{\figpath/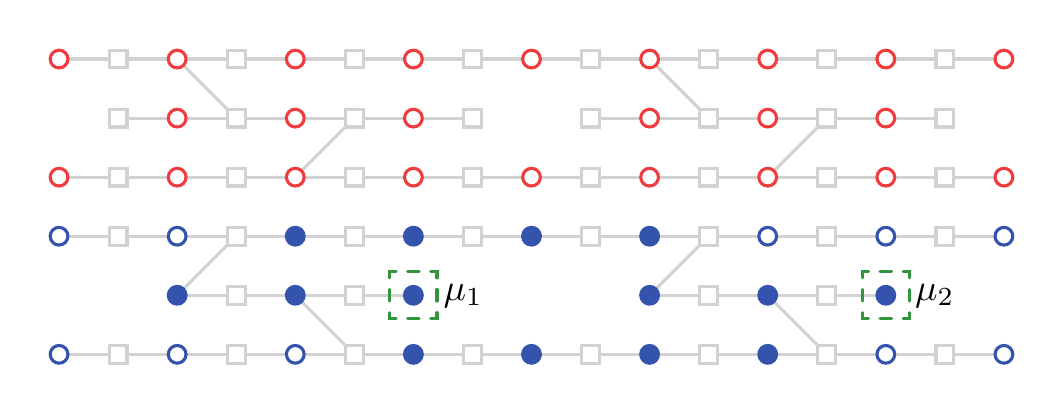}
\caption{
The checker codeword of the circuit in Fig. 3(a). Open and closed circles denote bits taking the values of zero and one, respectively. 
}
\label{fig:codewordsC}
\end{figure}

\begin{figure}[tbp]
\centering
\includegraphics[width=\linewidth]{\figpath/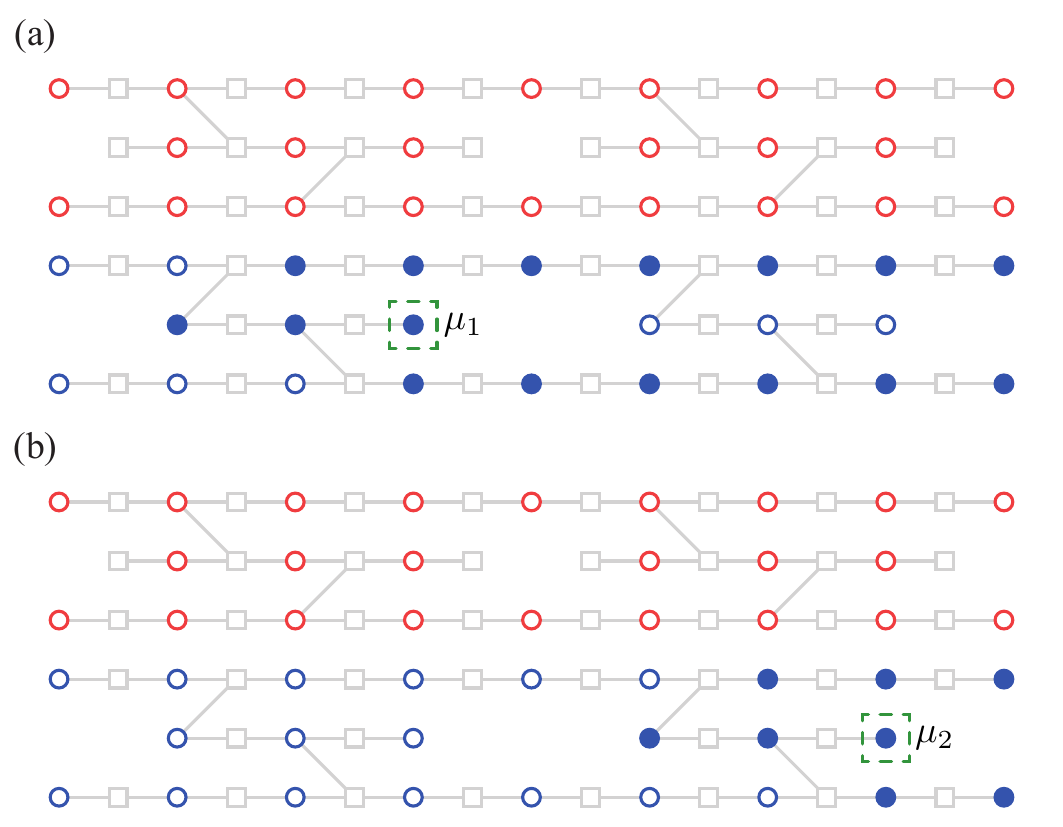}
\caption{
Emitter codewords of the circuit in Fig. 3(a). Open and closed circles denote bits taking the values of zero and one, respectively. 
}
\label{fig:codewordsE}
\end{figure}

According to their physical meanings, we can classify codewords of a stabiliser circuit into four categories: checkers, detectors, emitters and propagators. Propagators are further classified into pseudo propagators and genuine propagators. In the following, we take the parity check circuit of the [[2,1,1]] code as an example. See Figs. \ref{fig:codewordsP}, \ref{fig:codewordsD}, \ref{fig:codewordsC} and \ref{fig:codewordsE} for its codewords. 

{\it Propagators} describe the transformation of Pauli operators from the beginning to the end of the circuit. On the Tanner graph, each layer of bits represents the Pauli operator at the corresponding time. For example, in the $t = 0$ layer in Fig. \ref{fig:codewordsP}(a) [also see Fig. \ref{fig:ZZ}(b)], there are two $x$ bits taking the value of one. This means that $\sigma_{in} = X_1X_2$ is the input Pauli operator to the circuit. In the same codeword, the $t = T$ layer means that $\sigma_{out} = X_1X_2$ is the output of the circuit. This codeword describes the transformation of the logical operator $X_1X_2$: The logical operator is preserved in the circuit. We call such a codeword with nonzero-valued bits in both $t = 0$ and $t = T$ layers a propagator. 

Similarly, codewords in Figs. \ref{fig:codewordsP}(b) and (c) are also propagators, which describe the transformation of the logical operator $Z_1$ and stabiliser generator $Z_1Z_2$, respectively. 

{\it Detectors} describe measurements on Pauli operators. The two-qubit measurement $Z_1Z_2$ is realised through a single-qubit measurement $Z_\mathrm{a}$ (Here, $Z_\mathrm{a}$ is the $Z$ operator of the ancilla qubit). The codeword in Fig. \ref{fig:codewordsD}(a) describes the first $Z_1Z_2$ measurement. In the codeword, $\sigma_{in} = Z_1Z_2$ is the input at $t = 0$. Value-one bits in the codeword vanish at $t = 3$, and $Z_\mathrm{a}$ is the output at $t = 3$. This means that $Z_1Z_2$ at $t = 0$ is transformed to $Z_\mathrm{a}$ at $t = 3$; then $Z_\mathrm{a}$ is measured. We call such a codeword with nonzero-valued bits in the $t = 0$ layer but only zero-valued bits in the $t = T$ layer a detector. 

Similarly, the codeword in Fig. \ref{fig:codewordsD}(b) is also a detector, which describes the second measurement of the stabiliser generator $Z_1Z_2$: $Z_1Z_2$ at $t = 0$ is transformed to $Z_\mathrm{a}$ at $t = 7$. 

{\it Checkers} describe the parity check performed on measurement outcomes. The two outcomes $\mu_1$ and $\mu_2$ take the same value when the circuit is error-free. The parity check on them is represented by the codeword in Fig. \ref{fig:codewordsC}, which is the sum of the two codewords in Figs. \ref{fig:codewordsD}(a) and (b). We call such a codeword with only zero-valued bits in both $t = 0$ and $t = T$ layers a checker. 

{\it Emitters} describe preparing an eigenstate of Pauli operators. Consider the sum of the two codewords in Figs. \ref{fig:codewordsP}(c) and \ref{fig:codewordsD}(a), which is shown in Fig. \ref{fig:codewordsE}(a). In this codeword, value-one bits start from a bit representing qubit initialisation, and $\sigma_{out} = Z_1Z_2$ is the output operator at $t = T$. In this case, the final state is an eigenstate of $Z_1Z_2$, and the eigenvalue depends on the measurement outcome $\mu_1$. We call such a codeword with only zero-valued bits in the $t = 0$ layer but nonzero-valued bits in the $t = T$ layer an emitter. 

Similarly, the codeword in Fig. \ref{fig:codewordsE}(b) is also an emitter, which is the sum of two codewords in Figs. \ref{fig:codewordsP}(c) and \ref{fig:codewordsD}(b). 

Now, we can formally define checkers, detectors, emitters and propagators. The addition of two checkers is also a checker. Therefore, checkers form a subspace, which is $C_c = \mathrm{ker}\mathbf{A} \cap \mathrm{ker}\mathbf{P}_0 \cap \mathrm{ker}\mathbf{P}_T$. Similarly, checkers and detectors form the subspace $C_{c,d} = \mathrm{ker}\mathbf{A} \cap \mathrm{ker}\mathbf{P}_T$. Codewords in $C_{c,d} - C_c$ are detectors. Checkers and emitters also form a subspace $C_{c,e} = \mathrm{ker}\mathbf{A} \cap \mathrm{ker}\mathbf{P}_0$. Codewords in $C_{c,e} - C_c$ are emitters. Codewords in $C - C_{c,d}\cup C_{c,e}$ are propagators. 

For a checker codeword $\mathbf{c}\in C_c$, we have $\sigma_{in}(\mathbf{c}) = \sigma_{out}(\mathbf{c}) = \openone$. The corresponding codeword equation is $1 = \nu(\mathbf{c})\sum_{\boldsymbol{\mu}} \mu_R(\mathbf{c},\boldsymbol{\mu})\mathrm{Tr}\rho_T(\boldsymbol{\mu})$. Here, $\mathrm{Tr}\rho_T(\boldsymbol{\mu})$ is the probability of $\boldsymbol{\mu}$. This equation means that $\nu(\mathbf{c})\mu_R(\mathbf{c},\boldsymbol{\mu}) = 1$ holds with the probability of one. Therefore, it can be used in the parity check for detecting errors. 

For a detector codeword $\mathbf{c}\in C_{c,d}-C_c$, we have $\sigma_{out}(\mathbf{c}) = \openone$. The corresponding codeword equation is $\mathrm{Tr}\sigma_{in}(\mathbf{c})\rho_0 = \nu(\mathbf{c})\sum_{\boldsymbol{\mu}} \mu_R(\mathbf{c},\boldsymbol{\mu})\mathrm{Tr}\rho_T(\boldsymbol{\mu})$, which means that $\nu(\mathbf{c})\mu_R(\mathbf{c},\boldsymbol{\mu})$ is the measurement outcome of $\sigma_{in}(\mathbf{c})$. Notice that $\mathrm{Tr}\rho_T(\boldsymbol{\mu})$ is the probability of $\boldsymbol{\mu}$.  

For an emitter codeword $\mathbf{c}\in C_{c,e}-C_c$, we have $\sigma_{in}(\mathbf{c}) = \openone$. The corresponding codeword equation is $1 = \nu(\mathbf{c})\sum_{\boldsymbol{\mu}} \mu_R(\mathbf{c},\boldsymbol{\mu})\mathrm{Tr}\sigma_{out}(\mathbf{c})\rho_T(\boldsymbol{\mu})$, i.e. $\nu(\mathbf{c})\mu_R(\mathbf{c},\boldsymbol{\mu})\frac{\mathrm{Tr}\sigma_{out}(\mathbf{c})\rho_T(\boldsymbol{\mu})}{\mathrm{Tr}\rho_T(\boldsymbol{\mu})} = 1$ holds with the probability of one. Here, $\frac{\mathrm{Tr}\sigma_{out}(\mathbf{c})\rho_T(\boldsymbol{\mu})}{\mathrm{Tr}\rho_T(\boldsymbol{\mu})}$ is the mean value of $\sigma_{out}(\mathbf{c})$ in the normalised final state. Therefore, the equation means that the normalised final state $\frac{\rho_T(\boldsymbol{\mu})}{\mathrm{Tr}\rho_T(\boldsymbol{\mu})}$ is an eigenstate of $\sigma_{out}(\mathbf{c})$ with the eigenvalue $\nu(\mathbf{c})\mu_R(\mathbf{c},\boldsymbol{\mu})$. 

\subsection{Pseudo propagators and genuine propagators}

{\it Genuine propagators} represent coherent correlations, and {\it pseudo propagators} represent incoherent correlations. Consider codewords in Figs. \ref{fig:codewordsP}(b) and (c). They represent transformations from $Z_1$ and $Z_1Z_2$ at $t = 0$ to $Z_1$ and $Z_1Z_2$ at $t = T$, respectively. Their difference is that the stabiliser generator $Z_1Z_2$ is measured in the circuit. Because of the measurement, the superposition between two eigenstates of $Z_1Z_2$ is unpreserved. In comparison, the superposition between two eigenstates of the logical operator $Z_1$ is preserved. We call codewords in Figs. \ref{fig:codewordsP}(b) and (c) genuine and pseudo propagators, respectively. A pseudo propagator is a linear combination of checkers, detectors and emitters. For example, the codeword in Fig. \ref{fig:codewordsP}(c) is the sum of two codewords in Figs. \ref{fig:codewordsD}(a) and \ref{fig:codewordsE}(a). A genuine propagator cannot be decomposed into checkers, detectors and emitters. 

For the formal definition, the subspace spanned by checkers, detectors and emitters is $C_{c,d,e} = \mathrm{Span}(C_{c,d}\cup C_{c,e})$. Codewords in $C_{c,d,e} - C_{c,d}\cup C_{c,e}$ are pseudo propagators. Codewords in $C - C_{c,d,e}$ are genuine propagators. The difference between these two types of propagators are given by the following theorem and corollary. 

\begin{theorem}
Genuine propagators represent coherent correlations, and all other codewords represent incoherent correlations. Let $\mathbf{c}\in C$ be a codeword of the stabiliser circuit. Then, the following two statements hold: 
\begin{itemize}
\item[i)] If and only if $\mathbf{c}\in C_{c,d,e}$, its input operator commutes with input operators of all codewords, i.e. $[\sigma_{in}(\mathbf{c}), \sigma_{in}(\mathbf{c}')] = 0$ for all $\mathbf{c}'\in C$; 
\item[ii)] If and only if $\mathbf{c}\in C_{c,d,e}$, its output operator commutes with output operators of all codewords, i.e. $[\sigma_{out}(\mathbf{c}), \sigma_{out}(\mathbf{c}')] = 0$ for all $\mathbf{c}'\in C$. 
\end{itemize}
\label{the:propagator}
\end{theorem}

\begin{corollary}
Let $\mathbf{c}\in C - C_{c,d,e}$ be a genuine propagator of the stabiliser circuit. There always exists a genuine propagator whose input and output operators anti-commute with the input and output operators of $\mathbf{c}$, respectively, i.e. $\exists \mathbf{c}'\in C - C_{c,d,e}$ such that $\{\sigma_{in}(\mathbf{c}), \sigma_{in}(\mathbf{c}')\} = \{\sigma_{out}(\mathbf{c}), \sigma_{out}(\mathbf{c}')\} = 0$. 
\label{coro:propagator}
\end{corollary}

See Appendix \ref{app:propagator} for the proofs. For genuine propagators, the two anti-commutative operators $\sigma_{in}(\mathbf{c})$ and $\sigma_{in}(\mathbf{c}')$ defines a logical qubit. The superposition in such a logical qubit is preserved in the circuit and transferred to the logical qubit defined by $\sigma_{out}(\mathbf{c})$ and $\sigma_{out}(\mathbf{c}')$. 

\section{Error correction, logical operations and code distance}
\label{sec:distance}

In this section, we discuss the error correction in the LDPC representation. Usually, we start from a quantum error correction code and then compose a circuit to realise the code. In this section, we show that we can think in an alternative way: We can start from a circuit and ask what quantum error correction codes and logical operations can be realised by the circuit. This section ends with a definition of the code distance for a stabiliser circuit, which quantifies its fault tolerance. 

\subsection{Errors in spacetime}
\label{sec:errors}

Before discussing the quantum error correction, we consider the impact of errors on codeword equations. We can represent Pauli errors that occurred in the circuit with a vector $\mathbf{e}\in \mathbb{F}_2^{\vert V_B\vert}$. For bits $v=\hat{x}_q^{(t)}$ ($v=\hat{z}_q^{(t)}$), $\mathbf{e}_v = 1$ means that there is a $Z$ ($X$) error on qubit-$q$ following operations in layer-$t$, i.e. the error flips the sign of the Pauli operator corresponding to $v$. To avoid ambiguity, we call a $Z$ or $X$ error that occurred on a certain qubit at a certain time a single-bit error, and we call $\mathbf{e}$ a spacetime error. A spacetime error may include many single-bit errors. 

For a codeword $\mathbf{c}$ of the circuit, $\sigma(\mathbf{P}_t\mathbf{c}^\mathrm{T})$ is the Pauli operator at the time $t$. Similarly, $\mathbf{P}_t\mathbf{e}^\mathrm{T}$ represents Pauli errors that occurred at the time $t$. When $(\mathbf{P}_t\mathbf{c}^\mathrm{T})_v = 1$ and $(\mathbf{P}_t\mathbf{e}^\mathrm{T})_v = 1$ hold at the same time for a certain $v$, the sign of $\sigma(\mathbf{P}_t\mathbf{c}^\mathrm{T})$ is flipped for once. Therefore, the sign is flipped by errors at the time $t$ if and only if $(\mathbf{P}_t\mathbf{c}^\mathrm{T})^\mathrm{T}(\mathbf{P}_t\mathbf{e}^\mathrm{T}) = 1$, and the sign is flipped by all the errors in spacetime if and only if $\mathbf{c}\mathbf{e}^\mathrm{T} = 1$. 

To express the impact of errors, we introduce a generalised codeword equation that holds when the circuit has errors. With the spacetime error $\mathbf{e}$, the output state becomes $\rho_T(\mathbf{e},\boldsymbol{\mu}) = \mathcal{M}_{\mathbf{e},\boldsymbol{\mu}}\rho_0$, where $\mathcal{M}_{\mathbf{e},\boldsymbol{\mu}}$ are completely positive maps, and $\sum_{\boldsymbol{\mu}} \mathcal{M}_{\mathbf{e},\boldsymbol{\mu}}$ is completely positive and trace-preserving. 

\begin{theorem}
Let $\mathbf{c}\in C$ be a codeword of the stabiliser circuit, and let $\mathbf{e}\in \mathbb{F}_2^{\vert V_B\vert}$ be a spacetime error. The error flips the sign of the output Pauli operator if and only if $\mathbf{c}\mathbf{e}^\mathrm{T} = 1$. For all initial states $\rho_0$, the following equation holds: 
\begin{eqnarray}
&&\mathrm{Tr}\sigma_{in}(\mathbf{c})\rho_0 \notag \\
&=& (-1)^{\mathbf{c}\mathbf{e}^\mathrm{T}}\nu(\mathbf{c})\sum_{\boldsymbol{\mu}} \mu_R(\mathbf{c},\boldsymbol{\mu})\mathrm{Tr}\sigma_{out}(\mathbf{c})\rho_T(\mathbf{e},\boldsymbol{\mu}).
\label{eq:generalised}
\end{eqnarray}
\label{the:generalised}
\end{theorem}

We leave the proof to Appendix \ref{app:generalised}. 

\subsection{Error correction of a stabiliser circuit}
\label{sec:error_correction}

With the generalised codeword equation, we can find that a checker $\mathbf{c}\in C_c$ can detect errors. When the circuit is error-free, measurement outcomes always satisfy $\nu(\mathbf{c})\mu_R(\mathbf{c},\boldsymbol{\mu}) = 1$. With the spacetime error $\mathbf{e}$, the codeword equation of a checker becomes $\nu(\mathbf{c})\mu_R(\mathbf{c},\boldsymbol{\mu}) = (-1)^{\mathbf{c}\mathbf{e}^\mathrm{T}}$. Therefore, the checker can detect errors $\mathbf{e}$ that satisfy $\mathbf{c}\mathbf{e}^\mathrm{T} = 1$, and errors are detected when one observes $\nu(\mathbf{c})\mu_R(\mathbf{c},\boldsymbol{\mu}) = -1$. 

Detectors, emitters and propagators can also detect errors under certain conditions. In general, the initial and final states may correspond to different quantum error correction codes. Let $S_{in}$ and $S_{out}$ be the corresponding stabiliser groups, respectively. Then, the initial state is prepared in a stabiliser state of $S_{in}$, and generators of $S_{out}$ are measured on the final state by operations subsequent to the circuit. A detector $\mathbf{c}\in C_{c,d} - C_c$ is a measurement of $\sigma_{in}(\mathbf{c})$. It can detect errors if $\sigma_{in}(\mathbf{c})\in S_{in}\times\{\pm 1\}$ because one can compare the measurement outcome with the initial stabiliser state. An emitter $\mathbf{c}\in C_{c,e} - C_c$ prepares the eigenstate of $\sigma_{out}(\mathbf{c})$. It can detect errors if $\sigma_{out}(\mathbf{c})\in S_{out}\times\{\pm 1\}$ because one can compare the eigenvalue with the subsequent measurement outcome. Similarly, a propagator $\mathbf{c}\in C - C_{c,d}\cup C_{c,e}$ can detect errors if $\sigma_{in}(\mathbf{c})\in S_{in}\times\{\pm 1\}$ and $\sigma_{out}(\mathbf{c})\in S_{out}\times\{\pm 1\}$. 

In summary, codewords $\mathbf{c}$ can detect errors if and only if $\sigma_{in}(\mathbf{c})\in S_{in}\times\{\pm 1\}$ and $\sigma_{out}(\mathbf{c})\in S_{out}\times\{\pm 1\}$ (notice that $\openone\in S_{in},S_{out}$). Because $S_{in}\times\{\pm 1\}$ and $S_{out}\times\{\pm 1\}$ are groups, these error-detecting codewords form a subspace. Let $\mathbf{c}_1,\mathbf{c}_2,\ldots\in C$ be a basis of this subspace. Then, we have the check matrix of the error correction realised by the circuit, 
\begin{eqnarray}
\mathbf{B} = \left(\begin{matrix}
\mathbf{c}_1 \\
\mathbf{c}_2 \\
\vdots
\end{matrix}\right).
\end{eqnarray}
We call it the {\it error-correction check matrix} to distinguish it from the check matrix $\mathbf{A}$ representing the circuit. 

A valid error-correction check matrix $\mathbf{B}$ must satisfy the following two conditions: i) $\mathbf{A}\mathbf{B}^\mathrm{T} = 0$; and ii) for all $\mathbf{c},\mathbf{c}'\in \mathrm{rowsp}(\mathbf{B})$, $[\sigma_{in}(\mathbf{c}),\sigma_{in}(\mathbf{c}')] = [\sigma_{out}(\mathbf{c}),\sigma_{out}(\mathbf{c}')] = 0$. Here, $\mathrm{rowsp}(\bullet)$ denotes the row space. 

We have constructed the error-correction check matrix assuming two stabiliser groups $S_{in}$ and $S_{out}$. We can also first construct a valid error-correction check matrix and then work out the two groups. Let $\{\mathbf{c}_i\}$ be a basis of $\mathrm{rowsp}(\mathbf{B})$. Then, $\langle \{\sigma_{in}(\mathbf{c}_i)\}\rangle$ is a subgroup of $S_{in}\times\{\pm 1\}$. To remove the ambiguity due to the sign, let $\{\mathbf{b}_j\}$ be a basis of $\mathrm{Span}\left(\{(\mathbf{P}_0\mathbf{c}_i^\mathrm{T})^\mathrm{T}\}\right)$. Then, $S_{in} = \langle \{\sigma(\mathbf{b}_j)\}\rangle$. We can work out $S_{out}$ in a similar way, in which $\mathbf{P}_0$ is replaced by $\mathbf{P}_T$. 

\subsection{Logical operations in a stabiliser circuit}
\label{sec:logical_operation}

In quantum error correction, a stabiliser circuit realises certain operations on logical qubits. Let $L_{in}$ and $L_{out}$ be logical Pauli groups of the initial and final states, respectively. Then, $(S_{in}, L_{in})$ and $(S_{out}, L_{out})$ define the initial and final quantum error correction codes, respectively. A detector $\mathbf{c}\in C_{c,d} - C_c$ represents the measurement on a logical Pauli operator if $\sigma_{in}(\mathbf{c})\in S_{in}L_{in}$. An emitter $\mathbf{c}\in C_{c,e} - C_c$ represents preparing logical qubits in the eigenstate of a logical Pauli operator if $\sigma_{out}(\mathbf{c})\in S_{out}L_{out}$. A propagator $\mathbf{c}\in C - C_{c,d}\cup C_{c,e}$ represents the transformation of logical Pauli operators if $\sigma_{in}(\mathbf{c})\in S_{in}L_{in}$ and $\sigma_{out}(\mathbf{c})\in S_{out}L_{out}$. The transformation is coherent if $\mathbf{c}\in C - C_{c,d,e}$ is a genuine propagator. 

Similar to the error-correction check matrix $\mathbf{B}$, all codewords satisfying $\sigma_{in}(\mathbf{c})\in S_{in}L_{in}$ and $\sigma_{out}(\mathbf{c})\in S_{out}L_{out}$ form a subspace. We can find that all error-detecting codewords are in this subspace. Let $\mathbf{c}_1,\mathbf{c}_2,\ldots,\mathbf{c}'_1,\mathbf{c}'_2,\ldots\in C$ be a basis of the subspace, where $\mathbf{c}_1,\mathbf{c}_2,\ldots$ is the basis of $\mathrm{rowsp}(\mathbf{B})$. Then, we have the {\it logical generator matrix} 
\begin{eqnarray}
\mathbf{L} = \left(\begin{matrix}
\mathbf{c}'_1 \\
\mathbf{c}'_2 \\
\vdots
\end{matrix}\right).
\end{eqnarray}

\begin{definition}
Let $\mathbf{A}$ be the check matrix representing a stabiliser circuit. The error correction check matrix $\mathbf{B}$ and logical generator matrix $\mathbf{L}$ are said to be compatible with $\mathbf{A}$ if and only if the following two conditions are satisfied: i) $\mathbf{A}\mathbf{B}^\mathrm{T} = \mathbf{A}\mathbf{L}^\mathrm{T} = 0$; and ii) $\mathrm{rowsp}(\mathbf{B})$ and $\mathrm{rowsp}(\mathbf{L})$ are linearly independent. 
\end{definition}

Similar to the error-correction check matrix $\mathbf{B}$, we can work out logical Pauli operators from $\mathbf{L}$: $\sigma_{in}(\mathbf{c}'_i)$ and $\sigma_{out}(\mathbf{c}'_i)$ are logical Pauli operators in $L_{in}$ and $L_{out}$, respectively. 

\subsection{Categories of errors}

A spacetime error $\mathbf{e}$ can be detected by certain error-detecting codewords if and only if $\mathbf{B}\mathbf{e}\neq 0$, i.e. the subspace of undetectable errors is $\mathrm{ker}\mathbf{B}$. 

With the spacetime error $\mathbf{e}$, the logical information may be incorrect in the final state, and logical measurement outcomes may also be incorrect. Let $\mathbf{c}\in \mathrm{rowsp}(\mathbf{L})$ be a codeword representing a logical operation. When $\mathbf{c}\mathbf{e}^\mathrm{T} = 1$, the codeword equation of $\mathbf{c}$ has an extra sign. In this situation, there is an error in the logical measurement outcome $\nu(\mathbf{c})\mu_R(\mathbf{c},\boldsymbol{\mu})$ if $\mathbf{c}$ is a detector, or on the output logical Pauli operator $\sigma_{out}(\mathbf{c})$ if $\mathbf{c}$ is an emitter or propagator. For an undetectable error $\mathbf{e}\in \mathrm{ker}\mathbf{B}$, it does not cause any logical error if and only if $\mathbf{L}\mathbf{e}^\mathrm{T} = 0$. Therefore, the set of trivial errors is $\mathrm{ker}\mathbf{B}\cap \mathrm{ker}\mathbf{L}$. The set of (undetectable) logical errors is $\mathrm{ker}\mathbf{B} - \mathrm{ker}\mathbf{B}\cap \mathrm{ker}\mathbf{L}$. 

Some errors always have exactly the same effect. We call them equivalent errors. 

\begin{definition}
Two spacetime errors $\mathbf{e}_1$ and $\mathbf{e}_2$ are said to be equivalent if and only if $\mathbf{c}\mathbf{e}_1^\mathrm{T} = \mathbf{c}\mathbf{e}_2^\mathrm{T}$ for all codewords $\mathbf{c}\in C$ of the stabiliser circuit. 
\end{definition}

\begin{proposition}
Two spacetime errors $\mathbf{e}_1$ and $\mathbf{e}_2$ are equivalent if and only if $\mathbf{e}_1+\mathbf{e}_2\in \mathrm{rowsp}(\mathbf{A})$. 
\label{prop:equivalence}
\end{proposition}

\subsection{Code distance of a stabiliser circuit}
\label{sec:distance}

In a stabiliser circuit with the spacetime error $\mathbf{e}$, the total number of $X$ and $Z$ errors is $\vert\mathbf{e}\vert$. Then, it is straightforward to define the code distance of a stabiliser circuit. 

\begin{definition}
Let $\mathbf{A}$ be the check matrix representing the stabiliser circuit. Suppose the error correction check matrix $\mathbf{B}$ and logical generator matrix $\mathbf{L}$ are compatible with $\mathbf{A}$. The code distance of a stabiliser circuit, termed circuit code distance, is 
\begin{eqnarray}
d(\mathbf{A},\mathbf{B},\mathbf{L}) = \min_{\mathbf{e}\in\mathrm{ker}\mathbf{B} - \mathrm{ker}\mathbf{B}\cap \mathrm{ker}\mathbf{L}} \vert\mathbf{e}\vert.
\end{eqnarray}
\label{def:distance}
\end{definition}

An interesting situation is that rows of $\mathbf{B}$ and $\mathbf{L}$ are complete, i.e. span the code of $\mathbf{A}$. In this situation, 
\begin{eqnarray}
d(\mathbf{A},\mathbf{B},\mathbf{L}) = \min_{\mathbf{e}\in\mathrm{ker}\mathbf{B} - \mathrm{rowsp}(\mathbf{A})} \vert\mathbf{e}\vert,
\end{eqnarray}
where $\mathrm{rowsp}(\mathbf{A})$ is the dual code of $C$. This is similar to the code distance of CSS codes: Think of that $\mathbf{A}$ and $\mathbf{B}$ are the $X$-operator and $Z$-operator check matrices of a CSS code, respectively; then, the circuit code distance mirrors the CSS code distance of $X$ errors. 

In the conventional definition of the quantum code distance, the distance is the minimum number of single-qubit errors resulting in a logical error. In comparison, $d(\mathbf{A},\mathbf{B},\mathbf{L})$ is the minimum number of single-bit errors resulting in a logical error. The minimum number of single-qubit errors may be smaller than $d(\mathbf{A},\mathbf{B},\mathbf{L})$, because two errors $X$ and $Z$ may occur on the same qubit. Instead, $d(\mathbf{A},\mathbf{B},\mathbf{L})/2$ is a lower bound of the minimum number of single-qubit errors in the circuit constituting a logical error. 

\section{Bit-check symmetry}
\label{sec:symmetry}

This section introduces the bit-check symmetry. We start with primitive operations and show that their Tanner graphs have the symmetry. Then, we proceed to general stabiliser circuits. For a general circuit, its plain Tanner graph may not have the symmetry; however, it is always equivalent to a Tanner graph with the symmetry. The plain Tanner graph and corresponding symmetric Tanner graph are related by a simple graph operation, which is called bit splitting in this paper. The two graphs are said to be equivalent because the two codes are isomorphic and have the same circuit code distance. Therefore, the symmetric Tanner graph can represent the stabiliser circuit equally well as the plain Tanner graph. 

\subsection{Bit-check symmetry of primitive operations}

\begin{figure}[tbp]
\centering
\includegraphics[width=\linewidth]{\figpath/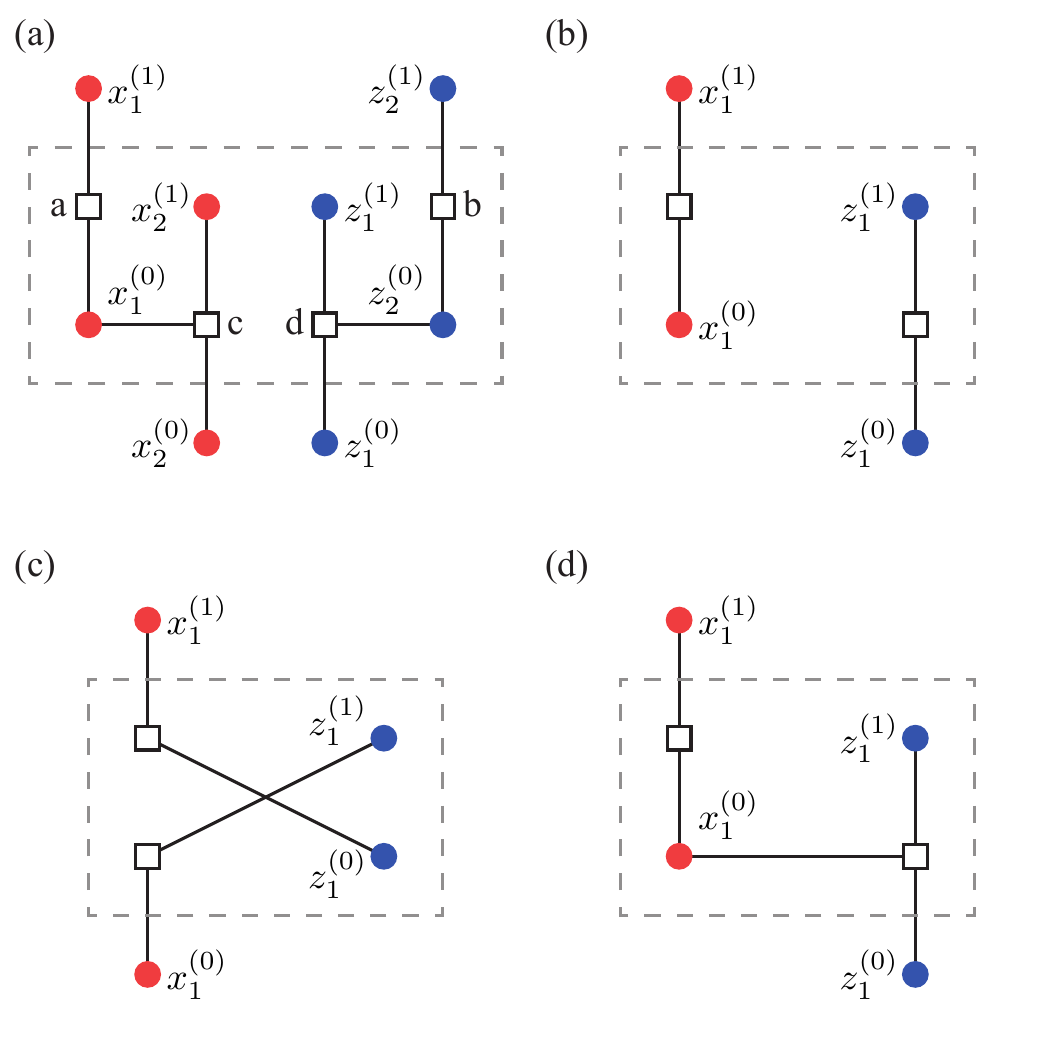}
\caption{
Redrawn Tanner graphs of the (a) controlled-NOT gate, (b) identity gate, (c) Hadamard gate and (d) phase gate. 
}
\label{fig:gates_symmetry}
\end{figure}

\begin{figure}[tbp]
\centering
\includegraphics[width=\linewidth]{\figpath/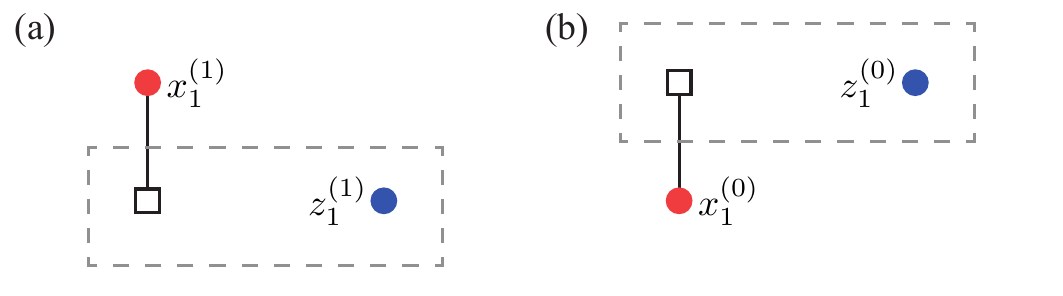}
\caption{
Redrawn Tanner graphs of the (a) initialisation and (b) measurement. 
}
\label{fig:IandM_symmetry}
\end{figure}

To illustrate the symmetry, we redraw Tanner graphs of primitive operations as shown in Figs. \ref{fig:gates_symmetry} and \ref{fig:IandM_symmetry}. Graphs are the same, and only the locations of vertices are moved. We can find that subgraphs in grey dashed boxes are symmetric. For each bit, there is a check on the same level corresponding to the same qubit; we call them the {\it dual vertices} of each other. For example, in the controlled-NOT gate, $\{\hat{z}_1^{(1)},a\}$, $\{\hat{x}_2^{(1)},b\}$, $\{\hat{z}_2^{(0)},c\}$ and $\{\hat{x}_1^{(0)},d\}$ are dual vertices of each other. If we exchange all pairs of dual vertices, subgraphs stay the same. 

We can find that some bits are not on the symmetric subgraph. We call them {\it long terminals}. For example, in the controlled-NOT gate, $\hat{x}_1^{(1)}$, $\hat{z}_2^{(1)}$, $\hat{x}_2^{(0)}$ and $\hat{z}_1^{(0)}$ are long terminals. Other bits are called {\it short terminals}. When $\hat{x}_q^{(t)}$ ($\hat{z}_q^{(t)}$) is a long terminal, $\hat{z}_q^{(t)}$ ($\hat{x}_q^{(t)}$) is always a short terminal, and vice versa. All long terminals have a degree of one, and they are all connected with different checks. 

After deleting long terminals from the Tanner graph, we obtain the symmetric subgraph. We can describe deleting long terminals with a matrix $\mathbf{D}\in \mathbb{F}_2^{\vert V_B\vert\times\vert V_C\vert}$. A deleting matrix $\mathbf{D}$ is a full rank matrix, and only one entry takes the value of one in each column, i.e. $\vert \mathbf{D}_{\bullet,a}\vert = 1$ for all $a\in V_C$ (Therefore, a deleting matrix has the property $\mathbf{D}^\mathrm{T}\mathbf{D} = \openone_{\vert V_C\vert}$). Each row of the deleting matrix corresponds to a bit. Rows corresponding to long terminals are zeros, i.e. a bit $v\in V_B$ is a long terminal if and only if $\mathbf{D}_{v,\bullet} = 0$. We use $D$ to denote the set of deleting matrices. 

Let $\mathbf{A}$ be the check matrix of a primitive operation. Then, $\mathbf{A}\mathbf{D}$ is the check matrix of the symmetric subgraph. In a Tanner graph, each bit (check) vertex corresponds to a column (row) of the check matrix. If we arrange the rows properly (which can be realised by choosing a proper $\mathbf{D}$), $\mathbf{A}\mathbf{D}$ is a symmetric matrix due to the symmetry of the subgraph. 

\begin{definition}
Let $\mathbf{A}$ be the check matrix of a Tanner graph. It has bit-check symmetry if and only if there exists a deleting matrix $\mathbf{D}\in D$ such that: 
\begin{itemize}
\item[i)] Upon the deletion of long terminals, the resulting matrix is symmetric, i.e. $(\mathbf{A}\mathbf{D})^\mathrm{T} = \mathbf{A}\mathbf{D}$; 
\item[ii)] All long terminals have the vertex degree of one, i.e. $\vert \mathbf{A}_{\bullet,v}\vert = 1$ for all bits $v$ in $V_L = \{v\in V_B\vert \mathbf{D}_{v,\bullet} = 0\}$; 
\item[iii)] All long terminals are connected with different checks, i.e. for all bits $u,v\in V_L$, $\mathbf{A}_{\bullet,u}\neq \mathbf{A}_{\bullet,v}$ if $u\neq v$. 
\end{itemize}
\label{def:symmetry}
\end{definition}

\subsection{Bit-check symmetry of stabiliser circuits}
\label{sec:CircuitSymmetry}

\begin{figure}[tbp]
\centering
\includegraphics[width=\linewidth]{\figpath/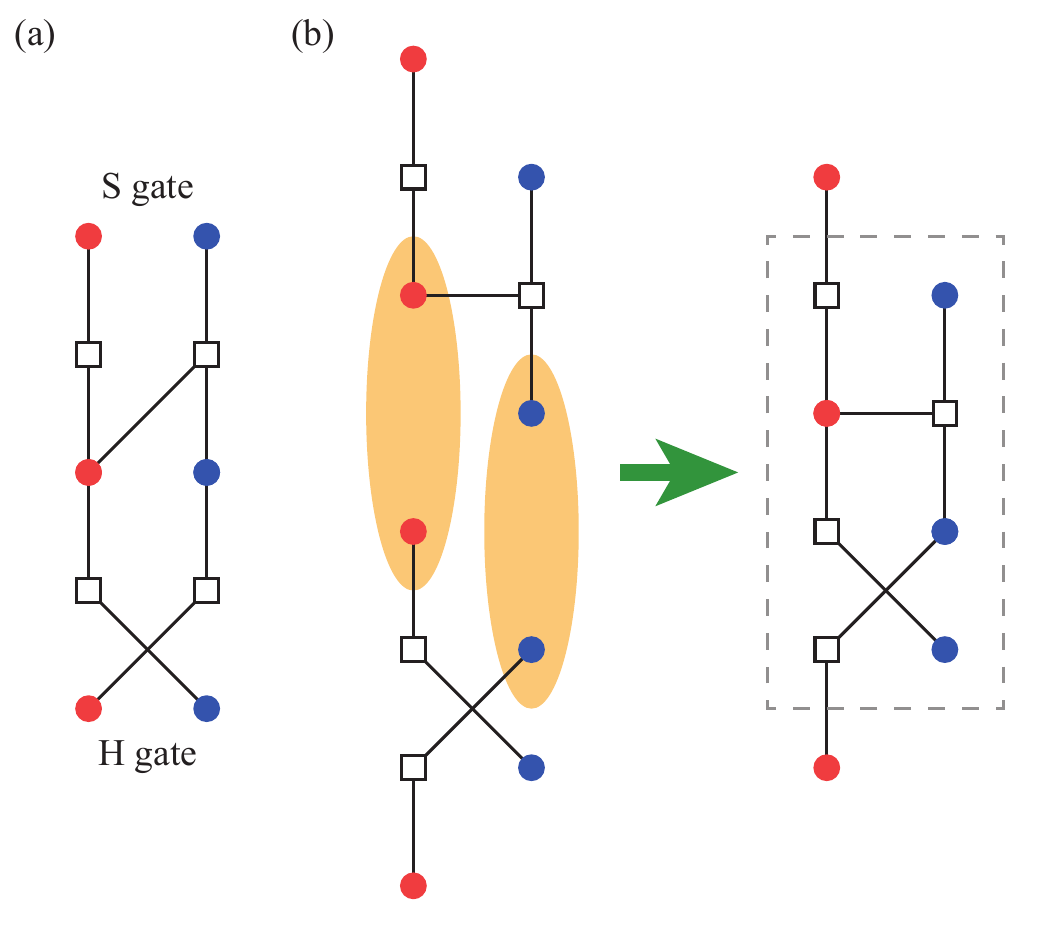}
\caption{
(a) Tanner graph of the circuit $SH$. (b) Redrawn Tanner graph of the circuit $SH$. The two graphs are the same. 
}
\label{fig:circuit_symmetry1}
\end{figure}

\begin{figure}[tbp]
\centering
\includegraphics[width=\linewidth]{\figpath/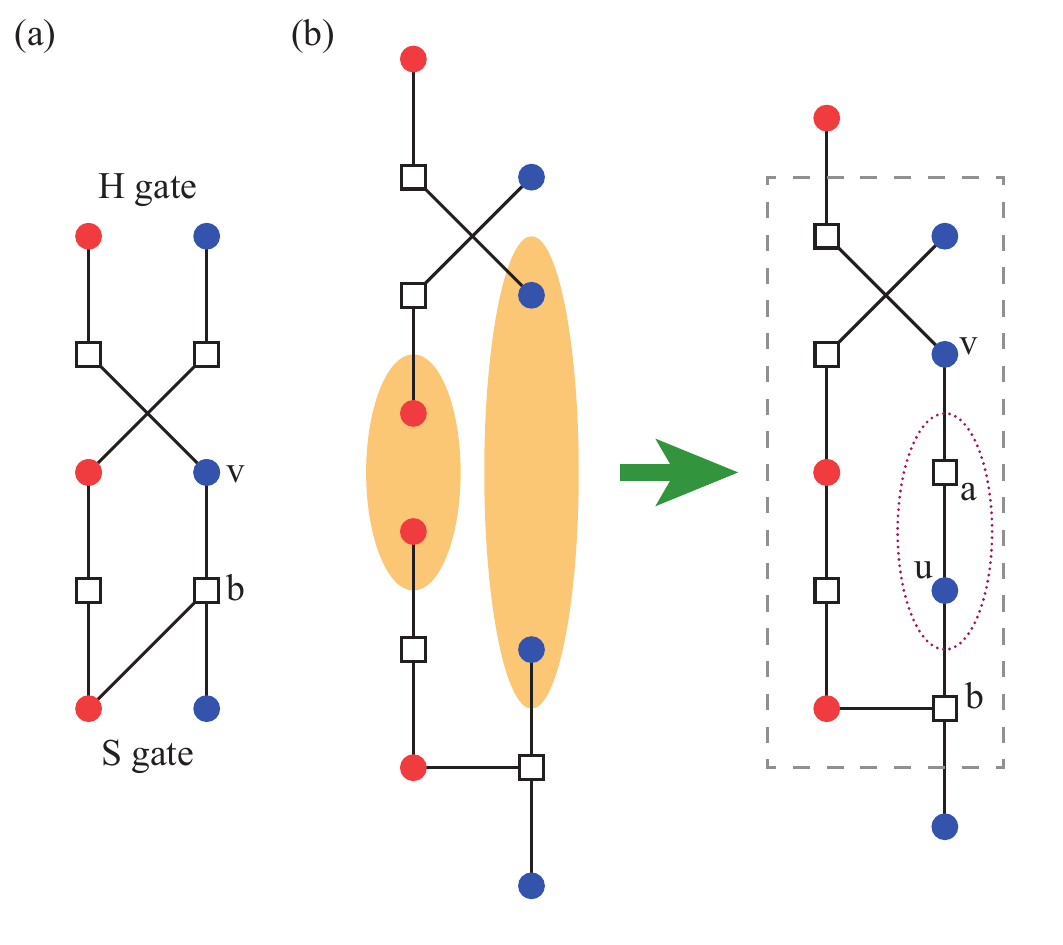}
\caption{
(a) Tanner graph of the circuit $HS$. (b) Redrawn Tanner graph of the circuit $HS$ with bit splitting. Before bit splitting (without vertices $u$ and $a$), the two graphs are the same. Because bits $v$ and $u$ are connected through $a$, which is a degree-2 check, they always take the same value in a codeword. According to Proposition \ref{prop:equivalence}, single-bit errors on $v$ and $u$ are equivalent. 
}
\label{fig:circuit_symmetry2}
\end{figure}

Now, let’s consider a circuit. We take circuits consisting of a Hadamard gate and a phase gate as examples. If the Hadamard gate is applied first then the phase gate, Fig. \ref{fig:circuit_symmetry1}(a) illustrates the plain Tanner graph of the circuit. The plain Tanner graph is obtained by merging two Tanner graphs of primitive gates [see Fig. \ref{fig:circuit_symmetry1}(b)]: The two bits in each orange oval are merged into one bit. Because each long terminal is merged with a short terminal, the resulting plain Tanner graph has bit-check symmetry, i.e. the subgraph in the grey dashed box is symmetric. The subgraph is symmetric because each bit has a dual vertex, and the dual vertex of each merged bit is the dual vertex of the short terminal. 

If the phase gate is applied first then the Hadamard gate, Fig. \ref{fig:circuit_symmetry2}(a) illustrates the plain Tanner graph of the circuit. In this case, two long terminals are merged, and two short terminals are merged, as shown in Fig. \ref{fig:circuit_symmetry2}(b). Then, the resulting plain Tanner graph does not have bit-check symmetry. To recover the symmetry, we can add one bit $u$ and one check $a$ to the edge connecting $v$ and $b$. We call such an operation bit splitting; see Definition \ref{def:bit_splitting}. Bit splitting preserves the code and circuit code distance; see Lemma \ref{lem:bit_splitting} and Appendix \ref{app:bit_splitting} for its proof. After the bit splitting, the subgraph in the grey dashed box is symmetric. 

\begin{figure}[tbp]
\centering
\includegraphics[width=\linewidth]{\figpath/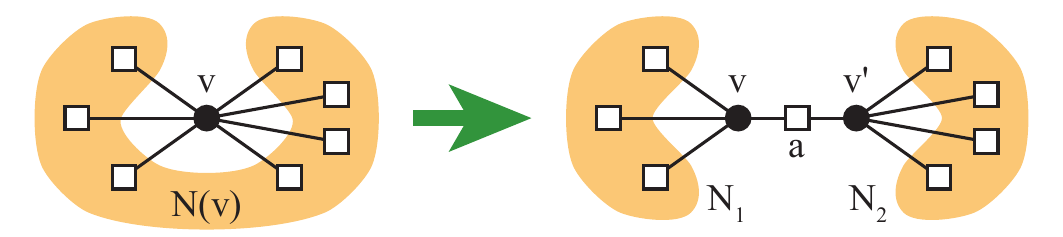}
\caption{
Bit splitting on $v$. In general, one of $N_1$ or $N_2$ could be empty. 
}
\label{fig:bit_splitting}
\end{figure}

\begin{definition}
Let $N(v)$ be the neighbourhood of a bit $v$. Let $N_1$ and $N_2$ be two disjoint subsets covering $N(v)$. A bit splitting on $v$ with subsets $\{N_1,N_2\}$ is the following operation on the Tanner graph (see Fig. \ref{fig:bit_splitting}): Disconnect $v$ from checks in $N_2$, add a bit $v'$, connect $v'$ to checks in $N_2$, add a check $a$, and connect $a$ to $v$ and $v'$. 
\label{def:bit_splitting}
\end{definition}

\begin{lemma}
Let $\mathbf{A}$ and $\mathbf{A}'$ be check matrices, and the Tanner graph of $\mathbf{A}'$ is generated by applying bit splitting on the Tanner graph of $\mathbf{A}$. Suppose the error correction check matrix $\mathbf{B}$ and logical generator matrix $\mathbf{L}$ are compatible with $\mathbf{A}$. Then, the following statements hold: 
\begin{itemize}
\item[i)] Two codes are isomorphic, i.e. there exists a linear bijection $\phi:\mathrm{ker}\mathbf{A}\rightarrow\mathrm{ker}\mathbf{A}'$; 
\item[ii)] The error correction check matrix $\mathbf{B}' = \phi(\mathbf{B})$ and logical generator matrix $\mathbf{L}' = \phi(\mathbf{L})$ are compatible with $\mathbf{A}'$; 
\item[iii)] Their circuit code distances are the same, i.e. 
\begin{eqnarray}
d(\mathbf{A},\mathbf{B},\mathbf{L}) = d(\mathbf{A}',\mathbf{B}',\mathbf{L}').
\end{eqnarray}
\end{itemize}
\label{lem:bit_splitting}
\end{lemma}

The above approach of recovering bit-check symmetry is universal. Tanner graphs of primitive operations have the following properties: For each qubit, there are two input (output) bits; one of them is a long terminal, and the other is a short terminal. When two primitive Tanner graphs are merged, a pair of input bits are merged with a pair of output bits. Therefore, there are only two cases. The first case is the symmetric merging exemplified by Fig. \ref{fig:circuit_symmetry1}, in which each long terminal is merged with a short terminal. The second case is the asymmetric merging exemplified by Fig. \ref{fig:circuit_symmetry2}, in which two long (short) terminals are merged. We can always recover the symmetry by applying bit splitting in the asymmetric merging. 

Here is an alternative way of understanding bit-check symmetry. In the case of asymmetric merging, we can apply bit splitting on one of the two short terminals before merging. The bit splitting is applied on a primitive Tanner graph, which generates a new primitive Tanner graph. Because of the properties of bit splitting, the new one and the original one are equivalent. On the new primitive Tanner graph, the original short (long) terminal becomes the long terminal. Accordingly, the asymmetric merging becomes symmetric mergying. 

\begin{proposition}
For all stabiliser circuits, their plain Tanner graphs can be transformed into Tanner graphs with bit-check symmetry through bit splitting operations. 
\label{prop:symmetry}
\end{proposition}

\begin{figure}[tbp]
\centering
\includegraphics[width=\linewidth]{\figpath/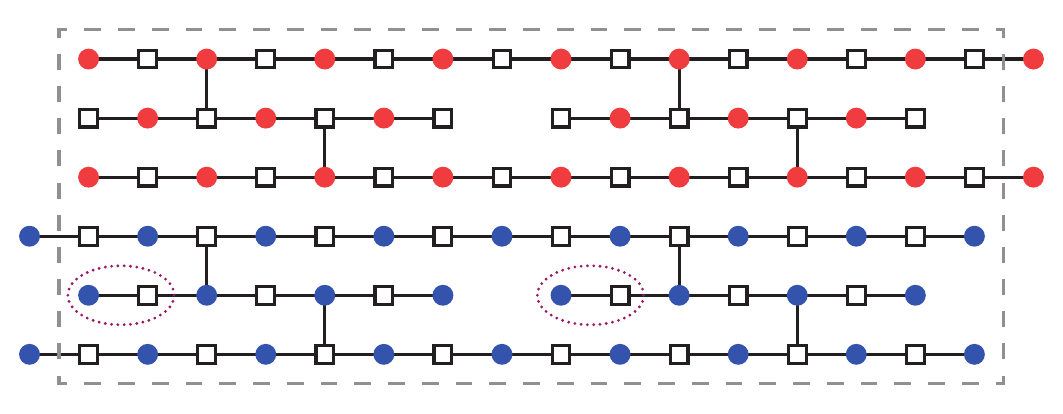}
\caption{
Symmetric Tanner graph of the circuit in Fig. \ref{fig:ZZ}(a). Before bit splitting (without vertices in purple dotted ovals), it is the same as the graph in Fig. \ref{fig:ZZ}(b) (see Step 3). After bit splitting, the subgraph in the grey dashed box is symmetric. 
}
\label{fig:ZZ_symmetry}
\end{figure}

In what follows, symmetric Tanner graph of a stabiliser circuit refers to the Tanner graph with bit-check symmetry generated from the plain Tanner graph through bit splitting, and {\it symmetric subgraph} refers to the subgraph generated by deleting long terminals from the symmetric Tanner graph. In Fig. \ref{fig:ZZ_symmetry}, we draw the symmetric Tanner graph of the repeated $Z_1Z_2$ measurement as an example. 

\section{Symmetric splitting}
\label{sec:splitting}

In this section and the next section, we develop tools for constructing fault-tolerant circuits from LDPC codes. We focus on stabiliser circuits consisting of the controlled-NOT gate and single-qubit operations. For such circuits, Tanner graphs have a maximum degree of three. However, a general LDPC code may have larger vertex degrees. Therefore, we may need to reduce vertex degrees to convert it into a stabiliser circuit. 

In this section, we develop a method that can transform an input symmetric Tanner graph to an output symmetric Tanner graph with reduced vertex degrees. The transformation must have the following properties: 
\begin{itemize}
\item[i)] The bit-check symmetry is preserved; 
\item[ii)] Codewords are preserved, i.e. two codes are isomorphic; 
\item[iii)] The circuit code distance may be reduced by a factor, but the factor is bounded. 
\end{itemize}
Then, if the input Tanner graph has a good circuit code distance, we can obtain an output Tanner graph with a good circuit code distance. Symmetric splitting is the transformation satisfying all the requirements. We will show that in symmetric splitting, the circuit code distance may be reduced by a factor, and the factor is bounded by the maximum vertex degree of the input Tanner graph. 

\subsection{Operations in symmetric splitting}

\RED{\begin{figure}[tbp]
\centering
\includegraphics[width=\linewidth]{\figpath/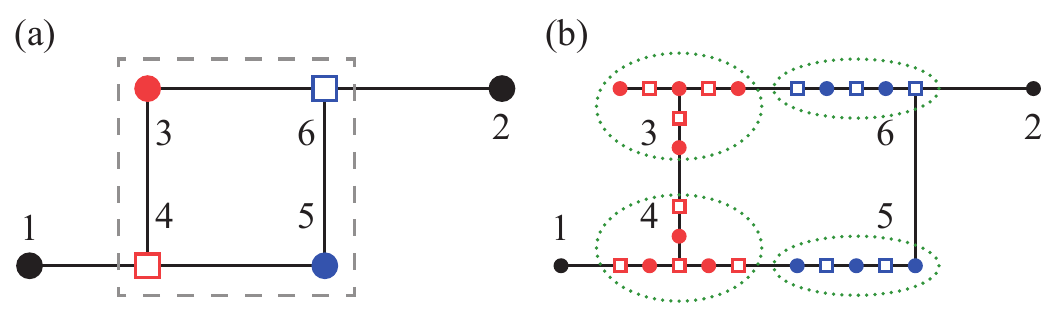}
\caption{
(a) The input Tanner graph. The two red (blue) vertices are dual vertices of each other. (b) The output Tanner graph. Each vertex on the subgraph in the dashed grey box is replaced with a tree (in the green dotted circles). 
}
\label{fig:symmetric_splitting0}
\end{figure}

We consider an input Tanner graph with bit-check symmetry; see Fig. \ref{fig:symmetric_splitting0}(a) for example. Long terminals of the graph are vertices $1$ and $2$. The symmetric subgraph is in the grey dashed box. On the subgraph, each vertex has a dual vertex: vertices $3$ and $4$ ($5$ and $6$) are dual vertices of each other. If we exchange dual vertices, the subgraph is unchanged. 

\begin{figure}[tbp]
\centering
\includegraphics[width=\linewidth]{\figpath/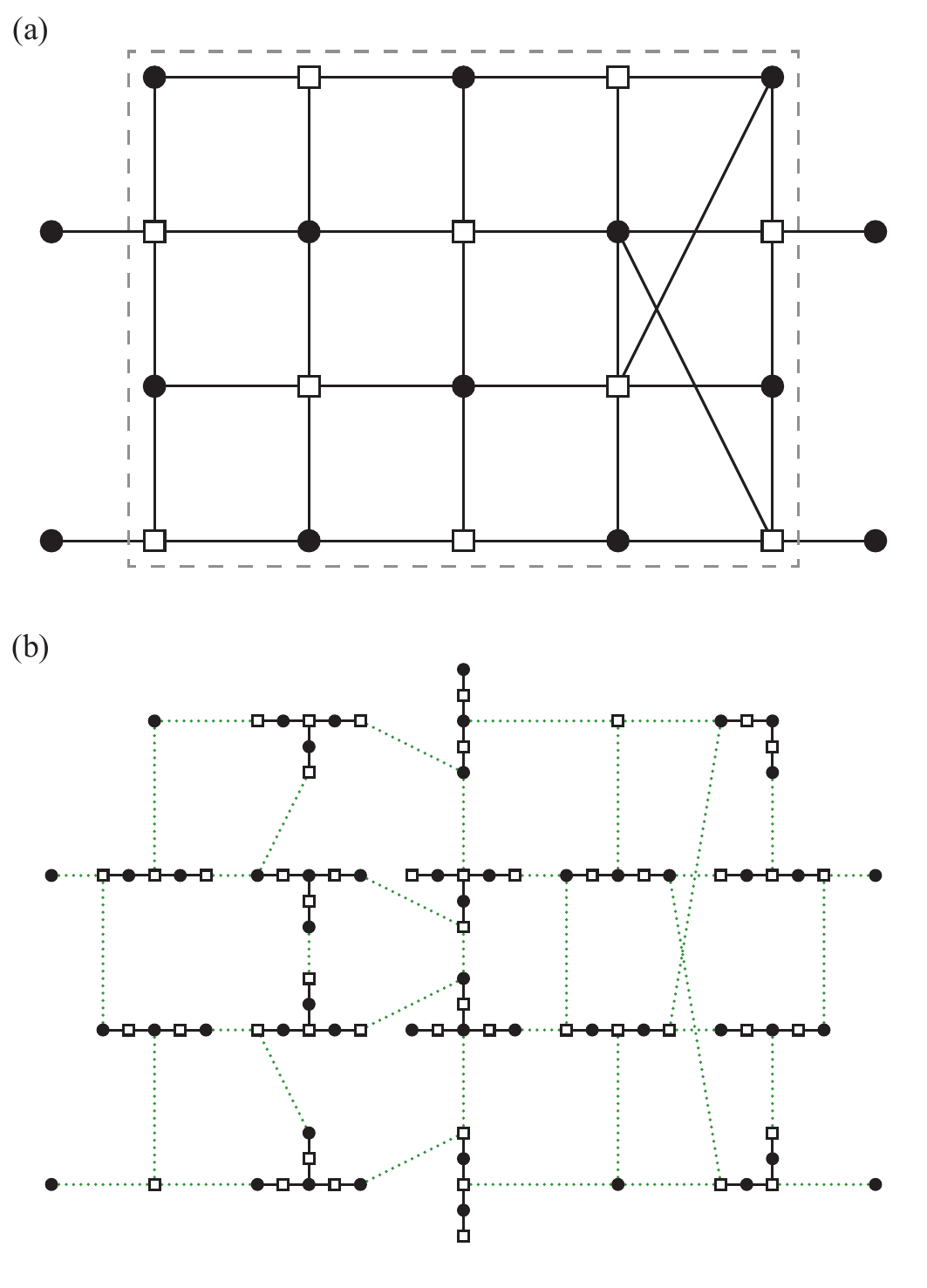}
\caption{
(a) The input Tanner graph. On the subgraph in the grey dashed box, each bit has a dual check. (b) The output Tanner graph. Each vertex on the subgraph is replaced with a tree. Intra-tree edges are denoted by black solid lines. Inter-tree edges and edges on long terminals are denoted by green dotted lines. 
}
\label{fig:symmetric_splitting1}
\end{figure}

In the symmetric splitting, we replace each vertex on the symmetric subgraph with a tree, as shown in Fig. \ref{fig:symmetric_splitting0}(b). Then, the trees are connected according to the original connections between vertices. We choose the trees and the way of connecting them such that the output graph is still symmetric. Specifically, for each pair of dual vertices, their trees are the same, but the roles of bit and check vertices are exchanged. On a {\it bit tree}, which replaces a bit vertex ($3$ or $5$), all check vertices have a degree of two; on a {\it check tree}, which replaces a check vertex ($4$ or $6$), all bit vertices have a degree of two. When connecting trees, we connect bit vertices on bit trees with check vertices on check trees, and check (bit) vertices on bit (check) trees do not participate connecting trees. It is similar when connecting a check tree with a long terminal. An example of a larger graph is shown in Fig. \ref{fig:symmetric_splitting1}. }

\begin{figure}[tbp]
\centering
\includegraphics[width=\linewidth]{\figpath/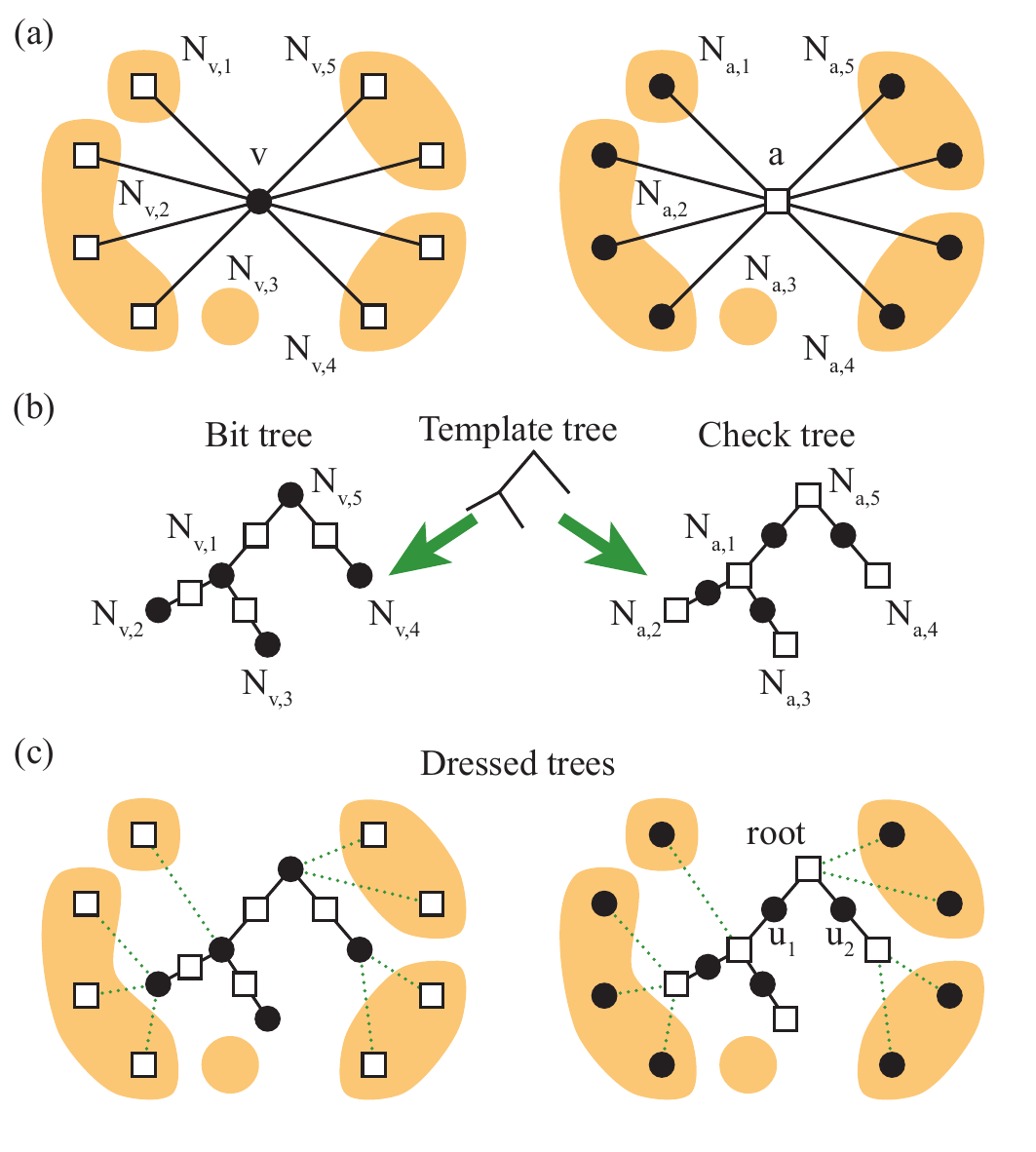}
\caption{
(a) Neighbourhood subsets for a pair of dual vertices. (b) Bit and check trees generated from the same template tree. (c) Dressed bit and check trees with connections (denoted by green dotted lines) to adjacent trees and long terminals. 
}
\label{fig:symmetric_splitting2}
\end{figure}

Symmetric splitting consists of the following operations: 
\begin{itemize}
\item[1.] On the input graph, divide the neighbourhood of each vertex into subsets in a symmetric way. Let $v$ and $a$ be dual vertices of each other. Neighbourhoods $N(v)$ and $N(a)$ are divided into $r_v$ disjoint subsets $N_{v,1},N_{v,2},\ldots,N_{v,r_v}$ and $N_{a,1},N_{a,2},\ldots,N_{a,r_v}$, respectively. Here, $N_{v,i}$ and $N_{a,i}$ contain dual vertices of each other. See Fig. \ref{fig:symmetric_splitting2}(a). 
\item[2.] For each pair of dual vertices on the input graph, compose a template tree with $r_v$ vertices. Generate two trees from the template tree; see Fig. \ref{fig:symmetric_splitting2}(b). To generate the bit (check) tree, a bit (check) is placed on each vertex of the template tree, and a check (bit) is placed on each edge of the template tree. Then, we have a forest, as shown in Fig. \ref{fig:symmetric_splitting1}(b). 
\item[3.] On the bit (check) tree, associate each bit (check) to a neighbourhood subset of $v$ ($a$), as shown in Fig. \ref{fig:symmetric_splitting2}(b). Let’s use $\hat{N}_{v,i}$ ($\hat{N}_{a,i}$) as labels of bits (checks) on the tree. To preserve the symmetry, $\hat{N}_{v,i}$ and $\hat{N}_{a,i}$ correspond to the same vertex on the template tree. 
\item[4.] For each edge $\{v,b\}$ on the symmetric subgraph, add an edge to the forest that connects $\hat{N}_{v,i}$ with $\hat{N}_{b,j}$. Here, $i$ and $j$ are chosen such that $b\in N_{v,i}$ and $v\in N_{b,j}$. See Figs. \ref{fig:symmetric_splitting1}(b) and \ref{fig:symmetric_splitting2}(c). 
\item[5.] For each edge $\{w,b\}$ incident on a long terminal $w$, add a bit $w$ to the forest, and add an edge that connects $w$ with $\hat{N}_{b,j}$. Here, $j$ is chosen such that $w\in N_{b,j}$. See Figs. \ref{fig:symmetric_splitting1}(b) and \ref{fig:symmetric_splitting2}(c). 
\end{itemize}

\subsection{Preservation of the bit-check symmetry and codewords}
\label{sec:preservation}

Symmetric splitting preserves bit-check symmetry. The forest generated in step-3 is symmetric. In step-4, edges are added to the forest in a symmetric way: For each edge $\{\hat{N}_{v,i},\hat{N}_{b,j}\}$ added to the forest, there is a dual edge $\{\hat{N}_{a,i},\hat{N}_{u,j}\}$ added to the forest. Here, $a$ and $u$ are dual vertices of $v$ and $b$, respectively. 

Let $\mathbf{A}$ and $\mathbf{A}'$ be check matrices of the input and output Tanner graphs. There is a one-to-one correspondence between codewords of $\mathbf{A}$ and $\mathbf{A}'$. 

On a bit tree, bits are connected through degree-2 checks, and they all take the same value in a codeword. Therefore, a bit tree plays the role of a single bit. 

Similarly, a check tree plays the role of a single check. When trees are connected with each other as shown in Fig. \ref{fig:symmetric_splitting1}(b), a check tree is connected with some bits, which are either on adjacent bit trees or long terminals. With these bits added to the check tree, we call it the {\it dressed check tree}. The added bits are leaves of the dressed tree; see Fig. \ref{fig:symmetric_splitting2}(c). Now, let’s consider the values of bits in a codeword. Take any check as the root of the dressed tree. We can find that the value of a bit is the sum of all its descendant bit leaves. For example, $u_1$ ($u_2$) takes the sum of all leaves connected with $\hat{N}_{a,1}$, $\hat{N}_{a,2}$ and $\hat{N}_{a,3}$ ($\hat{N}_{a,4}$ and $\hat{N}_{a,5}$). To satisfy the root check, the sum of all bit leaves must be zero. 

According to the above analysis, there is a linear bijection $\psi:\mathrm{ker}\mathbf{A}\rightarrow\mathrm{ker}\mathbf{A}'$ between two codes. Let $\mathbf{c}\in \mathrm{ker}\mathbf{A}$. The map is $\psi(\mathbf{c})_w = \mathbf{c}_w$ if $w$ is a long terminal, $\psi(\mathbf{c})_u = \mathbf{c}_v$ if $u$ is on the bit tree of $v$, and $\psi(\mathbf{c})_u = \sum_{v\in D_{a,u}}\psi(\mathbf{c})_v$ if $u$ is on the check tree of $a$. Here, $D_{a,u}$ denotes descendant bit leaves of $u$ on the dressed tree of $a$, which are either long terminals or on bit trees. 

\subsection{Impact on the circuit code distance}

To analyse the impact on the circuit code distance, we consider mapping errors on the input Tanner graph to errors on the output Tanner graph. For the input and output Tanner graphs, the error spaces are $\mathbb{F}_2^{\vert V_B\vert}$ and $\mathbb{F}_2^{\vert V_B'\vert}$, respectively. Here, $V_B$ and $V_B'$ are bit sets of the input and output graphs, respectively. There exists a linear injection $\psi_{err}:\mathbb{F}_2^{\vert V_B\vert}\rightarrow \mathbb{F}_2^{\vert V_B'\vert}$ with the following properties: i) For all codewords $\mathbf{c} = \mathrm{ker}\mathbf{A}$ and errors $\mathbf{e}\in \mathbb{F}_2^{\vert V_B\vert}$, $\mathbf{c}\mathbf{e}^\mathrm{T} = \psi(\mathbf{c})\psi_{err}(\mathbf{e})^\mathrm{T}$; and ii) $\vert \mathbf{e}\vert = \vert \psi_{err}(\mathbf{e})\vert$. The map is $\psi_{err}(\mathbf{e})_v = \mathbf{e}_v$ if $v$ is a long terminal, $\psi_{err}(\mathbf{e})_{\hat{N}_{v,1}} = \mathbf{e}_v$ if $v\in V_B$ is not a long terminal, and $\psi_{err}(\mathbf{e})_u = 0$ for all other bits $u\in V_B'$ on the output graph. In this map, errors on the output graph are on either long terminals or $\hat{N}_{v,1}$ vertices on bit trees, and check trees are clean. 

Let’s consider a logical error $\mathbf{e}\in \mathbb{F}_2^{\vert V_B\vert}$ with the minimum weight $\vert \mathbf{e}\vert = d(\mathbf{A},\mathbf{B},\mathbf{C})$. Then, $\psi_{err}(\mathbf{e})$ is also a logical error and has the same weight. Therefore, symmetric splitting never increases the circuit code distance. 

The circuit code distance is reduced if we can find an error in $\mathbb{F}_2^{\vert V_B'\vert}$ that is equivalent to $\psi_{err}(\mathbf{e})$ but has a smaller weight. The weight is not reduced by bit trees. On a bit tree, bits are connected through degree-2 checks. Therefore, on the same bit three, all single-bit errors are equivalent (according to Proposition \ref{prop:equivalence}), and they are all equivalent to the single-bit error on $\hat{N}_{v,1}$. In the equivalence, the weight is preserved. 

The weight may be reduced by check trees. Take the dressed check tree in Fig. \ref{fig:symmetric_splitting2}(c) as an example. Suppose there is a single-bit error on each leaf connected with $\hat{N}_{a,4}$ and on each leaf connected with $\hat{N}_{a,5}$. This multi-bit error on leaves is equivalent to a single-bit error on $u_2$. Therefore, the weight is reduced from four to one. 

In the case related to the code distance, there are at most $\lfloor g_a/2\rfloor$ single-bit errors on leaves. Here, $g_a$ is the vertex degree of $a$ on the input graph. If there are more than $\lfloor g_a/2\rfloor$ single-bit errors, this multi-bit error is equivalent to a multi-bit error on the complementary leaf set, whose weight is smaller than $\lfloor g_a/2\rfloor$. These $\lfloor g_a/2\rfloor$ single-bit errors may be equivalent to one single-bit error due to the previous analysis, i.e. the weight is reduced by a factor of $1/\lfloor g_a/2\rfloor$. Let $g_{max}$ be the maximum vertex degree of the input graph, symmetric splitting may reduce the circuit code distance by a factor of $1/\lfloor g_{max}/2\rfloor$. 

\begin{theorem}
Let $\mathbf{A}$ and $\mathbf{A}'$ be check matrices with bit-check symmetry, and the Tanner graph of $\mathbf{A}'$ is generated by applying symmetric splitting on the Tanner graph of $\mathbf{A}$. Suppose the error correction check matrix $\mathbf{B}$ and logical generator matrix $\mathbf{L}$ are compatible with $\mathbf{A}$. Then, the following statements hold: 
\begin{itemize}
\item[i)] Two codes are isomorphic, i.e. there exists a linear bijection $\psi:\mathrm{ker}\mathbf{A}\rightarrow\mathrm{ker}\mathbf{A}'$; 
\item[ii)] The error correction check matrix $\mathbf{B}' = \psi(\mathbf{B})$ and logical generator matrix $\mathbf{L}' = \psi(\mathbf{L})$ are compatible with $\mathbf{A}'$; 
\item[iii)] The circuit code distance may be reduced by a bounded factor, i.e. 
\begin{eqnarray}
d(\mathbf{A}',\mathbf{B}',\mathbf{L}') \geq \frac{d(\mathbf{A},\mathbf{B},\mathbf{L})}{\lfloor g_{max}/2\rfloor}.
\end{eqnarray}
\end{itemize}
\label{the:symmetric_splitting}
\end{theorem}

A formal proof can be found in Append \ref{app:symmetric_splitting}. 

\begin{figure}[tbp]
\centering
\includegraphics[width=\linewidth]{\figpath/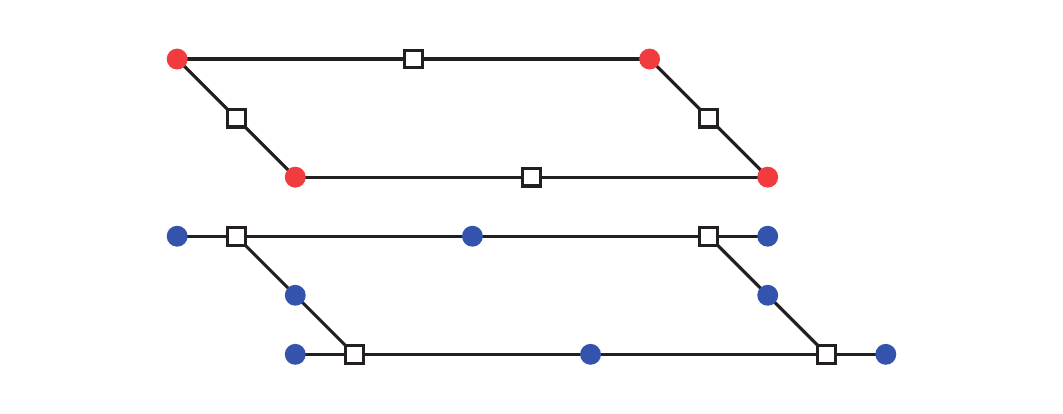}
\caption{
A simplified symmetric Tanner graph. The graph can be transformed to the plain Tanner graph in Fig. \ref{fig:ZZ}(b) through symmetric splitting, bit splitting and the inverse operation of bit splitting. 
}
\label{fig:ZZsimpllification}
\end{figure}

\RED{We remark that, in addition to reducing the vertex degree, we can also simplify the Tanner graph representation of a quantum circuit through bit splitting and symmetric splitting. For example, the Tanner graph in Fig. \ref{fig:ZZ}(b) can be simplified to that shown in Fig. \ref{fig:ZZsimpllification}. }

\section{Construction of fault-tolerant circuits from LDPC codes}
\label{sec:construction}

Given an arbitrary Tanner graph with bit-check symmetry, we can convert it into a stabiliser circuit. In this section, we show that the above statement holds in general by presenting a universal approach for converting Tanner graphs into circuits. Notice that there are usually many different circuits corresponding to the same Tanner graph. Therefore, there is a large room for optimisation, in which symmetric splitting would be a useful tool. In this work, we focus on the existence of such stabiliser circuits and leave the optimisation for future work. 

\subsection{Path partition}

\begin{figure}[tbp]
\centering
\includegraphics[width=\linewidth]{\figpath/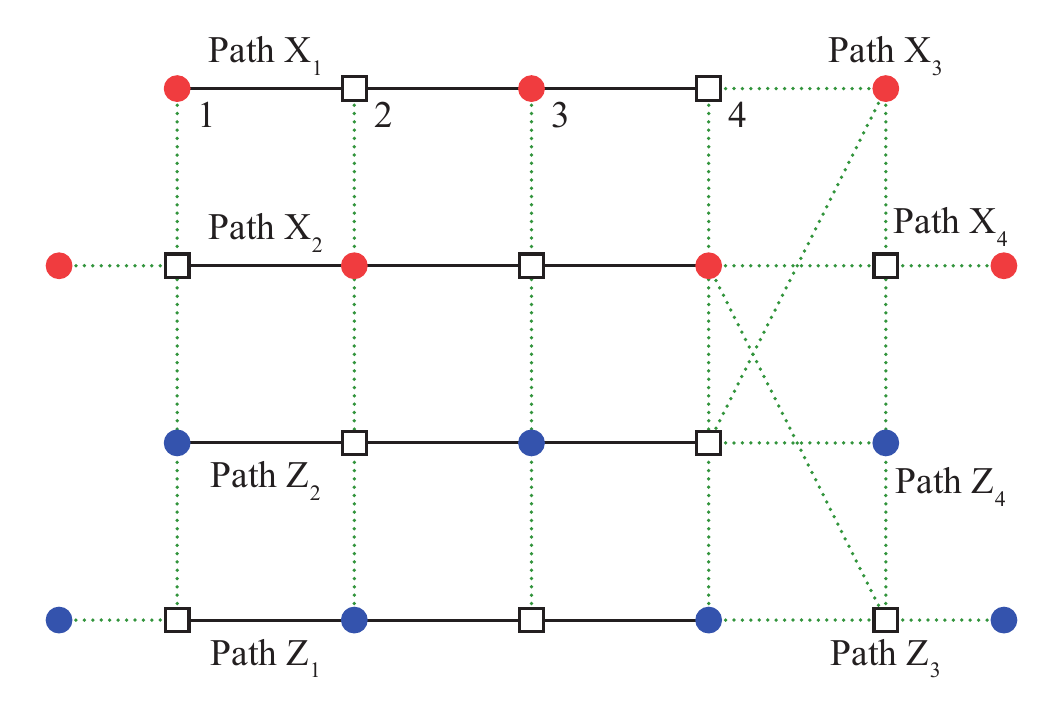}
\caption{
A path partition of the Tanner graph in Fig. \ref{fig:symmetric_splitting1}(a). Intra-path edges are denoted by black solid lines. Inter-path edges and edges on long terminals are denoted by green dotted lines. The numbers are time labels on the path $X_1$. Except for long terminals, vertices connected through green dotted lines have the same time label. 
}
\label{fig:circuit_construction1}
\end{figure}

In the universal approach, the first step is identifying qubits. This is achieved by finding a path partition of its symmetric subgraph. There are three conditions on the path partition: i) The partition must be symmetric, i.e. each path has a dual path, and the two paths are exchanged when exchanging each pair of dual vertices; ii) long terminals are connected with the ends of paths; and the third condition will be introduced later. See Fig. \ref{fig:circuit_construction1} for an example. 

Each pair of dual paths represents a qubit, and each path represents a Pauli operator $X_q$ or $Z_q$. Here, $q = 1,2,\ldots$ is the label of qubits. If a path represents $X_q$, its dual path represents $Z_q$, and vice versa. 

The second step is identifying time layers. Given the path partition, assign a time label $\tau = 1,2,\ldots$ to each vertex on the path in the following way: First, dual vertices have the same time label; and second, the time label increases strictly along a direction of the path. Each time label corresponds to $\Delta t$ layers in the stabiliser circuit. 

Now, we have the third condition on the path partition. iii) For each edge connecting two paths, it is always incident on two vertices with the same time label: Let $N_S(v)$ and $N_P(v)$ be neighbourhoods of $v$ on the symmetric subgraph and the corresponding path, respectively; then, $v$ and all vertices in $N_S(v)-N_P(v)$ have the same time label $\tau$. There always exists a path partition satisfying the three conditions: Each path has only one vertex (let’s call it a path), and $\tau = 1$ for all vertices on the symmetric subgraph. 

\subsection{Circuit construction}

\begin{figure}[tbp]
\centering
\includegraphics[width=\linewidth]{\figpath/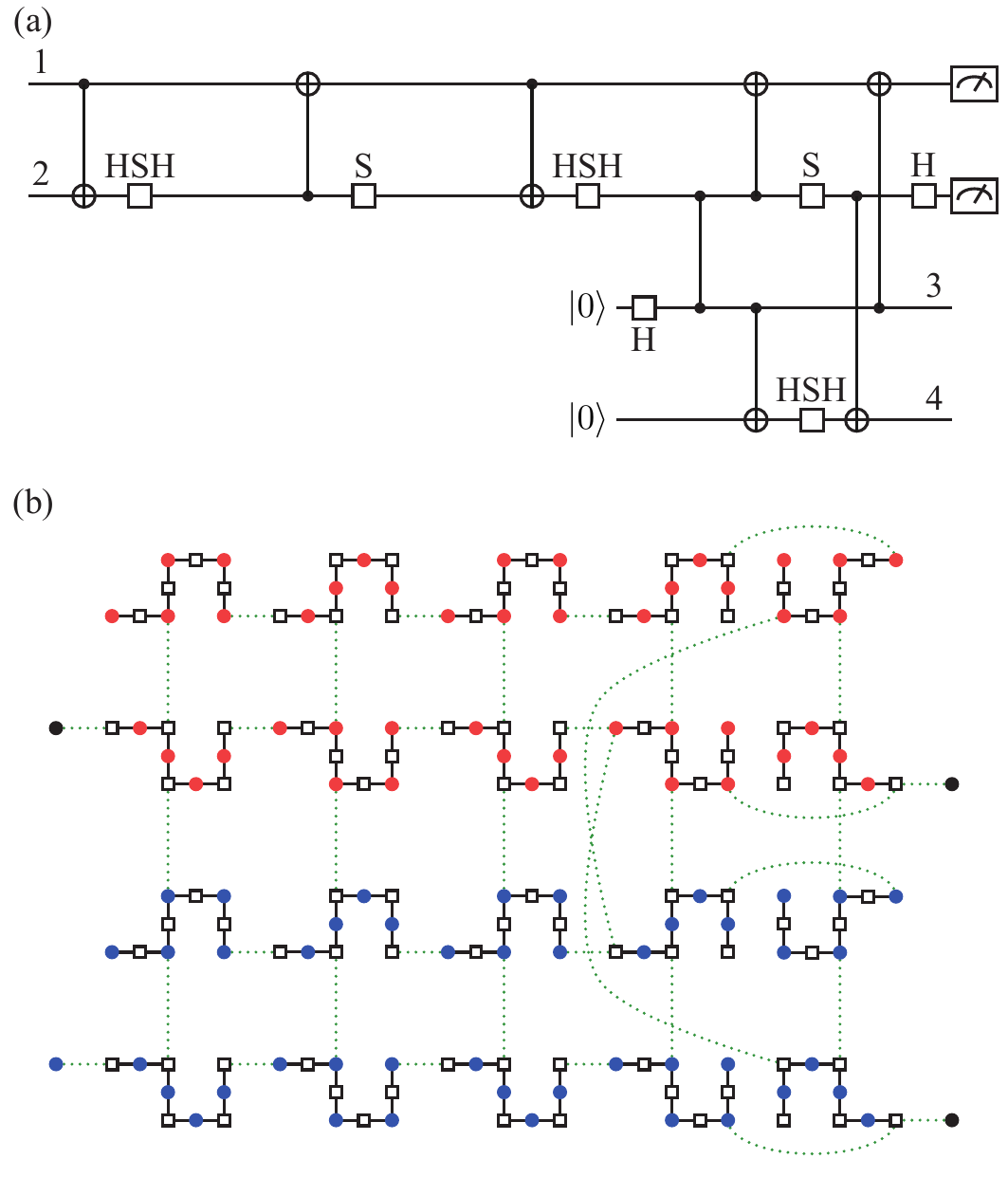}
\caption{
(a) Circuit constructed according to the path partition in Fig. \ref{fig:circuit_construction1}. (b) Symmetric Tanner graph of the circuit in (a). We can obtain this Tanner graph via symmetric splitting on the graph in Fig. \ref{fig:circuit_construction1}. In the symmetric splitting, bit and check trees are microscopic paths. Intra-tree edges are denoted by black solid lines. Inter-tree edges and edges on long terminals are denoted by green dotted lines. 
}
\label{fig:circuit_construction2}
\end{figure}

Now, we can draw the stabiliser circuit. See Fig. \ref{fig:circuit_construction2}(a) for an example. First, add initialisation and measurement operations on qubits. For each qubit $q$, we need to look at the four ends of paths $X_q$ and $Z_q$. Let $\tau_{min}$ and $\tau_{max}$ be the minimum and maximum time labels on the two paths. In the four ends, two of them have the time label $\tau_{min}$, and two of them have the time label $\tau_{max}$. In the two $\tau_{min}$ ($\tau_{max}$) ends, one of them is a bit, and the other is a check. If the $\tau_{min}$ check is not connected with a long terminal, the qubit is initialised at $t = (\tau_{min}-1)\Delta t$. The initialisation basis depends on the $\tau_{min}$ bit. If the bit is on the $Z_q$ ($X_q$) path, the basis is $Z$ ($X$). It is similar for the two $\tau_{max}$ ends. If the $\tau_{max}$ check is not connected with a long terminal, the qubit is measured at $t = \tau_{max}\Delta t + 1$. If the $\tau_{max}$ bit is on the $Z_q$ ($X_q$) path, the measurement basis is $Z$ ($X$). Gates on the qubit are applied from $t = (\tau_{min}-1)\Delta t + 1$ to $t = \tau_{max}\Delta t$. 

\begin{figure}
\begin{minipage}{\linewidth}
\begin{table}[H]
\begin{tabular}{|c|c|c|}
\hline
Path of $v$ & Path of $a$ & Gate \\ \hline\hline
$X_q$ & $Z_q$ & $S_q$ \\ \hline
$Z_q$ & $X_q$ & $H_q S_q H_q$ \\ \hline
$X_q$ & $X_{q'}$ & \multirow{2}{*}{$\Lambda_{q,q'}$} \\ \cline{1-2}
$Z_{q'}$ & $Z_q$ & \\ \hline
$X_q$ & $Z_{q'}$ & \multirow{2}{*}{$H_{q'}\Lambda_{q,q'}H_{q'}$} \\ \cline{1-2}
$X_{q'}$ & $Z_q$ & \\ \hline
$Z_q$ & $X_{q'}$ & \multirow{2}{*}{$H_q\Lambda_{q,q'}H_q$} \\ \cline{1-2}
$Z_{q'}$ & $X_q$ & \\ \hline
\end{tabular}
\caption{
Gates corresponding to inter-path edges. Each edge is incident on a bit $v$ and a check $a$. If the edge connects two paths of different qubits $q$ and $q'$, there is another edge (on dual vertices of $v$ and $a$) connecting the other two paths of qubits $q$ and $q'$ because of the symmetric. The two edges correspond to the same gate in the circuit. In the table, $\Lambda_{q,q'}$ denotes the controlled-NOT gate on the control qubit $q$ and target qubit $q'$, and $H_q$ and $S_q$ denote the Hadamard and phase gate on qubit-$q$, respectively. 
}
\label{table:gates}
\end{table}
\end{minipage}
\end{figure}

Next, add a gate to the circuit for each inter-path edge. If the edge connects paths $X_q$ and $Z_q$ of the same qubit, it corresponds to a single-qubit gate. If the edge connects paths of different qubits, it corresponds to a two-qubit gate. The gates are listed in Table \ref{table:gates}. When two vertices of the edge have the time label $\tau$, the corresponding gate is applied in the time window from $t = (\tau-1)\Delta t + 1$ to $t = \tau\Delta t$. 

Each time label corresponds to a time window consisting of $\Delta t$ time steps. Each time window may include multiple gates (Gates in the same time window are always commutative). Some of them are applied on the same qubit, so they must be arranged properly to avoid conflict. Therefore, $\Delta t$ must be sufficiently large to implement all the gates in the same time window. Notice that $\Delta t = 1 + n(n-1)/2$ is sufficiently large because we can implement single-qubit gates at $t = (\tau-1)\Delta t + 1$ and then two-qubit gates without conflict. 

\begin{theorem}
All Tanner graphs with bit-check symmetry can be converted into stabiliser circuits. Given an arbitrary Tanner graph with bit-check symmetry, there exists a stabiliser circuit such that its symmetric Tanner graph is related to the given Tanner graph through symmetric splitting. 
\label{the:circuit_construction}
\end{theorem}

We already have a universal approach for constructing the stabiliser circuit according to a given Tanner graph. To prove the theorem, we only need to examine the symmetric Tanner graph of the constructed circuit. See Appendix \ref{app:circuit_construction1} for the proof. An example of the symmetric Tanner graph is given in Fig. \ref{fig:circuit_construction2}(b). This graph can be generated through symmetric splitting, in which bit and check trees are paths (we call them {\it microscopic paths} to be distinguished from paths in the path partition). 

\begin{corollary}
If a Tanner graph with bit-check symmetry has a favourable circuit code distance, a fault-tolerant circuit can be constructed according to the Tanner graph. Let $d$ be the circuit code distance and $g_{max}$ be the maximum vertex degree of the Tanner graph, and let $d'$ be the circuit code distance of the constructed circuit. Then, 
\begin{eqnarray}
d' \geq \frac{d}{\lfloor g_{max}/2\rfloor}.
\end{eqnarray}
\label{coro:circuit_construction}
\end{corollary}

The above corollary holds as a consequence of Proposition \ref{prop:symmetry}, Lemma \ref{lem:bit_splitting}, Theorem \ref{the:symmetric_splitting} and Theorem \ref{the:circuit_construction}. Although it holds for all Tanner graphs with bit-check symmetry, one looks for LDPC codes when constructing fault-tolerant circuits because of the distance reduction due to the maximum vertex degree. 

\RED{\subsection{Verification of logical operations}
\label{sec:verification}

In addition to fault tolerance, another goal when selecting linear codes to construct quantum circuits is to achieve certain operations on logical qubits. Here, we will explain how to select linear codes to meet this requirement, or rather, how to determine if a linear code can generate a proper circuit to implement the desired logical operations. Note that to construct a fault-tolerant quantum circuit, we need a linear code and a decomposition of the code, which is represented by a 3-tuple $(\mathbf{A},\mathbf{B},\mathbf{L})$. 

Given a logical operation, or more generally, a logical circuit, we can represent it using a linear code, just like a physical circuit. Let $\mathbf{a}$ be the check matrix of the logical circuit. The correlations of logical qubits in the logical circuit are described by codewords $\bar{\mathbf{c}}\in \mathrm{ker}\,\mathbf{a}$. To implement the same correlations in the physical circuit, we need $\mathrm{rowsp}(\mathbf{L})$ and $\mathrm{ker}\,\mathbf{a}$ to be isomorphic. More specifically, the physical circuit should implement the same transformations on logical Pauli operators as the logical circuit. Below, we explain this in detail. 

For a codeword $\bar{\mathbf{c}}\in \mathrm{ker}\,\mathbf{a}$, projections similar to $\mathbf{P}_0$ and $\mathbf{P}_T$ will yield vectors $\bar{\mathbf{c}}_{in}$ and $\bar{\mathbf{c}}_{out}$, which represent the input and output logical operators, respectively. The input and output logical operators are denoted as $\bar{\sigma}(\bar{\mathbf{c}}_{in})$ and $\bar{\sigma}(\bar{\mathbf{c}}_{out})$, where $\bar{\sigma}$ indicates that the operators are generated by logical operators. If the physical circuit implements the same transformations on logical Pauli operators, there should be a codeword $\mathbf{c}\in \mathrm{rowsp}(\mathbf{L})$ such that $\sigma(\mathbf{c}_{in}) = \bar{\sigma}(\bar{\mathbf{c}}_{in})$ and $\sigma(\mathbf{c}_{out}) = \bar{\sigma}(\bar{\mathbf{c}}_{out})$, where $\mathbf{c}_{in}$ and $\mathbf{c}_{out}$ are the vectors obtained from $\mathbf{c}$ through projections similar to $\mathbf{P}_0$ and $\mathbf{P}_T$ (we will discuss these two projections in detail later). In fact, we do not need to verify the above for each codeword $\bar{\mathbf{c}}\in \mathrm{ker}\,\mathbf{a}$ but only need to consider the generator matrix $\mathbf{g}$ of the code. After projection, we can obtain two matrices $\mathbf{g}_{in}$ and $\mathbf{g}_{out}$. Similarly, after projection, the matrix $\mathbf{L}$ yields two matrices $\mathbf{L}_{in}$ and $\mathbf{L}_{out}$. If the physical circuit implements the same transformations on logical Pauli operators, there should be full rank matrices $\mathbf{J}_{in}$ and $\mathbf{J}_{out}$ such that $\mathbf{L}_{in} = \mathbf{g}_{in}\mathbf{J}_{in}$ and $\mathbf{L}_{out} = \mathbf{g}_{out}\mathbf{J}_{out}$. Here, $\sigma(\mathbf{J}_{in})$ and $\sigma(\mathbf{J}_{out})$ are generators of the input and output logical Pauli groups, respectively. 

Finally, we discuss how to project the codewords $\mathbf{c}\in \mathrm{rowsp}(\mathbf{L})$. Unlike the plain Tanner graph, for a general symmetric Tanner graph, we need to assign time coordinates to bit vertices before projection. In the process of constructing quantum circuits from linear codes, we need to perform path partition first. The two ends of each path correspond to the smallest and largest time labels. According to the requirements of path partition, all long terminals can only appear at the ends of paths, so their time coordinates can be $t = 0,T$. Additionally, suppose a check vertex $a$ is connected to a long terminal, and a bit vertex $u$ is the dual vertex of check vertex $a$, then the time coordinate of $u$ can also be $t = 0,T$. For all other vertices at the ends of paths, which are not connected to long terminals, they correspond to initialisation or measurement, so their time coordinates are not $t = 0,T$. Based on the above discussion, we can find all vertices with time coordinates $t = 0,T$ by following these steps (note that we do not need to determine the path partition in advance): first, find all long terminals $v$, then find all check vertices $a$ connected to long terminals $v$, and finally find all bit vertices $u$ dual to check vertices $a$; the set of long terminals $v$ and bit vertices $u$ constitutes the set of vertices with time coordinates $t = 0,T$. We divide this set into two disjoint subsets, in which a long terminal $v$ and the corresponding bit vertex $u$ are always in the same subset. Then, we can construct the two projections $\mathbf{P}_0$ and $\mathbf{P}_T$ according to the two subsets. 

To verify whether $(\mathbf{A},\mathbf{B},\mathbf{L})$ can implement the logical circuit $\mathbf{a}$, we need to construct appropriate projections $\mathbf{P}_0$ and $\mathbf{P}_T$, and then check if $(\mathbf{L}_{in},\mathbf{L}_{out})$ matches $(\mathbf{g}_{in},\mathbf{g}_{out})$. If they match, we need to construct the circuit in a way consistent with the projections $\mathbf{P}_0$ and $\mathbf{P}_T$, so that we can obtain a physical circuit that implements the logical circuit $\mathbf{a}$. The remaining degrees of freedom in the process of constructing the circuit will only affect details such as time and qubit costs, while all the circuits will implement the same logical circuit. If we detect errors according to $\mathbf{B}$, all the circuits are fault-tolerant. 
}

\section{Example: Transversal stabiliser circuits of Calderbank-Shor-Steane codes}
\label{sec:example}

In this section, we illustrate how to represent fault-tolerant quantum circuits using LDPC codes with an example. First, we show that transversal stabiliser circuits of CSS codes can be represented in a uniform and simple form. Then, we analyse the fault tolerance using algebraic methods. 

For a CSS code, the stabiliser generator matrix is in the form 
\begin{eqnarray}
\left(\begin{matrix}
\mathbf{G}_X \\
0
\end{matrix} \left\vert~ \begin{matrix}
0 \\
\mathbf{G}_Z
\end{matrix}\right.\right), \notag
\end{eqnarray}
where $\mathbf{G}_X\in \mathbb{F}_2^{r_X\times n}$ and $\mathbf{G}_Z\in \mathbb{F}_2^{r_Z\times n}$ are two check matrices satisfying $\mathbf{G}_X\mathbf{G}_Z^\mathrm{T} = 0$, $n$ is the number of qubits, and $r_X$ ($r_Z$) is the number of $X$-operator ($Z$-operator) checks. Let $\mathbf{J}_Z\in \mathbb{F}_2^{k\times n}$ ($\mathbf{J}_X\in \mathbb{F}_2^{k\times n}$) be a generator matrix that completes the kernel of $\mathbf{G}_X$ ($\mathbf{G}_Z$): $\mathbf{G}_X\mathbf{J}_Z^\mathrm{T} = 0$ and $\mathrm{rowsp}(\mathbf{G}_Z)\oplus\mathrm{rowsp}(\mathbf{J}_Z) = \mathrm{ker}\mathbf{G}_X$ ($\mathbf{G}_Z\mathbf{J}_X^\mathrm{T} = 0$ and $\mathrm{rowsp}(\mathbf{G}_X)\oplus\mathrm{rowsp}(\mathbf{J}_X) = \mathrm{ker}\mathbf{G}_Z$). Here, $k = n - \mathrm{rank}\mathbf{G}_X - \mathrm{rank}\mathbf{G}_Z$ is the number of logical qubits. By properly choosing $\mathbf{J}_Z,\mathbf{J}_X$, one has $\mathbf{J}_X\mathbf{J}_Z^\mathrm{T} = \openone_k$. Similar to the stabiliser generator matrix, we can represent logical Pauli operators with the matrix 
\begin{eqnarray}
\left(\begin{matrix}
\mathbf{J}_X \\
0
\end{matrix} \left\vert~ \begin{matrix}
0 \\
\mathbf{J}_Z
\end{matrix}\right.\right), \notag
\end{eqnarray}
Each row of $\mathbf{J}_X$ corresponds to a logical $X$ operator, and the same row of $\mathbf{J}_Z$ corresponds to the logical $Z$ operator of the same logical qubit. 

\subsection{CSS-type logical-qubit circuit}

For general CSS codes, some logical operations are transversal, including the controlled-NOT gate, Pauli gates, initialisation and measurement in the $X$ and $Z$ bases. In a circuit consisting of these operations (Let's call it a {\it CSS-type circuit}), $X$ and $Z$ operators are decoupled. Accordingly, its Tanner graph has two disjoint subgraphs, and the check matrix of the symmetric Tanner graph is in the form 
\begin{eqnarray}
\mathbf{a} = \left(\begin{matrix}
0 & \mathbf{a}_Z \\
\mathbf{a}_X & 0
\end{matrix}\right),
\label{eq:a}
\end{eqnarray}
where $\mathbf{a}_\alpha\in \mathbb{F}_2^{m_\alpha^C\times m_\alpha^B}$ corresponds to the $\alpha$-operator subgraph, $\alpha = X,Z$, and $m_\alpha^B$ ($m_\alpha^C$) is the number of bits (checks) on the $\alpha$-operator subgraph. For example, for the logical controlled-NOT gate, the check matrix takes 
\begin{eqnarray}
\mathbf{a}_X = \left(\begin{matrix}
1 & 0 & 1 & 0 \\
1 & 1 & 0 & 1
\end{matrix}\right)
\end{eqnarray}
and 
\begin{eqnarray}
\mathbf{a}_Z = \left(\begin{matrix}
1 & 1 & 1 & 0 \\
0 & 1 & 0 & 1
\end{matrix}\right).
\end{eqnarray}

When the check matrix $\mathbf{a}$ has bit-check symmetry, there exits a deleting matrix in the form 
\begin{eqnarray}
\mathbf{d} = \left(\begin{matrix}
\mathbf{d}_X & 0 \\
0 & \mathbf{d}_Z
\end{matrix}\right),
\end{eqnarray}
where $\mathbf{d}_X\in \mathbb{F}_2^{m_X^B\times m_Z^C}$ and $\mathbf{d}_Z\in \mathbb{F}_2^{m_Z^B\times m_X^C}$ satisfy the condition $\mathbf{a}_X\mathbf{d}_X = (\mathbf{a}_Z\mathbf{d}_Z)^\mathrm{T}$. Then, the condition i) in Definition \ref{def:symmetry} is satisfied. Notice that conditions ii) and iii) also need to be satisfied. Taking the logical controlled-NOT gate as an example, we have 
\begin{eqnarray}
\mathbf{d}_X = \left(\begin{matrix}
1 & 0 \\
0 & 0 \\
0 & 0 \\
0 & 1
\end{matrix}\right)
\end{eqnarray}
and 
\begin{eqnarray}
\mathbf{d}_Z = \left(\begin{matrix}
0 & 0 \\
1 & 0 \\
0 & 1 \\
0 & 0
\end{matrix}\right).
\end{eqnarray}

Codewords of $\mathbf{a}$ are generated by a matrix in the form 
\begin{eqnarray}
\mathbf{g} = \left(\begin{matrix}
\mathbf{g}_X & 0 \\
0 & \mathbf{g}_Z
\end{matrix}\right),
\label{eq:g}
\end{eqnarray}
where $\mathbf{g}_\alpha$ is the code generator matrix of $\mathbf{a}_\alpha$, i.e. $\mathbf{a}_\alpha\mathbf{g}_\alpha^\mathrm{T} = 0$ and $\mathrm{rank}(\mathbf{g}_\alpha) = m_\alpha^B - \mathrm{rank}(\mathbf{a}_\alpha)$. For the logical controlled-NOT gate, 
\begin{eqnarray}
\mathbf{g}_X = \left(\begin{matrix}
1 & 0 & 1 & 1 \\
0 & 1 & 0 & 1
\end{matrix}\right)
\end{eqnarray}
and 
\begin{eqnarray}
\mathbf{g}_Z = \left(\begin{matrix}
1 & 0 & 1 & 0 \\
0 & 1 & 1 & 1
\end{matrix}\right).
\end{eqnarray}

\subsection{Physical-qubit circuit}

In this section, we directly give the LDPC code representing a physical circuit that realises the logical circuit $\mathbf{a}$. The explanation will be given in the next section. Here, we show that the LDPC code of the physical circuit is in a simple form depending on the LDPC code of the logical circuit. 

Given the check matrix $\mathbf{a}$ of the logical circuit, the check matrix of the physical circuit reads 
\begin{eqnarray}
\mathbf{A} = \left(\begin{matrix}
0 & \mathbf{A}_Z \\
\mathbf{A}_X & 0
\end{matrix}\right),
\label{eq:A}
\end{eqnarray}
where 
\begin{eqnarray}
\mathbf{A}_X = \left(\begin{matrix}
\mathbf{a}_X\otimes\openone_n & \openone_{m_X^C}\otimes\mathbf{G}_X^\mathrm{T} \\
\mathbf{d}_X^\mathrm{T}\otimes\mathbf{G}_Z & 0
\end{matrix}\right),
\label{eq:AX}
\end{eqnarray} 
and 
\begin{eqnarray}
\mathbf{A}_Z = \left(\begin{matrix}
\mathbf{a}_Z\otimes\openone_n & \openone_{m_Z^C}\otimes\mathbf{G}_Z^\mathrm{T} \\
\mathbf{d}_Z^\mathrm{T}\otimes\mathbf{G}_X & 0
\end{matrix}\right).
\label{eq:AZ}
\end{eqnarray}
When $\mathbf{a}$ has the bit-check symmetry, $\mathbf{A}$ also has the symmetry, and the corresponding deleting matrix is 
\begin{eqnarray}
\mathbf{D} = \left(\begin{matrix}
\mathbf{D}_X & 0 \\
0 & \mathbf{D}_Z
\end{matrix}\right),
\end{eqnarray}
where 
\begin{eqnarray}
\mathbf{D}_\alpha = \left(\begin{matrix}
\mathbf{d}_\alpha\otimes\openone_n & 0 \\
0 & \openone_{m_\alpha^C}\otimes\openone_{r_\alpha}
\end{matrix}\right).
\end{eqnarray} 
We can find that $\mathbf{A}_X\mathbf{D}_X = (\mathbf{A}_Z\mathbf{D}_Z)^\mathrm{T}$. 

We have already given the check matrix of the physical circuit. Next, we give the error correction check matrix $\mathbf{B}$ and logical generator matrix $\mathbf{L}$. The error correction check matrix is 
\begin{eqnarray}
\mathbf{B} = \left(\begin{matrix}
\mathbf{B}_X & 0 \\
0 & \mathbf{B}_Z
\end{matrix}\right),
\end{eqnarray}
where 
\begin{eqnarray}
\mathbf{B}_\alpha = \left(\begin{matrix}
\openone_{m_\alpha^B}\otimes\mathbf{G}_\alpha & \mathbf{a}_\alpha^\mathrm{T}\otimes\openone_{r_\alpha}
\end{matrix}\right).
\label{eq:Balpha}
\end{eqnarray}
The logical generator matrix is 
\begin{eqnarray}
\mathbf{L} = \left(\begin{matrix}
\mathbf{L}_X & 0 \\
0 & \mathbf{L}_Z
\end{matrix}\right),
\end{eqnarray}
where 
\begin{eqnarray}
\mathbf{L}_\alpha = \left(\begin{matrix}
\mathbf{g}_\alpha\otimes\mathbf{J}_\alpha & 0
\end{matrix}\right).
\end{eqnarray}
We have $\mathbf{A}_X\mathbf{B}_X^\mathrm{T} = \mathbf{A}_Z\mathbf{B}_Z^\mathrm{T} = \mathbf{A}_X\mathbf{L}_X^\mathrm{T} = \mathbf{A}_Z\mathbf{L}_Z^\mathrm{T} = 0$. We can find that $\mathbf{B}$ and $\mathbf{L}$ are compatible with $\mathbf{A}$, and $\mathbf{B}$ is a valid error-correction check matrix. 

\subsection{Explanation}

In quantum error correction, an elementary operation is measuring generators of the stabiliser group, and the measurement is usually repeated. Let's take the repeated parity-check measurement as an example. Suppose the measurement is repeated for $m$ times. In such a circuit, logical qubits go through the identity gate. The check matrix $\mathbf{a}$ of the logical circuit, in which the identity gate is repeated for $m$ times, takes 
\begin{eqnarray}
\mathbf{a}_X = \mathbf{a}_Z = \mathbf{R}_{m+1}.
\label{eq:aXaZ}
\end{eqnarray}
Here, $m_X^B = m_Z^B = m+1$, $m_X^C = m_Z^C = m$, and 
\begin{eqnarray}
\mathbf{R}_{m+1} = \left(\begin{matrix}
1 & 1 & 0 & 0 & 0 & \cdots \\
0 & 1 & 1 & 0 & 0 & \cdots \\
\vdots & \vdots & \vdots & \vdots & \vdots & \ddots
\end{matrix}\right) \in \mathrm{F}_2^{m\times (m+1)}
\end{eqnarray}
is the check matrix of the repetition code with the distance $m+1$. The corresponding deleting matrix $\mathbf{d}$ takes 
\begin{eqnarray}
\mathbf{d}_X = \left(\begin{matrix}
\openone_m \\
0_{1\times m}
\end{matrix}\right),
\label{eq:ddX}
\end{eqnarray}
and 
\begin{eqnarray}
\mathbf{d}_Z = \left(\begin{matrix}
0_{1\times m} \\
\openone_m
\end{matrix}\right).
\label{eq:ddZ}
\end{eqnarray}
Here, $0_{1\times m}$ denotes a $1\times m$ zero matrix. The code generator matrix of $\mathbf{a}_\alpha$ is 
\begin{eqnarray}
\mathbf{g}_\alpha = \left(\begin{matrix}
1 & 1 & 1 & 1 & 1 & \cdots
\end{matrix}\right).
\end{eqnarray}

\begin{figure}[tbp]
\centering
\includegraphics[width=\linewidth]{\figpath/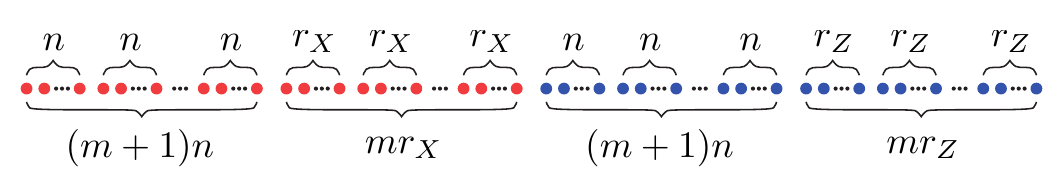}
\caption{
Bits on the Tanner graph of the repeated parity-check measurement. 
}
\label{fig:CSS}
\end{figure}

Next, we explain why the check matrix $\mathbf{A}$ represents the repeated parity-check measurement if one takes $\mathbf{a}$ according to Eq. (\ref{eq:aXaZ}). First, we need to explain the physical meaning of each column (each bit on the Tanner graph). There are two blocks of columns in $\mathbf{A}$. The block of $\mathbf{A}_X$ ($\mathbf{A}_Z$) represents $X$ ($Z$) Pauli operators. Let's look at the $X$ block. The $X$ block has two sub-blocks; see Eq. (\ref{eq:AX}) and Fig. \ref{fig:CSS}. The first sub-block has $m+1$ mini-blocks (notice that $m_X^B = m_Z^B = m+1$ and $m_X^C = m_Z^C = m$), and each mini-block has $n$ columns. In the first sub-block, the first mini-block represents the input $X$ operator of the $n$ physical qubits, and the $(l+1)$th mini-block represents the output $X$ operator after $l$ cycles of the parity-check measurement. The second sub-block has $m$ mini-blocks, and each mini-block has $r_X$ columns. In the second sub-block, the $l$th mini-block represents the $l$th cycle of the parity-check measurement; and in a codeword, their values determine whether corresponding measurement outcomes appear in the codeword equation. The meaning of the $Z$ block is similar. 

\begin{figure}[tbp]
\centering
\includegraphics[width=\linewidth]{\figpath/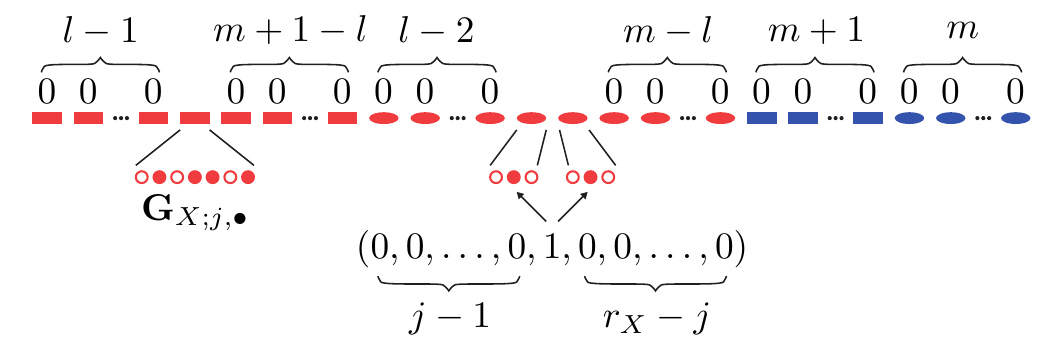}
\caption{
A codeword of the repeated parity-check measurement. The codeword is the $[(l-1)r_X+j]$th row of $\mathbf{B}$. Each rectangle represents a mini-block with the size of $n$, and each oval represents are mini-block with the size of $r_X$ or $r_Z$. Only three mini-blocks are nonzero. The second and the third nonzero mini-blocks take the same value. Take the Steane code as an example and suppose the second $X$ stabiliser generator is $X_2X_4X_5X_7$. Then, in the first non-zero mini-block, the second, fourth, fifth and seventh bits take the value of one, and others take the value of zero. In the second and third nonzero mini-blocks, the second bit takes the value of one, and others take the value of zero. 
}
\label{fig:CSS_B}
\end{figure}

With the meaning of columns clarified, we can understand the error correction check matrix $\mathbf{B}$ as follows. Rows of $\mathbf{B}$ are codewords of $\mathbf{A}$. $\mathbf{B}_X$ contributes to the first $(m+1)r_X$ rows. Let's consider the $[(l-1)r_X+j]$th row of $\mathbf{B}$ ($1<l<m+1$ and $1\leq j\leq r_X$). Look at Eq. (\ref{eq:Balpha}) and take $\alpha = X$: In the $l$th row of $\openone_{m_X^B}$, only the $l$th entry is one, and others are zero; and in the $l$th row of $\mathbf{a}_X^\mathrm{T}$, only the $(l-1)$th and $l$th entries are ones, and others are zeros. Therefore, the $[(l-1)r_X+j]$th row of $\mathbf{B}$ takes the following value; see Fig. \ref{fig:CSS_B}. In the first sub-block of the $X$ block, the $l$th mini-block takes the value of $\mathbf{G}_{X;j,\bullet}$ (corresponding to the $j$th $X$ stabiliser generator), and other mini-blocks are zero. In the second sub-block of the $X$ block, only the $(l-1)$th and $l$th mini-blocks are nonzero, and only the $j$th entry is one in these two mini-blocks. The $Z$ block is zero. According to the meaning of columns, this codeword is a checker. The checker includes two measurements, the measurement of the $j$th $X$ stabiliser generator in the $(l-1)$th and $l$th cycles. The product of two measurement outcomes should be $+1$ if the circuit is error-free. Some errors can flip the sign and be detected by the checker, including measurement errors on the two measurements, and $Z$ errors between the two cycles on qubits in the support of the $j$th $X$ stabiliser generator (qubits $q$ such that $\mathbf{G}_{X;j,q} = 1$). Similarly, rows with $l = 1$ are detectors, and rows with $l = m+1$ are emitters. Rows contributed by $\mathbf{B}_Z$ are similar. Therefore, the matrix $\mathbf{B}$ describes the parity checks in the repeated parity-check measurement. 

We can understand the logical generator matrix $\mathbf{L}$ in a similar way. Rows of $\mathbf{L}$ are also codewords of $\mathbf{A}$. $\mathbf{L}_X$ contributes to the first $k$ rows of $\mathbf{L}$. The $j$th row of $\mathbf{L}$ ($1\leq j\leq k$) takes the following value. In the first sub-block of the $X$ block, all mini-block takes the value of $\mathbf{J}_{X;j,\bullet}$ (corresponding to the $j$th logical $X$ operator). All other mini-blocks take the value of zero. Such a codeword is a genuine propagator, which transforms the $j$th $X$ logical operator from the beginning to the end of the circuit. Rows contributed by $\mathbf{L}_Z$ are similar. Therefore, the matrix $\mathbf{L}$ describes that logical operators are preserved in the repeated parity-check measurement. 

{\bf General CSS-type logical-qubit circuits.} The above explanation can be generalised to all CSS-type logical circuits. For a general circuit, the physical circuit $\mathbf{A}$ realises the logical circuit $\mathbf{a}$, which could be non-identity. To understand this, notice that codewords of $\mathbf{a}$ are generated by $\mathbf{g}$. These codewords represent correlations established by the logical circuit $\mathbf{a}$. For the physical circuit, codewords generated by $\mathbf{L}$ represent exactly the same correlations as $\mathbf{g}$ but on logical qubits. 

\subsection{Analysing the circuit code distance}

To complete the discussion on the example, we analyse the circuit code distance using algebraic methods. We will show that the circuit code distance $d(\mathbf{A},\mathbf{B},\mathbf{C})$ is the same as the code distance $d_{CSS}$ of the CSS code. For clarity, we call a logical error in the CSS code a code logical error, and a logical error in the circuit a circuit logical error. The minimum weight of code logical errors is $d_{CSS}$, and the minimum weight of circuit logical errors is $d(\mathbf{A},\mathbf{B},\mathbf{C})$. 

Let $\mathbf{e} = (\mathbf{e}_X,\mathbf{e}_Z)$ be a spacetime error, where $\mathbf{e}_\alpha\in \mathbb{F}_2^{m_\alpha^Bn+m_\alpha^Cr_\alpha}$ denotes errors on $\alpha = X,Z$ bits. The error $\mathbf{e}$ is a circuit logical error if and only if $\mathbf{B}\mathbf{e}^\mathrm{T} = 0$ (i.e. $\mathbf{B}_X\mathbf{e}_X^\mathrm{T} = \mathbf{B}_Z\mathbf{e}_Z^\mathrm{T} = 0$) and $\mathbf{L}\mathbf{e}^\mathrm{T} \neq 0$ ($\mathbf{L}_X\mathbf{e}_X^\mathrm{T} \neq 0$ or $\mathbf{L}_Z\mathbf{e}_Z^\mathrm{T} \neq 0$). 

Using row vectorisation, we can rewrite $\mathbf{e}_X$ in the form $\mathbf{e}_X = \left(\mathrm{vec}(\mathbf{u}^\mathrm{T})^\mathrm{T},\mathrm{vec}(\mathbf{v}^\mathrm{T})^\mathrm{T}\right)$, where $\mathbf{u}\in \mathbb{F}_2^{m_X^B\times n}$, $\mathbf{v}\in \mathbb{F}_2^{m_\alpha^C\times r_\alpha}$, $\mathrm{vec}(\bullet)$ denotes column vectorisation, and $\mathrm{vec}(\bullet^\mathrm{T})^\mathrm{T}$ is row vectorisation. An observation is that $\vert \mathbf{e}_X\vert \geq \vert \mathbf{w}\mathbf{u}\vert$ for all $\mathbf{w}\in \mathbb{F}_2^{m_X^B}$. As a consequence of the observation, $\vert \mathbf{e}_X\vert \geq \max_{p}\vert \mathbf{g}_{X;p,\bullet}\mathbf{u}\vert$. 

Suppose $\mathbf{B}_X\mathbf{e}_X^\mathrm{T} = 0$ and $\mathbf{L}_X\mathbf{e}_X^\mathrm{T} \neq 0$. Then, we have $\mathbf{u}$ and $\mathbf{v}$ satisfy the following conditions: $\mathbf{u}\mathbf{G}_X^\mathrm{T}+\mathbf{a}_X^\mathrm{T}\mathbf{v} = 0$ and $\mathbf{g}_X\mathbf{u}\mathbf{J}_X^\mathrm{T} \neq 0$. Multiple $\mathbf{g}_X$ to both sides of the first equation, it becomes $\mathbf{g}_X\mathbf{u}\mathbf{G}_X^\mathrm{T} = 0$. From the two equations, we can conclude that at least one row of $\mathbf{g}_X\mathbf{u}$ is a code logical error, i.e. $\max_{p}\vert \mathbf{g}_{X;p,\bullet}\mathbf{u}\vert \geq d_{CSS}$. Therefore, $\vert \mathbf{e}_X\vert \geq d_{CSS}$. 

Similarly, if $\mathbf{B}_Z\mathbf{e}_Z^\mathrm{T} = 0$ and $\mathbf{L}_Z\mathbf{e}_Z^\mathrm{T} \neq 0$, $\vert \mathbf{e}_Z\vert \geq d_{CSS}$. Overall, if $\mathbf{e}$ is a circuit logical error, $\vert \mathbf{e}\vert \geq d_{CSS}$, i.e. $d(\mathbf{A},\mathbf{B},\mathbf{C}) \geq d_{CSS}$. 

Consider a spacetime error, in which one of $n$-size mini-blocks takes a code logical error with a weight of $d_{CSS}$, and other mini-blocks are zero. Such a spacetime error represents that there is a code logical error occurring on qubits at a certain time. It is a circuit logical error with a weight of $d_{CSS}$. Therefore, $d(\mathbf{A},\mathbf{B},\mathbf{C}) = d_{CSS}$. 

Although the implementation of transversal operations on CSS codes is well established, the example demonstrates the possibility of devising fault-tolerant quantum circuits for executing certain logical operations by constructing an LDPC code. 

\RED{\section{A resource-efficient protocol for hypergraph product codes}
\label{sec:new_circuit}

In this section, we present an example of designing fault-tolerant quantum circuits using the LDPC representation. Inspired by the structure in the LDPC representation of transversal circuits, we propose a protocol of universal fault-tolerant quantum computing on hypergraph product codes. The protocol is based on encoding CSS-type logical circuits in a spatial dimension, which is constructed through algebraic analysis and takes advantage of the LDPC representation. 

A recent work demonstrates that hypergraph product and quasi-cyclic lifted product codes are promising candidates for fault-tolerant quantum computing because of their practical noise thresholds \cite{Xu2024}. In the same work, a protocol of fault-tolerant quantum computing is proposed. In comparison, our protocol can apply operations on more logical qubits in parallel and use fewer qubits to implement the operations. Specifically, logical qubits in a hypergraph product code form a two-dimensional array. In our protocol, we can operate logical qubits in a line-wise way, resulting in parallel operations and a reduction of the qubit cost by a factor of $\Omega(k_0d_0/n_0)$, where $[n_0,k_0,d_0]$ are parameters of the linear code for constructing the hypergraph product code. Notice that the primary reason for choosing a general hypergraph product code is the high encoding rate compared to the surface code, and the high encoding rate implies $k_0d_0/n_0\gg 1$, i.e. our protocol can reduce the qubit cost by a large factor. 

\subsection{Structure in transversal circuits}
\label{sec:structure}

A hypergraph product code is generated by two linear codes \cite{Tillich2014}. The two check matrices are 
\begin{eqnarray}
\mathbf{G}_X = \left(\begin{matrix}
\openone_{n_1}\otimes\mathbf{H}_2 & \mathbf{H}_1^\mathrm{T}\otimes\openone_{r_2}
\end{matrix}\right)
\end{eqnarray}
and 
\begin{eqnarray}
\mathbf{G}_Z = \left(\begin{matrix}
\mathbf{H}_1\otimes\openone_{n_2} & \openone_{r_1}\otimes\mathbf{H}_2^\mathrm{T}
\end{matrix}\right),
\end{eqnarray}
where $\mathbf{H}_i\in\mathbb{F}_2^{r_i\times n_i}$ is the check matrix of a linear code; we denote the quantum code by $\mathrm{HGP}(\mathbf{H}_1,\mathbf{H}_2)$. For simplicity, we suppose that $\mathbf{H}_i$ are full rank. Then, each $\mathbf{H}_i$ represents a $[n_i,k_i,d_i]$ code, where $k_i = n_i - r_i$. For the quantum code, parameters are $[[n,k,d]]$, where $n = n_1n_2+r_1r_2$, $k = k_1k_2$ and $d\geq \min\{d_1,d_2\}$. Let $\mathbf{K}_i\in\mathbb{F}_2^{k_i\times n_i}$ be the generator matrix for the code of $\mathbf{H}_i$, then the generator matrices for the quantum code are 
\begin{eqnarray}
\mathbf{J}_X = \left(\begin{matrix}
\mathbf{K}_1\otimes\mathbf{K}_2^\mathrm{rT} & 0_{k_1k_2\times r_1r_2}
\end{matrix}\right)
\end{eqnarray}
and 
\begin{eqnarray}
\mathbf{J}_Z = \left(\begin{matrix}
\mathbf{K}_1^\mathrm{rT}\otimes\mathbf{K}_2 & 0_{k_1k_2\times r_1r_2}
\end{matrix}\right),
\end{eqnarray}
where $\bullet^\mathrm{r}$ denotes the right inverse of the matrix $\bullet$. From now on, we neglect the dimensions of identity and zero matrices unless they cannot be derived from other elements or relevant matrices. 

Now, we compare check matrices of transversal circuits to the hypergraph product code. We can find that the first line in $\mathbf{A}_X$ [see Eq. (\ref{eq:AX})] and $\mathbf{B}_X$ [see Eq. (\ref{eq:Balpha})] resemble $\mathbf{G}_Z$ and $\mathbf{G}_X$ of a hypergraph product code, respectively. Motivated by this similarity, we consider the case that the quantum code taken in transversal circuits is a hypergraph product code. Substituting the hypergraph product code into $\mathbf{A}_X$, $\mathbf{B}_X$ and $\mathbf{L}_X$, we have 
\begin{eqnarray}
\mathbf{A}_X = \left(\begin{matrix}
\mathbf{a}_X\otimes\openone\otimes\openone & 0 & \openone\otimes\openone\otimes\mathbf{H}_2^\mathrm{T} \\
0 & \mathbf{a}_X\otimes\openone\otimes\openone & \openone\otimes\mathbf{H}_1\otimes\openone \\
\mathbf{d}_X^\mathrm{T}\otimes\mathbf{H}_1\otimes\openone & \mathbf{d}_X^\mathrm{T}\otimes\openone\otimes\mathbf{H}_2^\mathrm{T} & 0
\end{matrix}\right), \notag \\
\end{eqnarray}
\begin{eqnarray}
\mathbf{B}_X = \left(\begin{matrix}
\openone\otimes\openone\otimes\mathbf{H}_2 & \openone\otimes\mathbf{H}_1^\mathrm{T}\otimes\openone & \mathbf{a}_X^\mathrm{T}\otimes\openone\otimes\openone
\end{matrix}\right)
\end{eqnarray}
and 
\begin{eqnarray}
\mathbf{L}_X = \left(\begin{matrix}
\mathbf{g}_X\otimes\mathbf{K}_1\otimes\mathbf{K}_2^\mathrm{rT} & 0 & 0
\end{matrix}\right).
\label{eq:LX}
\end{eqnarray}
If we replace $\mathbf{d}_X^\mathrm{T}$ with an identity matrix, $\mathbf{A}_X$ and $\mathbf{B}_X$ resemble check matrices of a three-dimensional homological product code \cite{Bravyi2013}. 

In the three-dimensional code defined by $\mathbf{A}_X$ and $\mathbf{B}_X$, the two dimensions of $\mathbf{H}_1$ and $\mathbf{H}_2$ are spatial dimensions, and the logical circuit $\mathbf{a}_X$ is encoded in the temporal dimension. In the following, we develop new circuits by moving $\mathbf{a}_X$ to a spatial dimension. 

A straightforward way of encoding $\mathbf{a}_X$ into a spatial dimension is replacing $\mathbf{H}_1$ with 
\begin{eqnarray}
\widetilde{\mathbf{H}}_1 = \left(\begin{matrix}
\openone\otimes\mathbf{H}_1 \\
\mathbf{a}_X\otimes\mathbf{K}_1^\mathrm{rT}
\end{matrix}\right).
\end{eqnarray}
The corresponding generator matrix is $\widetilde{\mathbf{K}}_1 = \mathbf{g}_X\otimes\mathbf{K}_1$. We remark that $\mathbf{g}_X$, which satisfies $\mathbf{a}_X\mathbf{g}_X^\mathrm{T} = 0$, describes correlations between $X$ Pauli operators established in the circuit of $\mathbf{a}_X$; see Sec. \ref{sec:example}. Accordingly, the quantum code has the generator matrix 
\begin{eqnarray}
\widetilde{\mathbf{J}}_X = \left(\begin{matrix}
\mathbf{g}_X\otimes\mathbf{K}_1\otimes\mathbf{K}_2^\mathrm{rT} & 0
\end{matrix}\right),
\end{eqnarray}
which is exactly the same as $\mathbf{L}_X$ in Eq. (\ref{eq:LX}) up to some zero-valued elements. 

Taking this code of $\widetilde{\mathbf{H}}_1$, we could realise a logical circuit as the same as the one realised by a transversal circuit, which leads to a problem when multiple logical qubits are encoded in each block of the code $\mathrm{HGP}(\mathbf{H}_1,\mathbf{H}_2)$. When a transversal circuit is applied on physical qubits, the circuit on logical qubits is also transversal. For example, the transversal controlled-NOT gate on two blocks of $\mathrm{HGP}(\mathbf{H}_1,\mathbf{H}_2)$ results in logical controlled-NOT gates on each pair of logical qubits in the two blocks. From the point view of the LDPC representation, this transversal constraint is due to the tensor product structure. Next, we modify the code to individually address logical qubits in a block, i.e. we can choose to apply different operations on different logical qubits without the transversal constraint. 

\subsection{Spatially encoded CSS-type logical circuits}
\label{sec:encoding}

In this section, we partially solve the transversal constraint, and the problem is completely solved with the protocol in Sec. \ref{sec:universal}. For simplicity, we focus on logical circuits that are one-layer circuits consisting of the initialisation in the $X$ basis, measurement in the $X$ basis and Clifford gates. With the same method, we can also realise the logical initialisation and measurement in the $Z$ basis. Based on these one-layer logical circuits, we develop a protocol of universal fault-tolerant quantum computing as shown in Sec. \ref{sec:universal}. 

In the one-layer logical circuit, we suppose that $m_I$ logical qubits are initialised in the $X$ basis, $m_M$ logical qubits are measured in the $X$ basis, and CSS-type Clifford gates are applied on $m_G$ logical qubits. Therefore, the logical circuit has $m_M+m_G$ input logical qubits and $m_G+m_I$ output logical qubits. Without loss of generality, we suppose that the first $m_M$ input logical qubits are measured, the last $m_I$ output logical qubits are initialised, and gates are applied on other logical qubits. For such a logical circuit, the check matrix is the form of Eq. (\ref{eq:a}), where $\mathbf{a}_\alpha = (\mathbf{a}_{\alpha,in}, \mathbf{a}_{\alpha,out})$, 
\begin{eqnarray}
\mathbf{a}_{X,in} = \left(\begin{matrix}
0_{m_G\times m_M} & \mathbf{M}_X
\end{matrix}\right)
\end{eqnarray}
and 
\begin{eqnarray}
\mathbf{a}_{Z,in} = \left(\begin{matrix}
\openone_{m_M} & 0_{m_M\times m_G} \\
0_{m_G\times m_M} & \mathbf{M}_Z \\
0_{m_I\times m_M} & 0_{m_I\times m_G}
\end{matrix}\right)
\end{eqnarray}
act on input bit vertices, and 
\begin{eqnarray}
\mathbf{a}_{X,out} = \left(\begin{matrix}
\openone_{m_G} & 0_{m_G\times m_I}
\end{matrix}\right)
\end{eqnarray}
and 
\begin{eqnarray}
\mathbf{a}_{Z,out} = \left(\begin{matrix}
0_{m_M\times m_G} & 0_{m_M\times m_I} \\
\openone_{m_G} & 0_{m_G\times m_I} \\
0_{m_I\times m_G} & \openone_{m_I}
\end{matrix}\right)
\end{eqnarray}
act on output bit vertices. Here, $\mathbf{M}_X$ and $\mathbf{M}_Z$ are the two blocks of $\mathbf{M}_U$, where $U$ is the product of all Clifford gates in the circuit: They represent the action on $X$ and $Z$ operators, respectively, i.e. $\mathbf{M}_U(\mathbf{x}_0,\mathbf{z}_0)^\mathrm{T} = (\mathbf{x}_1,\mathbf{z}_1)^\mathrm{T}$, where $\mathbf{x}_1^\mathrm{T} = \mathbf{M}_X\mathbf{x}_1^\mathrm{T}$ and $\mathbf{z}_1^\mathrm{T} = \mathbf{M}_Z\mathbf{z}_1^\mathrm{T}$. Because Clifford gates preserve the commutative and anti-commutative relations between Pauli operators, $\mathbf{M}_Z^\mathrm{T}\mathbf{M}_X = \openone_{m_G}$. 

Similarly, we can express $\mathbf{g}_\alpha$ matrices in the input-output form, i.e. $\mathbf{g}_\alpha = (\mathbf{g}_{\alpha,in}, \mathbf{g}_{\alpha,out})$ [see Eq. (\ref{eq:g})]. Here, 
\begin{eqnarray}
\mathbf{g}_{X,in} = \left(\begin{matrix}
\openone_{m_M} & 0_{m_M\times m_G} \\
0_{m_G\times m_M} & \openone_{m_G} \\
0_{m_I\times m_M} & 0_{m_I\times m_G}
\end{matrix}\right),
\end{eqnarray}
\begin{eqnarray}
\mathbf{g}_{Z,in} = \left(\begin{matrix}
0_{m_G\times m_M} & \openone_{m_G}
\end{matrix}\right),
\end{eqnarray}
\begin{eqnarray}
\mathbf{g}_{X,out} = \left(\begin{matrix}
0_{m_M\times m_G} & 0_{m_M\times m_I} \\
\mathbf{M}_X^\mathrm{T} & 0_{m_G\times m_I} \\
0_{m_I\times m_G} & \openone_{m_I}
\end{matrix}\right),
\end{eqnarray}
and 
\begin{eqnarray}
\mathbf{g}_{Z,out} = \left(\begin{matrix}
\mathbf{M}_Z^\mathrm{T} & 0_{m_G\times m_I}
\end{matrix}\right).
\end{eqnarray}

We consider two linear codes. Their check matrices are $\mathbf{H}_{in}\in\mathbb{F}_2^{r_{in}\times n_{in}}$ and $\mathbf{H}_{out}\in\mathbb{F}_2^{r_{out}\times n_{out}}$, and code parameters are $[n_{in/out},k_{in/out},d_{in/out}]$. For these two codes, generator matrices are denoted by $\mathbf{K}_{in}$ and $\mathbf{K}_{out}$, respectively. By taking $\mathbf{H}_1 = \mathbf{H}_{in},\mathbf{H}_{out}$, we can generate two quantum codes $\mathrm{HGP}(\mathbf{H}_{in},\mathbf{H}_1)$ and $\mathrm{HGP}(\mathbf{H}_{out},\mathbf{H}_2)$, which are depicted by matrices $(\mathbf{G}_{X,in/out},\mathbf{G}_{Z,in/out},\mathbf{J}_{X,in/out},\mathbf{J}_{Z,in/out})$. The two quantum codes encode the input and output logical qubits, respectively, i.e. we take $k_{in} = m_M+m_G$ and $k_{out} = m_G+m_I$. 

Now, we introduce a new code for encoding the logical circuit. Its check matrix reads 
\begin{eqnarray}
\overline{\mathbf{H}} = \left(\begin{matrix}
\mathbf{H}_{in} & 0 \\
0 & \mathbf{H}_{out} \\
\mathbf{a}_{X,in}\mathbf{K}_{in}^\mathrm{rT} & \mathbf{a}_{X,out}\mathbf{K}_{out}^\mathrm{rT}
\end{matrix}\right).
\label{eq:Hbar}
\end{eqnarray}
Here, $\mathbf{K}_{in/out}^\mathrm{rT} = (\openone_{k_{in/out}},0)$ is sparse for a generator matrix in the standard form $\mathbf{K}_{in/out} = (\openone_{k_{in/out}},\cdots)$. For this new code, parameters are $[\overline{n},\overline{k},\overline{d}]$, where $\overline{n} = n_{in}+n_{out}$, $\overline{k} = m_M+m_G+m_I$, and $\overline{d}\geq \min\{d_{in},d_{out}\}$. Its generator matrix is 
\begin{eqnarray}
\overline{\mathbf{K}} = \left(\begin{matrix}
\mathbf{g}_{X,in}\mathbf{K}_{in} & \mathbf{g}_{X,out}\mathbf{K}_{out}
\end{matrix}\right).
\end{eqnarray}
By taking $\mathbf{H}_1 = \overline{\mathbf{H}}$, we generate a quantum code $\mathrm{HGP}(\overline{\mathbf{H}},\mathbf{H}_2)$ depicted by matrices $(\overline{\mathbf{G}}_X,\overline{\mathbf{G}}_Z,\overline{\mathbf{J}}_X,\overline{\mathbf{J}}_Z)$, where 
\begin{eqnarray}
\overline{\mathbf{J}}_X = \left(\begin{matrix}
\mathbf{g}_{X,in}\mathbf{K}_{in}\otimes\mathbf{K}_2^\mathrm{rT} & \mathbf{g}_{X,out}\mathbf{K}_{out}\otimes\mathbf{K}_2^\mathrm{rT} & 0
\end{matrix}\right).
\end{eqnarray}
Although $\overline{\mathbf{J}}_X$ is different from $\mathbf{L}_X$ in Eq. (\ref{eq:LX}), it establishes the same correlation as $\mathbf{L}_X$. We explain it next. 

\begin{figure}[tbp]
\centering
\includegraphics[width=\linewidth]{\figpath/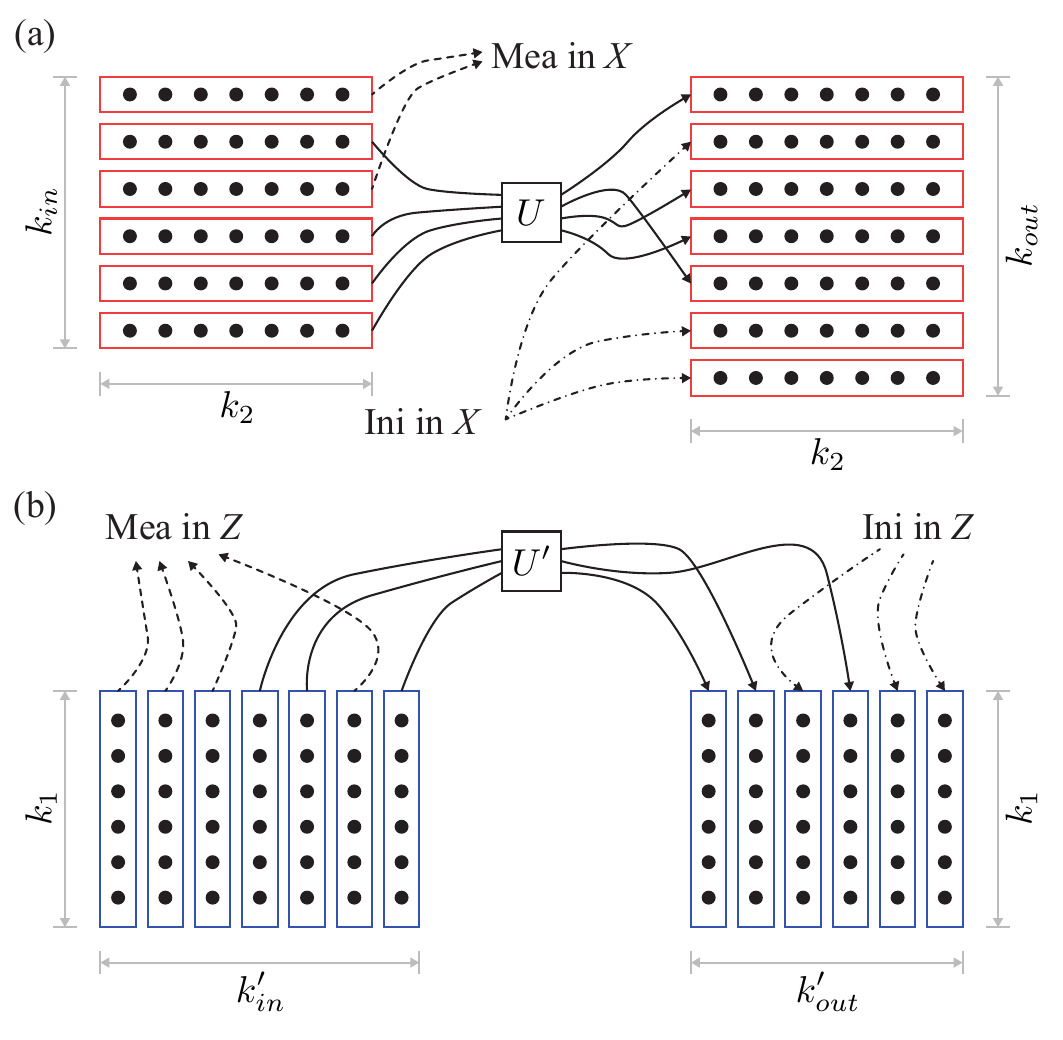}
\caption{
(a) Row operation. Logical qubits in the code $\mathrm{HGP}(\mathbf{H}_{in/out},\mathbf{H}_2)$ form a $k_{in/out}\times k_2$ array. Each black circle represents a logical qubit. Through the code $\mathrm{HGP}(\overline{\mathbf{H}},\mathbf{H}_2)$, we can implement the row operation, in which the circuit $\mathbf{a}$ is applied on logical qubits in a column, and the same circuit is applied on all columns simultaneously. This figure illustrates an example: In each column of the input block, the first and third logical qubits are measured (Mea) in the $X$ basis, and other logical qubits go through the Clifford gate $U$; In each column of the output block, the second, sixth and seventh logical qubits are initialised (Ini) in the $X$ basis, and other logical qubits store the output state of the Clifford gate. (b) Column operation. The roles of rows and columns are exchanged compared with the row operation. The basis of the initialisation and measurement is $Z$ in the column operation. 
}
\label{fig:primitive_operations}
\end{figure}

In each block of a hypergraph product code, logical qubits form a two-dimensional array; see Fig. \ref{fig:primitive_operations} (a): There are $k_{in/out}\times k_2$ logical qubits in a block of $\mathrm{HGP}(\mathbf{H}_{in/out},\mathbf{H}_2)$. We can label each logical qubit with coordinates $(l_1,l_2)$, where $l_1 = 1,2,\ldots,k_{in/out}$ and $l_2 = 1,2,\ldots,k_2$. The logical $X$ and $Z$ operators of the logical qubit $(l_1,l_2)$ are represented by $(\mathbf{K}_{in/out})_{l_1,\bullet}\otimes(\mathbf{K}_2^\mathrm{rT})_{l_2,\bullet}$ and $(\mathbf{K}_{in/out}^\mathrm{rT})_{l_1,\bullet}\otimes(\mathbf{K}_2)_{l_2,\bullet}$, respectively. For simplicity, we consider the case that only gates are applied on logical qubits, i.e. $m_M = m_I = 0$ and $k_{in} = k_{out}$. We can analyse the general case in a similar way. With the simplification, the correlations established between two blocks of logical qubits become 
\begin{eqnarray}
\overline{\mathbf{J}}_X = \left(\begin{matrix}
\mathbf{K}_{in}\otimes\mathbf{K}_2^\mathrm{rT} & \mathbf{M}_X^\mathrm{T}\mathbf{K}_{out}\otimes\mathbf{K}_2^\mathrm{rT} & 0
\end{matrix}\right).
\end{eqnarray}
Accordingly, the logical operator $(\mathbf{K}_{in})_{l_1,\bullet}\otimes(\mathbf{K}_2^\mathrm{rT})_{l_2,\bullet}$ in the input block is mapped to the operator $(\mathbf{M}_X^\mathrm{T}\mathbf{K}_{out})_{l_1,\bullet}\otimes(\mathbf{K}_2^\mathrm{rT})_{l_2,\bullet}$ in the output block. From this, we have two conclusions. First, focusing on logical qubits in the column-$l_2$, we can find that $X$ operators of the input block are mapped to $X$ operators of the output block, and the map is $\mathbf{M}_X$. Second, the map is the same for all $l_2$. Then, there are two problems to be addressed. First, although the map on logical $X$ operators is the same as the desired circuit, we still need to verify the map on logical $Z$ operators, which is given in Sec. \ref{sec:CSS}. Second, the logical operations are still constrained in one dimension, and we solve this problem in Sec. \ref{sec:universal}. 

\subsection{Protocol for CSS-type logical circuits}
\label{sec:CSS}

\begin{figure}[tbp]
\centering
\includegraphics[width=\linewidth]{\figpath/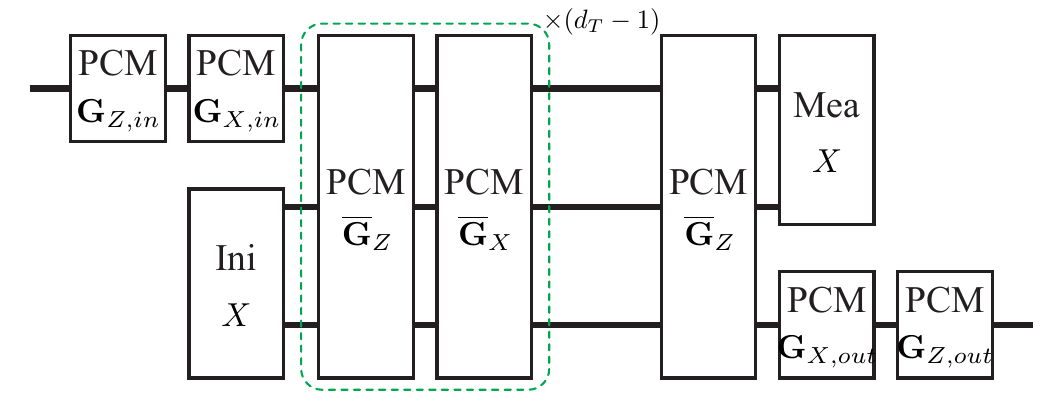}
\caption{
Circuit of the row operation. The input logical state is encoded in the code $(\mathbf{G}_{X,in},\mathbf{G}_{Z,in})$, and other qubits are initialised (Ini) in the $X$ basis. After a few rounds of parity-check measurements (PCMs) of the code $(\overline{\mathbf{G}}_X,\overline{\mathbf{G}}_Z)$, the output logical state is encoded in the code $(\mathbf{G}_{X,out},\mathbf{G}_{Z,out})$, and other qubits are measured (Mea) in the $X$ basis. To implement the column operation, we need to construct a code $(\overline{\mathbf{G}}_X,\overline{\mathbf{G}}_Z)$ according to the desired logical circuit and replace $X$ and $Z$ with each other in the circuit. 
}
\label{fig:circuit_new}
\end{figure}

We have encoded the logical circuit in a spatial dimension and constructed a check matrix $\overline{\mathbf{H}}$. In this section, we propose a protocol for implementing CSS-type logical circuits based on $\overline{\mathbf{H}}$ and analyse it in the LDPC representation. The protocol is a generalisation of lattice surgery \cite{Horsman2012,Litinski2019,Vuillot2019,Cohen2022}. 

We utilise three quantum codes, the input code $(\mathbf{G}_{X,in},\mathbf{G}_{Z,in})$, output code $(\mathbf{G}_{X,out},\mathbf{G}_{Z,out})$ and bridge code $(\overline{\mathbf{G}}_X,\overline{\mathbf{G}}_Z)$. If we remove the third row from $\overline{\mathbf{H}}$ in Eq. (\ref{eq:Hbar}), the bridge code becomes $(\mathbf{G}_{X,in}\oplus\mathbf{G}_{X,out},\mathbf{G}_{Z,in}\oplus\mathbf{G}_{Z,out})$, implying that the supports of input and output codes are in the support of the bridge code. The protocol includes three steps as shown in Fig. \ref{fig:circuit_new}. In the input state, the $m_M+m_G$ input logical qubits are encoded in the input code. First, we apply the $X$-basis initialisation on all qubits except for those in the support of the input code. Second, we implement parity check operations according to the bridge code for $d_T$ rounds. Finally, we apply the $X$-basis measurement on all qubits except for those in the support of the output code. In the output state, the $m_G+m_I$ output logical qubits are encoded in the output code. The logical operations realised by this protocol are depicted in Fig. \ref{fig:primitive_operations}(a). 

In the LDPC representation, we represent the stabiliser circuit described in the above protocol with check matrices 
\begin{eqnarray}
\mathbf{A}_X &=& \boldsymbol{\beta}_{in}\oplus (d_T-1)\openone_{\overline{n}}\oplus\boldsymbol{\beta}_{out} \notag \\
&&\oplus\boldsymbol{\gamma}_{Z,in}\oplus d_T\openone_{\overline{r}_Z}\oplus\boldsymbol{\gamma}_{Z,out} \notag \\
&&\times\left(\begin{matrix}
\mathbf{R}_{d_T+2}\otimes\openone_{\overline{n}} & \openone_{d_T+1}\otimes\overline{\mathbf{G}}_X^\mathrm{T} \\
\openone_{d_T+2}\otimes\overline{\mathbf{G}}_Z & 0
\end{matrix}\right) \notag \\
&&\times\boldsymbol{\beta}_{in}^\mathrm{T}\oplus d_T\openone_{\overline{n}}\oplus\boldsymbol{\beta}_{out}^\mathrm{T} \notag \\
&&\oplus\boldsymbol{\gamma}_{X,in}^\mathrm{T}\oplus (d_T-1)\openone_{\overline{r}_X}\oplus\boldsymbol{\gamma}_{X,out}^\mathrm{T}
\end{eqnarray}
and 
\begin{eqnarray}
\mathbf{A}_Z &=& \boldsymbol{\beta}_{in}\oplus d_T\openone_{\overline{n}}\oplus\boldsymbol{\beta}_{out} \notag \\
&&\oplus\boldsymbol{\gamma}_{X,in}\oplus (d_T-1)\openone_{\overline{r}_X}\oplus\boldsymbol{\gamma}_{X,out} \notag \\
&&\times\left(\begin{matrix}
\mathbf{R}_{d_T+3}\otimes\openone_{\overline{n}} & \openone_{d_T+2}\otimes\overline{\mathbf{G}}_Z^\mathrm{T} \\
\mathbf{d}\otimes\overline{\mathbf{G}}_X & 0
\end{matrix}\right) \notag \\
&&\times 2\boldsymbol{\beta}_{in}^\mathrm{T}\oplus (d_T-1)\openone_{\overline{n}}\oplus 2\boldsymbol{\beta}_{out}^\mathrm{T} \notag \\
&&\oplus\boldsymbol{\gamma}_{Z,in}^\mathrm{T}\oplus d_T\openone_{\overline{r}_Z}\oplus\boldsymbol{\gamma}_{Z,out}^\mathrm{T},
\end{eqnarray}
where $\mathbf{d} = (0_{(d_T+1)\times 1},\openone_{d_T+1},0_{(d_T+1)\times 1})$, and $i\mathbf{u}\oplus j\mathbf{v} = \mathrm{diag}(\mathbf{u},\mathbf{u},\ldots,\mathbf{v},\mathbf{v},\ldots)$ denotes the direct sum of matrices $\mathbf{u}$ and $\mathbf{v}$ with multiple numbers $i$ and $j$, respectively. Matrices $\boldsymbol{\beta}_{in/out}$ and $\boldsymbol{\gamma}_{\alpha,in/out}$ ($\alpha = X,Z$) are given in Appendix \ref{app:matrices}, and they are sparse matrices satisfying the following equations, 
\begin{eqnarray}
\boldsymbol{\beta}_{in/out}\boldsymbol{\beta}_{in/out}^\mathrm{T} &=& \openone_{n_{in/out}n_2+r_{in/out}r_2}, \\
\boldsymbol{\gamma}_{X,in/out}\boldsymbol{\gamma}_{X,in/out}^\mathrm{T} &=& \openone_{n_{in/out}r_2}, \\
\boldsymbol{\gamma}_{Z,in/out}\boldsymbol{\gamma}_{Z,in/out}^\mathrm{T} &=& \openone_{r_{in/out}n_2}, \\
\overline{\mathbf{G}}_X\boldsymbol{\beta}_{in/out}^\mathrm{T} &=& \boldsymbol{\gamma}_{X,in/out}^\mathrm{T}\mathbf{G}_{X,in/out}, \\
\mathbf{G}_{Z,in/out}\boldsymbol{\beta}_{in/out} &=& \boldsymbol{\gamma}_{Z,in/out}\overline{\mathbf{G}}_Z, \\
(\mathbf{g}_{X,in/out}\otimes\openone)\mathbf{J}_{X,in/out} &=& \overline{\mathbf{J}}_X\boldsymbol{\beta}_{in/out}^\mathrm{T}, \\
\overline{\mathbf{G}}_X(\mathbf{J}_{Z,in/out}\boldsymbol{\beta}_{in/out})^\mathrm{T} &=& 0.
\end{eqnarray}
Matrices $\boldsymbol{\beta}_{in/out}$ map qubits from the support of $\mathrm{HGP}(\overline{\mathbf{H}},\mathbf{H}_2)$ to the support of $\mathrm{HGP}(\mathbf{H}_{in/out},\mathbf{H}_2)$,  matrices $\boldsymbol{\gamma}_{\alpha,in/out}$ map $\alpha$ checks of $\mathrm{HGP}(\overline{\mathbf{H}},\mathbf{H}_2)$ to checks of $\mathrm{HGP}(\mathbf{H}_{in/out},\mathbf{H}_2)$, and they connect the three quantum codes. The error correction check matrices are 
\begin{eqnarray}
\mathbf{B}_X &=& \boldsymbol{\gamma}_{X,in}\oplus d_T\openone_{\overline{r}_X}\oplus\boldsymbol{\gamma}_{X,out} \notag \\
&&\times\left(\begin{matrix}
\openone_{d_T+2}\otimes\overline{\mathbf{G}}_X & \mathbf{R}_{d_T+2}^\mathrm{T}\otimes\openone_{\overline{r}_X}
\end{matrix}\right) \notag \\
&&\times\boldsymbol{\beta}_{in}^\mathrm{T}\oplus d_T\openone_{\overline{n}}\oplus\boldsymbol{\beta}_{out}^\mathrm{T} \notag \\
&&\oplus\boldsymbol{\gamma}_{X,in}^\mathrm{T}\oplus (d_T-1)\openone_{\overline{r}_X}\oplus\boldsymbol{\gamma}_{X,out}^\mathrm{T}
\end{eqnarray}
and 
\begin{eqnarray}
\mathbf{B}_Z &=& 2\boldsymbol{\gamma}_{Z,in}\oplus (d_T-1)\openone_{\overline{r}_Z}\oplus 2\boldsymbol{\gamma}_{Z,out} \notag \\
&&\times\left(\begin{matrix}
\openone_{d_T+3}\otimes\overline{\mathbf{G}}_Z & \mathbf{R}_{d_T+3}^\mathrm{T}\otimes\openone_{\overline{r}_Z}
\end{matrix}\right) \notag \\
&&\times 2\boldsymbol{\beta}_{in}^\mathrm{T}\oplus (d_T-1)\openone_{\overline{n}}\oplus 2\boldsymbol{\beta}_{out}^\mathrm{T} \notag \\
&&\oplus\boldsymbol{\gamma}_{Z,in}^\mathrm{T}\oplus d_T\openone_{\overline{r}_Z}\oplus\boldsymbol{\gamma}_{Z,out}^\mathrm{T}.
\end{eqnarray}
These check matrices satisfy the requirement $\mathbf{A}_X\mathbf{B}_X^\mathrm{T} = \mathbf{A}_Z\mathbf{B}_Z^\mathrm{T}$. In Appendix \ref{app:matrices}, we illustrate matrices $\mathbf{A}_\alpha$ and $\mathbf{B}_\alpha$ in the block matrix form by taking $d_T = 3$ as an example. 

Now, we verify the propagation of logical operators using logical generator matrices 
\begin{widetext}
\begin{eqnarray}
\mathbf{L}_X = \left(\begin{matrix}
(\mathbf{g}_{X,in}\otimes\openone)\mathbf{J}_{X,in} &d_T\times\overline{\mathbf{J}}_X & (\mathbf{g}_{X,out}\otimes\openone)\mathbf{J}_{X,out} & (d_T+1)\times 0
\end{matrix}\right)
\end{eqnarray}
and 
\begin{eqnarray}
\mathbf{L}_Z^{(l)} &=& \Biggl(\begin{matrix}
2\times(\mathbf{g}_{Z,in}\otimes\openone)\mathbf{J}_{Z,in} & (l-1)\times(\mathbf{g}_{Z,in}\otimes\openone)\mathbf{J}_{Z,in}\boldsymbol{\beta}_{in} & (d_T-l)\times(\mathbf{g}_{Z,out}\otimes\openone)\mathbf{J}_{Z,out}\boldsymbol{\beta}_{out} & 2\times(\mathbf{g}_{Z,out}\otimes\openone)\mathbf{J}_{Z,out}
\end{matrix}\Biggr. \notag \\ &&\Biggl.\begin{matrix}
l\times 0 & \boldsymbol{\delta} & (d_T-l+1)\times 0
\end{matrix}\Biggr),
\end{eqnarray}
\end{widetext}
where we have introduced a notation $i\times\mathbf{u} = (\mathbf{u},\mathbf{u},\ldots)$ to denote repeating the matrix $\mathbf{u}$ in a row with the multiple number $i$. The role of the matrix 
\begin{eqnarray}
\boldsymbol{\delta} = (0_{m_G\times r_{in}n_2},0_{m_G\times r_{out}n_2},\mathbf{M}_Z\otimes \mathbf{K}_2)
\end{eqnarray}
is to couple $\mathbf{J}_{Z,in}\boldsymbol{\beta}_{in}$ and $\mathbf{J}_{Z,out}\boldsymbol{\beta}_{out}$: It satisfies the relation $(\mathbf{g}_{Z,in}\otimes\openone)\mathbf{J}_{Z,in}\boldsymbol{\beta}_{in}+(\mathbf{g}_{Z,out}\otimes\openone)\mathbf{J}_{Z,out}\boldsymbol{\beta}_{out} = \boldsymbol{\delta}\overline{\mathbf{G}}_Z$. Here, $l = 1,2,\ldots,d_T$, and all the generator matrices $\mathbf{L}_Z^{(l)}$ are equivalent in the sense that $\mathrm{rowsp}(\mathbf{L}_Z^{(l)} + \mathbf{L}_Z^{(l')}) \in \mathrm{rowsp}(\mathbf{B}_Z)$, i.e. for all undetectable spacetime errors $\mathbf{e}_Z$, $\mathbf{L}_Z^{(l)}\mathbf{e}_Z^\mathrm{T} = \mathbf{L}_Z^{(l')}\mathbf{e}_Z^\mathrm{T}$. The generator matrices satisfy the requirement $\mathbf{A}_X\mathbf{L}_X^\mathrm{T} = \mathbf{A}_Z\mathbf{L}_Z^{(l)\mathrm{T}} = 0$ for all $l$. In Appendix \ref{app:matrices}, we also illustrate matrices $\mathbf{L}_X$ and $\mathbf{L}_Z^{(l)}$ in the block matrix form by taking $d_T = 3$ and $l = 2$ as an example. We can find that the stabiliser circuit maps input logical operators from $(\mathbf{g}_{\alpha,in}\otimes\openone)\mathbf{J}_{\alpha,in}$ to output logical operators $(\mathbf{g}_{\alpha,out}\otimes\openone)\mathbf{J}_{\alpha,out}$ ($\alpha = X,Z$). Therefore, the stabiliser circuit realises the desired logical circuit $\mathbf{a}$ (for both $X$ and $Z$ operators). 

{\bf Fault tolerance.} Now, we can justify the fault tolerance through the circuit code distance. For the three quantum codes utilised in the protocol, their code distances have the lower bound $\min\{d_{in},d_{out},d_2\}$. The circuit code distance has the lower bound $d(\mathbf{A},\mathbf{B},\mathbf{L}) \geq \min\{d_{in},d_{out},d_2,d_T\}$. We prove it in what follows. 

We modify the distance proof for transversal circuits. Let $\mathbf{e}_X = (\mathbf{u}_0,\ldots,\mathbf{u}_{d_T+1},\mathbf{v}_1,\ldots,\mathbf{v}_{d_T+1})$ be the spacetime error on $X$ operators. Here, $\mathbf{u}_i$ represent data qubit errors, and $\mathbf{v}_j$ represent measurement errors. Let $\mathbf{u}_{sum} = \mathbf{u}_0\boldsymbol{\beta}_{in}+\mathbf{u}_1+\cdots+\mathbf{u}_{d_T}+\mathbf{u}_{d_T+1}\boldsymbol{\beta}_{out}$. Suppose $\mathbf{e}_X$ is a logical error, i.e. $\mathbf{B}_X\mathbf{e}_X^\mathrm{T} = 0$ and $\mathbf{L}_X\mathbf{e}_X^\mathrm{T} \neq 0$. Then, $\overline{\mathbf{G}}_X\mathbf{u}_{sum}^\mathrm{T} = 0$ and $\overline{\mathbf{J}}_X\mathbf{u}_{sum}^\mathrm{T} \neq 0$, i.e. $\mathbf{u}_{sum}$ is a logical error for the quantum code $\mathrm{HGP}(\overline{\mathbf{H}},\mathbf{H}_2)$. Therefore, $\vert \mathbf{e}_X \vert \geq \vert \mathbf{u}_{sum} \vert \geq \min\{d_{in},d_{out},d_2\}$. 

The situation for errors on $Z$ operators is slightly more complicated. Let $\mathbf{e}_Z = (\mathbf{u}_0,\ldots,\mathbf{u}_{d_T+2},\mathbf{v}_1,\ldots,\mathbf{v}_{d_T+2})$ be the spacetime error on $Z$ operators. We prove by contradiction: suppose $\mathbf{e}_Z$ is a logical error but $\vert \mathbf{e}_Z \vert < \min\{d_{in},d_{out},d_2,d_T\}$. Since $\vert \mathbf{e}_Z \vert < d_T$, there exists $1\leq j\leq d_T$ such that $\mathbf{v}_{j+1} = 0$. Then, we can decompose the spacetime error into two errors $\mathbf{e}_Z = \mathbf{e}_{Z,in}+\mathbf{e}_{Z,out}$ , where $\mathbf{e}_{Z,in} = (\mathbf{u}_0,\ldots,\mathbf{u}_j,0,\ldots,0,\mathbf{v}_1,\ldots,\mathbf{v}_j,0,\ldots,0)$  and $\mathbf{e}_{Z,out} = (0,\ldots,0,\mathbf{u}_{j+1},\ldots,\mathbf{u}_{d_T+2},0,\ldots,0,\mathbf{v}_{j+2},\ldots,\mathbf{v}_{d_T+2})$. Let $\mathbf{u}_{in} = \mathbf{u}_0+\mathbf{u}_1+\mathbf{u}_2\boldsymbol{\beta}_{in}+\cdots+\mathbf{u}_j\boldsymbol{\beta}_{in}$. Since $\mathbf{e}_Z$ is an undetectable error and $\mathbf{v}_{j+1} = 0$, $\mathbf{B}_Z\mathbf{e}_{Z,in}^\mathrm{T} = 0$ and $\mathbf{G}_{Z,in}\mathbf{u}_{in}^\mathrm{T} = 0$. Because $\vert\mathbf{u}_{in}\vert \leq \vert \mathbf{e}_Z \vert < \min\{d_{in},d_2\}$, $\mathbf{u}_{in}$ is a trivial error for the quantum code $\mathrm{HGP}(\mathbf{H}_{in},\mathbf{H}_2)$, i.e. $\mathbf{J}_{Z,in}\mathbf{u}_{in}^\mathrm{T} = 0$. Then, $\mathbf{L}_Z^{(j)}\mathbf{e}_{in}^\mathrm{T} = (\mathbf{g}_{Z,in}\otimes\openone)\mathbf{J}_{Z,in}\mathbf{u}_{in}^\mathrm{T} = 0$. Similarly, $\mathbf{L}_Z^{(j)}\mathbf{e}_{out}^\mathrm{T} = 0$. Therefore, $\mathbf{e}_Z$ is a trivial error (i.e. $\mathbf{L}_Z^{(j)}\mathbf{e}_Z^\mathrm{T} = 0$), which is a contradiction. 

{\bf Row and column operations.} We have given a protocol that can implement a CSS-type circuit on logical qubits, in which each row is operated as a whole; see Fig. \ref{fig:primitive_operations}(a). We call such an operation the row operation. In a similar way, we can realise the column operation, which is illustrated in Fig. \ref{fig:primitive_operations}(b): Compared with the row operation, the roles of the following items are exchanged, including two check matrices $\mathbf{H}_1$ and $\mathbf{H}_2$ in hypergraph product codes, $X$ and $Z$ Pauli operators, and rows and columns in the logical qubit array. 

In the row operation, we use three codes $\mathrm{HGP}(\mathbf{H}_{in},\mathbf{H}_2)$, $\mathrm{HGP}(\mathbf{H}_{out},\mathbf{H}_2)$ and $\mathrm{HGP}(\overline{\mathbf{H}},\mathbf{H}_2)$ to implement a logical circuit. In the column operation, we also use three codes $\mathrm{HGP}(\mathbf{H}_1,\mathbf{H}_{in}')$, $\mathrm{HGP}(\mathbf{H}_1,\mathbf{H}_{out}')$ and $\mathrm{HGP}(\mathbf{H}_1,\overline{\mathbf{H}}')$. 
%This time, the first check matrix $\mathbf{H}_1$ in the hypergraph product is the same in three quantum codes, and the second check matrix is different. 
The input and output blocks are encoded in $\mathrm{HGP}(\mathbf{H}_1,\mathbf{H}_{in}')$ and $\mathrm{HGP}(\mathbf{H}_1,\mathbf{H}_{out}')$, respectively. The bridge code $\mathrm{HGP}(\mathbf{H}_1,\overline{\mathbf{H}}')$ is constructed according to the desired logical circuit. The circuit of the column operation is similar to Fig. \ref{fig:circuit_new}. 

The column operation realises logical operations different from the row operation. First, instead of the initialisation and measurement in the $X$ basis, the column operation realises a one-layer logical circuit consisting of the initialisation and measurement in the $Z$ basis on a number of logical qubits and Clifford gates on some other logical qubits. Given such a logical circuit, we can work out the corresponding $\mathbf{a}_{Z,in}$ and $\mathbf{a}_{Z,out}$. Then, the bridge check matrix in the column operation reads 
\begin{eqnarray}
\overline{\mathbf{H}}' = \left(\begin{matrix}
\mathbf{H}_{in}' & 0 \\
0 & \mathbf{H}_{out}' \\
\mathbf{a}_{Z,in}\mathbf{K}_{in}^\mathrm{\prime rT} & \mathbf{a}_{Z,out}\mathbf{K}_{out}^\mathrm{\prime rT}
\end{matrix}\right).
\end{eqnarray}
Second, the column operation implements the desired circuit on logical qubits in a row (instead of a column), and the same circuit is implemented on all rows, as shown in Fig. \ref{fig:primitive_operations}(b). 

In the previous discussions, we have assumed that the first $m_M$ input logical qubits are measured, the last $m_I$ output logical qubits are initialised, and gates are applied on other logical qubits. This assumption is unnecessary, and we can choose any subset of logical qubits to be measured and initialised, as shown in Fig. \ref{fig:primitive_operations}. Such an choice can be achieved by permuting rows in $\mathbf{K}_{in}$ and $\mathbf{K}_{out}$ (or $\mathbf{K}_{in}'$ and $\mathbf{K}_{out}'$) in the bridge code. 

\subsection{Protocol for universal fault-tolerant quantum computing}
\label{sec:universal}

\begin{figure}[tbp]
\centering
\includegraphics[width=\linewidth]{\figpath/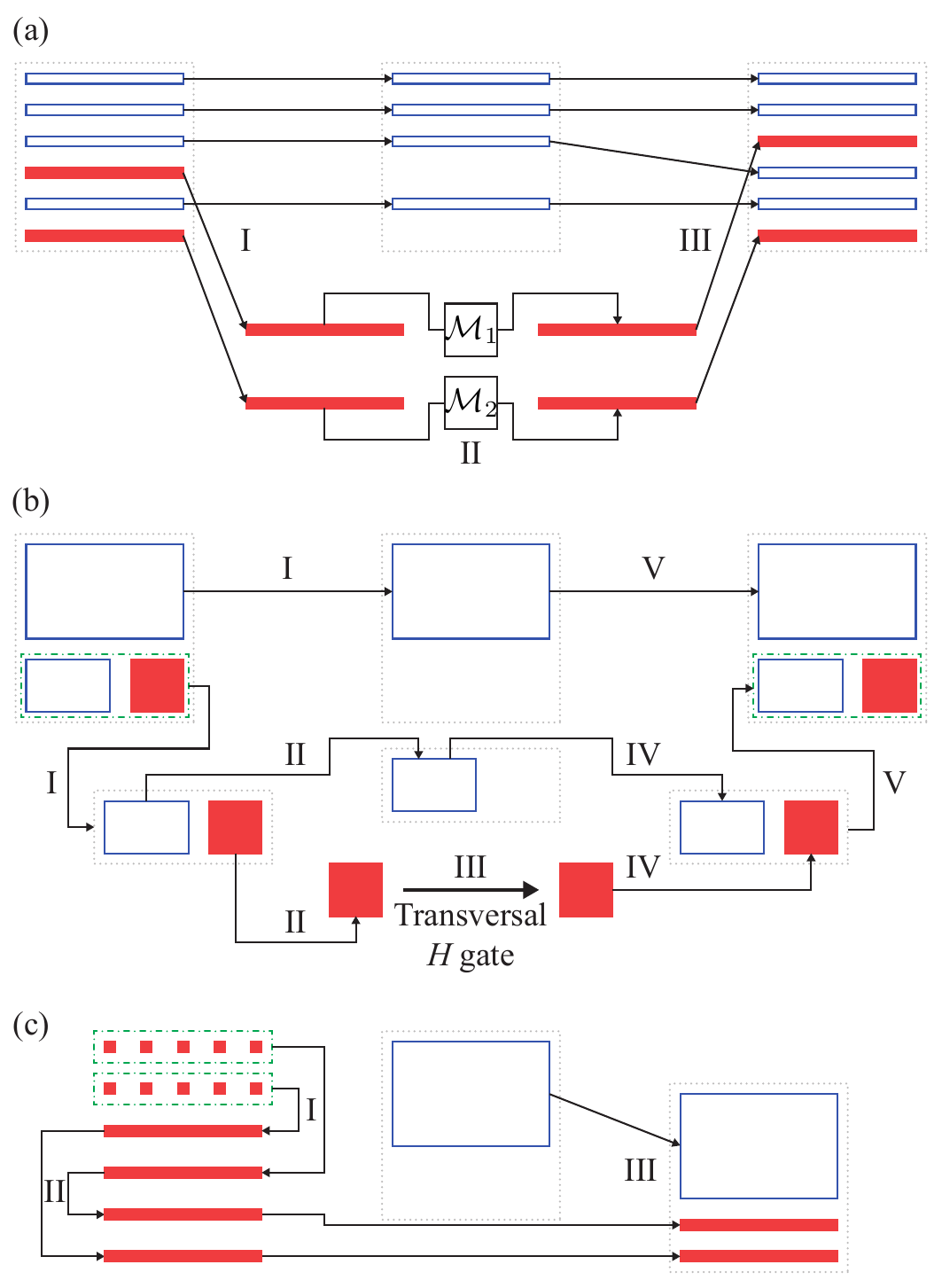}
\caption{
The protocol of universal fault-tolerant quantum computing. Each red filled rectangle represents an array of logical qubits in the active mode, and logical operations are applied on these active qubits. Each blue empty rectangle represents an array of logical qubits in the idle mode. Rectangles in the same dotted box are encoded in the same block. 
}
\label{fig:universal_operations}
\end{figure}

Our protocol includes four operations: The row and column operations illustrated in Fig. \ref{fig:primitive_operations}, logical Hadamard gate and magic state preparation. This operation set is sufficient for universal quantum computing \cite{Fowler2012}. In our protocol, the logical Hadamard gate and magic state preparation are also realised through the row and column operations. 

In our protocol, we suppose that logical qubits are stored in a block encoded in the code $\mathrm{HGP}(\mathbf{H}_0,\mathbf{H}_0)$, which serves as the memory. In addition to the memory, ancilla blocks are used for implementing logical operations on the memory. The ancilla blocks are also encoded in hypergraph product codes, constructed from check matrices $\mathbf{H}_0$, $\mathbf{H}_3$ and $\mathbf{R}_d$. Here, $\mathbf{H}_3$ is used for implementing the logical Hadamard gate, and $\mathbf{R}_d$ is the check matrix of the repetition code with the distance $d = d_R,d_R'$. The code of $\mathbf{H}_i$ has parameters $[n_i,k_i,d_i]$. The distances $d_0$, $d_3$ and $d_R$ are comparable, but $d_R'$ could be small. The protocol is illustrated in Fig. \ref{fig:universal_operations}. 

{\bf Individual-addressing CSS-type logical operations.} The protocol is shown in Fig. \ref{fig:universal_operations}(a). We use the row and column operations to transfer logical qubits between blocks. To address logical qubits individually, we transfer logical qubits in some rows from the memory to $\mathrm{HGP}(\mathbf{R}_{d_R},\mathbf{H}_0)$ blocks in a row operation. Notice that each $\mathrm{HGP}(\mathbf{R}_{d_R},\mathbf{H}_0)$ block stores only one row of logical qubits. Then, we implement CSS-type circuits on the $\mathrm{HGP}(\mathbf{R}_{d_R},\mathbf{H}_0)$ blocks using column operations. Because the $\mathrm{HGP}(\mathbf{R}_{d_R},\mathbf{H}_0)$ blocks are independent, we can choose to apply different circuits to them. Finally, we transfer the logical qubits in $\mathrm{HGP}(\mathbf{R}_{d_R},\mathbf{H}_0)$ blocks back to the memory in a row operation. In a similar way, we can operate logical qubits in columns. 

Suppose the number of $\mathrm{HGP}(\mathbf{R}_{d_R},\mathbf{H}_0)$ blocks is $2L$. Then, we can operate $L$ rows in the memory simultaneously. Let $l_1,l_2,\ldots,l_L$ be the label of the rows. On the row-$l_i$, we want to apply a CSS-type circuit $\mathcal{M}_i$. In step-I, we take $\mathbf{H}_{in} = \mathbf{H}_2 = \mathbf{H}_0$ and $\mathbf{H}_{out} = \mathbf{H}_0\oplus L\mathbf{R}_{d_R}$ in the row operation. Then, the entire output block $\mathrm{HGP}(\mathbf{H}_{out},\mathbf{H}_2)$ consists of $L+1$ independent blocks, including one $\mathrm{HGP}(\mathbf{H}_0,\mathbf{H}_0)$ block and $L$ blocks of $\mathrm{HGP}(\mathbf{R}_{d_R},\mathbf{H}_0)$. We construct the bridge code such that the Clifford gate is the identity gate applied on all input logical qubits (i.e. $m_M = 0$ and $\mathbf{M}_X = \openone_{k_0}$), and $L$ rows in the output block $\mathrm{HGP}(\mathbf{H}_{out},\mathbf{H}_2)$ are initialised. By permuting rows in $\mathbf{K}_{out}$, we can choose to transfer each row in the input $\mathrm{HGP}(\mathbf{H}_0,\mathbf{H}_0)$ block with one of labels $l_1,l_2,\ldots,l_L$ to an output $\mathrm{HGP}(\mathbf{R}_{d_R},\mathbf{H}_0)$ block, and we can choose to initialise rows with labels $l_1,l_2,\ldots,l_L$ in the output $\mathrm{HGP}(\mathbf{H}_0,\mathbf{H}_0)$  block. At this stage, we have used $L$ blocks of $\mathrm{HGP}(\mathbf{R}_{d_R},\mathbf{H}_0)$. In step-II, we use the other $L$ blocks of $\mathrm{HGP}(\mathbf{R}_{d_R},\mathbf{H}_0)$ to realise logical operations. In the column operations, we allocate a fresh $\mathrm{HGP}(\mathbf{R}_{d_R},\mathbf{H}_0)$ block (output block) to each $\mathrm{HGP}(\mathbf{R}_{d_R},\mathbf{H}_0)$ block carrying a row from the memory (input block). For each pair of $\mathrm{HGP}(\mathbf{R}_{d_R},\mathbf{H}_0)$ blocks, we apply the column operation taking $\mathbf{H}_1 = \mathbf{H}_0$ and $\mathbf{H}_{in}' = \mathbf{H}_{out}' = \mathbf{R}_{d_R}$. In the column operation applied on row-$l_i$, we construct the bridge code according to the desired circuit $\mathcal{M}_i$. In step-III, we take $\mathbf{H}_{in} = \mathbf{H}_0\oplus L\mathbf{R}_{d_R}$ and $\mathbf{H}_{out} = \mathbf{H}_2 = \mathbf{H}_0$ in the row operation, i.e. the entire input block $\mathrm{HGP}(\mathbf{H}_{in},\mathbf{H}_2)$ includes one $\mathrm{HGP}(\mathbf{H}_0,\mathbf{H}_0)$ block and $L$ blocks of $\mathrm{HGP}(\mathbf{R}_{d_R},\mathbf{H}_0)$. We construct the bridge code such that rows with labels $l_1,l_2,\ldots,l_L$ in the input $\mathrm{HGP}(\mathbf{H}_0,\mathbf{H}_0)$ block are measured (which are just initialised in step-I), and all other input rows are transferred to the output $\mathrm{HGP}(\mathbf{H}_0,\mathbf{H}_0)$ block. In this way, we realise independent logical operations on rows $l_1,l_2,\ldots,l_L$. By exchanging the roles of rows and columns, we can also realise independent logical operations on $L$ columns. 

Now, we consider the topology of the logical qubit network. Using the above protocol, we can implement logical operations on each row or column. In each row or column, logical qubits are all-to-all connected, e.g. we can apply the controlled-NOT gate on any pair of logical qubits. On such a network, we can exchange any pair of logical qubits with three swap gates. 

{\bf Logical Hadamard gate.} The protocol is shown in Fig. \ref{fig:universal_operations}(b). To implement the Hadamard gate on a set of logical qubits, we transfer the logical qubits from the memory to a $\mathrm{HGP}(\mathbf{H}_3,\mathbf{H}_3)$ block using row and column operations. Then, we apply the transversal Hadamard gate on the $\mathrm{HGP}(\mathbf{H}_3,\mathbf{H}_3)$ block. Finally, we transfer the logical qubits from the $\mathrm{HGP}(\mathbf{H}_3,\mathbf{H}_3)$ block back to the memory. 

Before the logical Hadamard gate, we need to transfer the target logical qubits to a sub-array of the memory through swap gates. Notice that we can exchange any pair of logical qubits with three swap gates. The sub-array has the size of $q\times q$, where $q\leq k_3$. In step-I, we transfer the sub-array to an $\mathrm{HGP}(\mathbf{H}_3,\mathbf{H}_0)$ block in a row operation. In step-II, we transfer the sub-array to the $\mathrm{HGP}(\mathbf{H}_3,\mathbf{H}_3)$ block in a column operation. In step-III, the transversal Hadamard gate is applied. In step-IV and step-V, the sub-array is transferred back to the memory. To deal with the case that the number of target logical qubits is not a square number, we may reserve some logical qubits only for filling the square. The maximum number of such reserved logical qubits is $k_3^2-(k_3-1)^2 = 2k_3-1$. 

{\bf Magic-state preparation.} The protocol is shown in Fig. \ref{fig:universal_operations}(c). First, we prepare magic states on the surface code $\mathrm{HGP}(\mathbf{R}_{d_R'},\mathbf{R}_{d_R'})$ according to Ref. \cite{Li2015}. We only need to encode low-fidelity magic states at this stage because we can increase the fidelity in magic state distillation. Therefore, we can take a small code distance, e.g. $d_R' = 3$. Then, we transfer magic states to the memory. 

We use $L'k_0$ blocks of $\mathrm{HGP}(\mathbf{R}_{d_R'},\mathbf{R}_{d_R'})$ and $L'$ blocks for each of $\mathrm{HGP}(\mathbf{R}_{d_R'},\mathbf{H}_0)$ and $\mathrm{HGP}(\mathbf{R}_{d_R},\mathbf{H}_0)$. In step-I, we transfer the prepared magic states from $\mathrm{HGP}(\mathbf{R}_{d_R'},\mathbf{R}_{d_R'})$ blocks to $\mathrm{HGP}(\mathbf{R}_{d_R'},\mathbf{H}_0)$ blocks in a column operation. In step-II, we transfer magic states from $\mathrm{HGP}(\mathbf{R}_{d_R'},\mathbf{H}_0)$ blocks to $\mathrm{HGP}(\mathbf{R}_{d_R},\mathbf{H}_0)$ blocks in a row operation. Then, the magic states are protected by a large code distance $\min\{d_R,d_0\}$. In step-III, we transfer magic states from $\mathrm{HGP}(\mathbf{R}_{d_R},\mathbf{H}_0)$ blocks to the memory. Given CSS-type logical operations and the logical Hamdarad gate, we can distil the magic states to high fidelity and use them to implement $S$ and $T$ gates \cite{Fowler2012}. 

{\bf Resource efficiency.} We compare our protocol to the protocol reported in Ref. \cite{Xu2024}, in which logical gates are realised by coupling logical qubits in the memory to surface-code logical qubits; we call it the surface-ancilla protocol. Both protocols utilise lattice surgery. In the surface-ancilla protocol, only logical qubits satisfying certain conditions can be operated simultaneously, such as logical qubits in different lines in a hypergraph product code. In our protocol, we can operate all logical qubits in parallel, given sufficient physical qubits. Next, we compare the qubit cost. Even if there is a way of operating an arbitrary number of logical qubits simultaneously in the surface-ancilla protocol, our protocol uses fewer qubits. 

Each block of the code $\mathrm{HGP}(\mathbf{H}_1,\mathbf{H}_2)$ uses $n = n_1n_2+r_1r_2$ physical qubits. For simplicity, we only count the $n_1n_2$ physical qubits. When $\mathbf{H}_1$ and $\mathbf{H}_2$ are full rank, $r_1r_2<n_1n_2$, i.e. the simplification causes a difference in the qubit count up to a factor of two. We also assume $d_0 = d_3 = d_R$ in the qubit cost estimation. In the surface-ancilla protocol, if we assume that the surface code has the same code distance as the memory, the qubit cost for operating one logical qubit is at least $2d_0^2$ ($d_0^2$ physical qubits for the surface code and at least $d_0^2$ physical qubits for a code connecting the memory with the surface code). 

In the CSS-type logical operations, we can operate $Lk_0$ logical qubits in parallel with $2n_0^2+2Ld_0n_0$ physical qubits. We can find that the qubit cost depends on the logical qubit number in the parallel operation. For a balance between the qubit and time costs, we take $L \approx n_0/d_0$, i.e. the total qubit number is about $4n_0^2$, four times the qubit number in the memory. In the surface-ancilla protocol, to operate the same number of logical qubits, the cost is at least $n_0^2 + 2Lk_0d_0^2 \approx n_0^2 + 2n_0k_0d_0$. The cost ratio is $1/4 + k_0d_0/(2n_0)$. If we take a proper linear code $\mathbf{H}_0$, the encoding rate $k_0/n_0$ is finite. When the code distance $d_0$ is large such that $3/(2d_0)$ is smaller than the encoding rate, our protocol uses fewer qubits. 

In the logical Hadamard gate, we can operate up to $k_3^2$ logical qubits in parallel with $2n_0^2 + 2n_0n_3 + n_3^2$ physical qubits. Suppose we take $n_3 \approx (\sqrt{3}-1)n_0$, the total qubit number is also about $4n_0^2$. In the surface-ancilla protocol, to operate the same number of logical qubits, the cost is at least $n_0^2 + 2k_3^2d_0^2$. The cost ratio is $1/4 + (2-\sqrt{3})k_3^2d_0^2/n_3^2$. If two codes $\mathrm{HGP}(\mathbf{H}_0,\mathbf{H}_0)$ and $\mathrm{HGP}(\mathbf{H}_3,\mathbf{H}_3)$ have similar encoding rates, the ratio is $\Omega(k_0^2d_0^2/n_0^2)$. 

In the magic state preparation, we can prepare $L'k_0$ logical qubits in magic states in parallel with $2n_0^2 + L'd_0n_0 + L'd_R'n_0 + L'k_0d_R'^2$ physical qubits. Suppose $d_0\gg d_R'$ and neglect $L'd_R'n_0 + L'k_0d_R'^2$, we can take $L' = 2n_0/d_0$ such that the total qubit number is about $4n_0^2$. In the surface-ancilla protocol, to prepare the same number of logical qubits in magic states, the cost is at least $n_0^2 + 2L'k_0d_0^2 \approx n_0^2 + 4n_0k_0d_0$. The cost ratio is $1/4 + k_0d_0/n_0$. }

\section{Conclusions}

In this work, we have devised a method to represent stabiliser circuits using classical LDPC codes. For each circuit, we can generate a Tanner graph following a straightforward approach, with the corresponding code encapsulating correlations and fault tolerance within the circuit. Through the LDPC representation, we can systematically determine the logical operations executed by the circuit and identify all correlations useful for error correction. This allows us to enhance error correction capabilities by utilising the full correlation set and optimising the error correction check matrix. With the error-correction check matrix, we can assess the fault tolerance of the circuit through the circuit code distance. However, the distance defined in the current work does not consider correlated errors. For instance, in a controlled-NOT gate, the occurrence of an error on two qubits with a probability comparable to that of a single-qubit error should be treated as a single error contribution rather than two. Currently, it remains unclear how to adapt the distance definition to accommodate this aspect, and we leave this question for future work. 

In addition to circuit analysis, the LDPC representation serves as a tool for optimising fault-tolerant circuits. As demonstrated in Sec. \ref{sec:construction}, multiple circuits can be generated from the same LDPC code. These circuits implement the same correlations on logical qubits, thereby implying the same logical operations while sharing a common lower bound for their circuit code distances. This affords opportunities for circuit optimisation. For instance, selecting single-vertex paths in the path partition minimises time costs, whereas reducing the number of paths minimises qubit costs. The optimisation objectives extend beyond time and qubit costs, encompassing considerations such as logical error rates and qubit-network topology. Given circuits realising the same logical operations with the same circuit code distance, preference is naturally given to those employing fewer error-prone operations and causing fewer correlated errors. In Sec. \ref{sec:construction}, we have exclusively explored the scenario where template trees are paths in symmetric splitting. Utilising general template trees offers even greater flexibility for circuit optimisation, particularly concerning qubit-network topology. 

An intriguing application of the LDPC representation lies in the creation of novel fault-tolerant circuits through the construction of classical LDPC codes. In this context, a prospective avenue for future exploration involves systematically deriving LDPC codes corresponding to existing fault-tolerant circuits. This endeavour may provide insights into the structural elements responsible for fault tolerance and logical correlations in quantum circuits, as we do in Secs. \ref{sec:example} and \ref{sec:new_circuit}. Remarkably, given the resemblance in distance definitions between stabiliser circuits and CSS codes, as elaborated in Sec. \ref{sec:distance}, methodologies devised for constructing CSS codes, particularly those tailored for quantum LDPC codes, hold promise for the development of fault-tolerant quantum circuits. 

\begin{acknowledgments}
This work is supported by the National Natural Science Foundation of China (Grant Nos. 12225507, 12088101) and NSAF (Grant No. U1930403).
\end{acknowledgments}

\appendix

\RED{\section{Tanner graphs of general Pauli measurements}
\label{app:measurements}

\begin{figure}[tbp]
\centering
\includegraphics[width=\linewidth]{\figpath/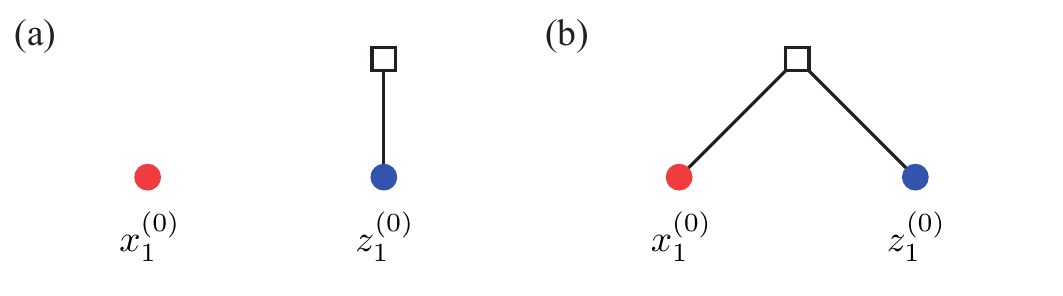}
\caption{
(a) Tanner graph for the measurement of the $X$ operator. (b) Tanner graph for the measurement of the $Y$ operator. 
}
\label{fig:MXY}
\end{figure}

The Tanner graphs representing $X$ and $Y$ measurements are shown in Fig. \ref{fig:MXY}. The principle behind constructing these Tanner graphs is that their codewords are preserved quantities. For example, in the case of the $X$ measurement, the $X$ operator is preserved, but $Y$ and $Z$ are not: regardless of the input state, the mean values of $Y$ and $Z$ are always zero after the measurement, while the mean value of $X$ remains unchanged. It can be verified that the Tanner graph of the $X$ ($Y$) measurement has only one codeword, corresponding to the $X$ ($Y$) operator. Based on this principle, we can also construct Tanner graphs that represent multi-qubit measurements. An example is shown in Fig. \ref{fig:Mgeneral}(a). By combining Tanner graphs of measurements, we can generate Tanner graphs representing a sequence of non-commuting measurements, as shown in Fig. \ref{fig:Mgeneral}(b). 

\begin{figure}[tbp]
\centering
\includegraphics[width=\linewidth]{\figpath/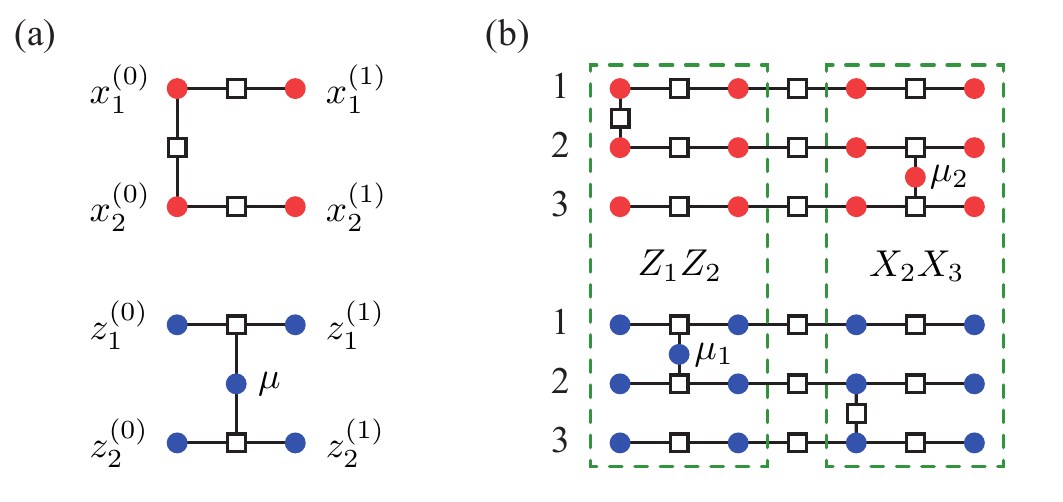}
\caption{
(a) Tanner graph for the measurement of $Z_1Z_2$ operator. The additional bit vertex corresponds to the measurement outcome $\mu$. (b) Tanner graph for a sequence of two anti-commuting measurements $Z_1Z_2$ and $X_2X_3$. The two additional bit vertices correspond to outcomes $\mu_1$ and $\mu_2$ of the two measurements, respectively. 
}
\label{fig:Mgeneral}
\end{figure}

An alternative way of constructing Tanner graphs representing general Pauli measurements is through combining primitive Tanner graphs. The $X$ measurement is equivalent to an $H$ gate followed by a $Z$ measurement, with its Tanner graph shown in Fig. \ref{fig:MXY_combined}(a). The two Tanner graphs of the $X$ measurement are equivalent. We can transform the graph in Fig. \ref{fig:MXY_combined}(a) to the graph in Fig. \ref{fig:MXY}(a) by deleting vertices in purple dotted ovals. Such an operation is the inverse operation of the bit splitting (see Definition \ref{def:bit_splitting}), which preserves properties of the Tanner graph in representing stabiliser circuits (see Lemma \ref{lem:bit_splitting}). 

\begin{figure}[tbp]
\centering
\includegraphics[width=\linewidth]{\figpath/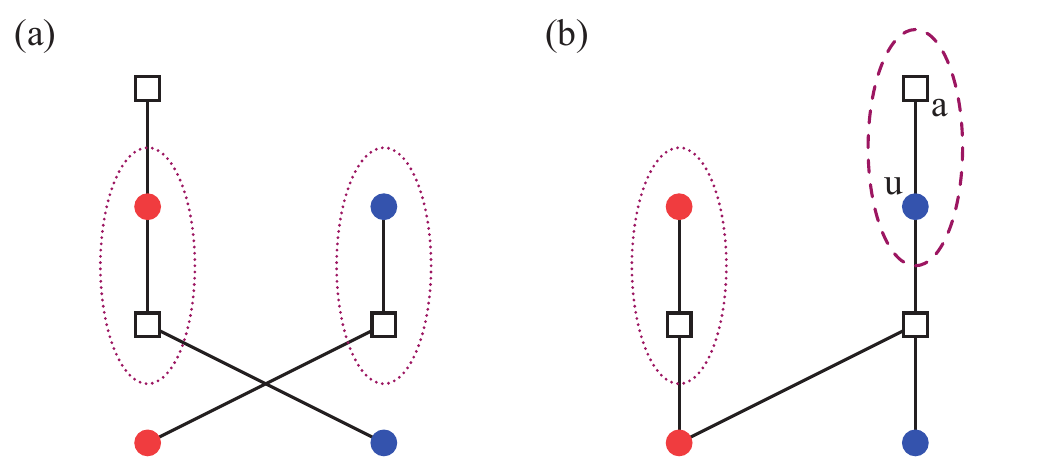}
\caption{
(a) Tanner graph for the measurement of the $X$ operator constructed by combining an $H$ gate and a $Z$ measurement. (b) Tanner graph for the measurement of the $Y$ operator constructed by combining an $S$ gate and an $X$ measurement. 
}
\label{fig:MXY_combined}
\end{figure}

Similarly, the $Y$ measurement is equivalent to an $S$ gate followed by an $X$ measurement, with its Tanner graph shown in Fig. \ref{fig:MXY_combined}(b). The two Tanner graphs of the $Y$ measurement are equivalent. We can transform the graph in Fig. \ref{fig:MXY_combined}(b) to the graph in Fig. \ref{fig:MXY}(b) by deleting vertices in purple dotted and dashed ovals. The dotted oval corresponds to the inverse bit splitting. Vertices in the dashed oval are redundant: because of the check vertex $a$, the bit vertex $u$ always takes the value of zero; therefore, removing the two vertices changes nothing. }

\section{Proof of Theorem \ref{the:errorfree}}
\label{app:errorfree}

According to Proposition \ref{prop:gate}, the error-free codeword equation holds when the circuit consists of only Clifford gates. For a general Clifford circuit, we prove the equation by constructing an equivalent circuit, called extended circuit. In the extended circuit, all initialisation operations are in the first layer, all measurement operations are in the last layer, and only Clifford gates are applied in other layers. The equivalence is justified in two aspects: i) by tracing out ancilla qubits, the two circuits realise the same transformation, and ii) their Tanner graphs are equivalent in the sense that they are related by bit splitting (see Definition \ref{def:bit_splitting} and Lemma \ref{lem:bit_splitting}). 

We introduce some notations first and then construct the extended circuit to complete the proof. 

\subsection{Lists of initialisation and measurement operations}

Similar to the bit set of measurements $V_M$, we use a bit set $V_I\subseteq V_B$ to denote initialisations (all in the $Z$ basis): $\hat{z}_q^{(t)}\in V_I$ if an initialisation is performed on qubit-$q$ at the time $t$. The number of initialisations is $n_I = \vert V_I\vert$. 

We can group initialisations into two subsets $V_{I,S}$ and $V_{I,P}$. Initialisations in $V_{I,S}$ are the first operations, i.e. the first operation on qubit-$q$ is an initialisation at the time $t$ if $\hat{z}_q^{(t)}\in V_{I,S}$. All other initialisations are in $V_{I,P}$. 

Similarly, we can group measurements into two subsets $V_{M,S}$ and $V_{M,P}$. Measurements in $V_{M,S}$ are the last operations, i.e. the last operation on qubit-$q$ is a measurement at the time $t$ if $\hat{z}_q^{(t)}\in V_{I,S}$. All other measurements are in $V_{M,P}$. 

Because all measured qubits are reinitialised before reusing, initialisations in $V_{I,P}$ and measurements in $V_{M,P}$ occur pairwisely in the circuit. If $\hat{z}_q^{(t-1)}\in V_{M,P}$, there exists $\hat{z}_q^{(s)}\in V_{I,P}$ such that after the measurement $\hat{z}_q^{(t-1)}$, the initialisation $\hat{z}_q^{(s)}$ is the subsequent operation on qubit-$q$ ($s>t$); and vice versa. The number of measurement-initialisation pairs is $n_P = \vert V_{I,P}\vert = \vert V_{M,P}\vert$. We give each pair a label $l$, then we use $v^M_l=\hat{z}_q^{(t-1)}\in V_{M,P}$ and $v^I_l=\hat{z}_q^{(s)}\in V_{I,P}$ to denote the measurement and initialisation of pair-$l$. 

\subsection{Extended circuit}

\begin{figure}[tbp]
\centering
\includegraphics[width=\linewidth]{\figpath/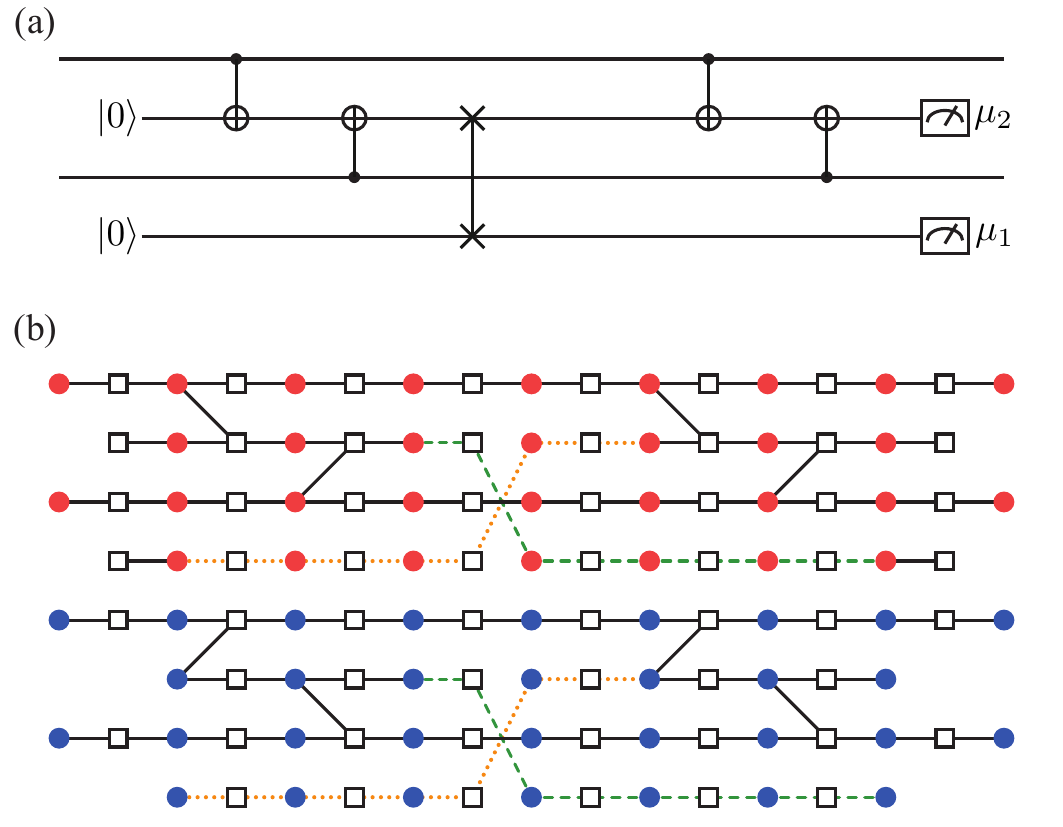}
\caption{
(a) Extended circuit of the circuit in Fig. \ref{fig:ZZ}(a). (b) Tanner graph of the extended circuit. This graph is generated by applying bit splitting operations to the original Tanner graph in Fig. \ref{fig:ZZ}(b). The bit-splitting paths are denoted by green dashed lines and orange dotted lines. 
}
\label{fig:ZZextended}
\end{figure}

In the extended circuit, we move all initialisations to the first layer and all measurements to the last layer. This is achieved by introducing a fresh ancilla qubit for each measurement-initialisation pair. Given a stabiliser circuit, its {\it extended circuit} is constructed in the following way [see Fig. \ref{fig:ZZextended}(a)]: 
\begin{itemize}
\item[1.] For all Clifford gates in the original circuit, apply the same gates on the same qubits in the same layers in the extended circuit; 
\item[2.] For all initialisations in $V_{I,S}$, apply initialisations on the same qubits in the same bases in the extended circuit, however, they are all moved to the first layer; 
\item[3.] For all measurements in $V_{M,S}$, apply measurements on the same qubits in the same bases in the extended circuit, however, they are all moved to the last layer in the extended circuit; 
\item[4.] Each measurement-initialisation pair in the original circuit is mapped to three operations in the extended circuit. Consider the pair-$l$ consisting of the measurement $v^M_l = \hat{z}_q^{(t-1)}$ and initialisation $v^I_l = \hat{z}_q^{(s)}$. The first operation is an initialisation on qubit-$(n+l)$ in the first layer. The second operation is a swap gate between qubit-$q$ and qubit-$(n+l)$ in layer-$t$. The third operation is a measurement on qubit-$(n+l)$ in the last layer. 
\end{itemize}

Without loss of generality, we suppose that in the original circuit, all non-identity Clifford gates are applied in layers $t = 2,3,\ldots,T-1$. Then, we have such an extended circuit: All initialisations are applied in the first layer, Clifford gates are applied in layers $t = 2,3,\ldots,T-1$, which form an $(n+n_P)$-qubit Clifford gate $\overline{U}$, and all measurements are applied in the last layer. From now on, notations with the overline are relevant to the extended circuit. 

Now, let's consider the Tanner graph of the extended circuit. Let $Q_I$ and $Q_M$ be two subsets of qubits: In the extended circuit, qubits in $Q_I$ are initialised in the first layer, and qubits in $Q_M$ are measured in the last layer. The bit set of the extended Tanner graph is 
\begin{eqnarray}
\overline{V}_B = \overline{V}_{B,0}'\cup \overline{V}_{B,1}\cup\cdots\cup\overline{V}_{B,T-1}\cup \overline{V}_{B,T}',
\end{eqnarray}
where $\overline{V}_{B,0}' = \{\hat{x}_q^{(t)},\hat{z}_q^{(t)}\in \overline{V}_{B,0}\vert q\notin Q_I\}$, $\overline{V}_{B,T}' = \{\hat{x}_q^{(t)},\hat{z}_q^{(t)}\in \overline{V}_{B,T}\vert q\notin Q_M\}$ and $\overline{V}_{B,t} = \{\hat{x}_q^{(t)},\hat{z}_q^{(t)}\vert q = 1,2,\ldots,n+n_P\}$. On the graph, the bit number is $2(n+n_P-n_I)$ in the first layer, the bit number is $2(n+n_P-n_M)$ in the last layer, and the bit number is $2(n+n_P)$ in each layer with $t = 1,2,\ldots,T-1$. See Fig. \ref{fig:ZZextended}(b). 

We can find that the extended Tanner graph can be generated from the original Tanner graph through a sequence of bit splitting operations. These bit splitting operations are applied along some paths. For each initialisation $\hat{z}_q^{(t)}\in V_{I,S}$, there are two paths connecting $\hat{x}_q^{(t)}$ and $\hat{z}_q^{(t)}$ with $\hat{x}_q^{(1)}$ and $\hat{z}_q^{(1)}$, respectively. For each measurement $\hat{z}_q^{(t-1)}\in V_{M,S}$, there are two paths connecting $\hat{x}_q^{(t-1)}$ and $\hat{z}_q^{(t-1)}$ with $\hat{x}_q^{(T-1)}$ and $\hat{z}_q^{(T-1)}$, respectively. There are four paths for each measurement-initialisation pair $v^M_l = \hat{z}_q^{(t-1)}$ and $v^I_l = \hat{z}_q^{(s)}$, which connect $\hat{x}_q^{(t-1)}$, $\hat{z}_q^{(t-1)}$, $\hat{x}_q^{(s)}$ and $\hat{z}_q^{(s)}$ with $\hat{x}_{n+l}^{(T-1)}$, $\hat{z}_{n+l}^{(T-1)}$, $\hat{x}_{n+l}^{(1)}$ and $\hat{z}_{n+l}^{(1)}$, respectively. These paths correspond to identity gates and swap gates in the extended circuit. 

Due to the bit splitting operations, two codes are isomorphic according to Lemma \ref{lem:bit_splitting}. Let $\mathbf{c}$ be a codeword of the original Tanner graph. Then, $\overline{\mathbf{c}} = \varphi(\mathbf{c})$ is a codeword of the extended Tanner graph. Here, $\varphi(\bullet)$ is a linear bijection yielded by composing $\phi$ functions: $\varphi(\mathbf{c})_v = \mathbf{c}_v$ if $v\in V_B$ is a bit on the original Tanner graph, and $\varphi(\mathbf{c})_u = \mathbf{c}_v$ if $u\notin V_B$ and $v$ are on the same bit-splitting path. Notice that bits on the same bit-splitting path are coupled by degree-2 checks; therefore, they always take the same value in a codeword. 

\subsection{Proof of the error-free codeword equation}

\begin{lemma}
Let $\overline{\mathbf{P}}_t$ be layer projections acting on the extended Tanner graph, and let $\sigma(\overline{\mathbf{b}})$ be Pauli operators on the extended qubit set when $\overline{\mathbf{b}}\in \mathbb{F}_2^{2(n+n_P)}$. For all codewords $\mathbf{c}$ of the original circuit and their corresponding codewords $\overline{\mathbf{c}} = \varphi(\mathbf{c})$ of the extended circuit, the following two equations hold: 
\begin{eqnarray}
\sigma(\overline{\mathbf{P}}_1\overline{\mathbf{c}}^\mathrm{T}) &=& \left[\sigma(\mathbf{P}_0\mathbf{c}^\mathrm{T})\otimes I^{\otimes n_P}\right] \notag \\
&&\times \prod_{u=\hat{z}_q^{(t)}\in V_{I,S}}Z_q^{\mathbf{c}_u} \prod_{v=v^I_l\in V_{I,P}}Z_{n+l}^{\mathbf{c}_v}, \\
\sigma(\overline{\mathbf{P}}_{T-1}\overline{\mathbf{c}}^\mathrm{T}) &=& \left[\sigma(\mathbf{P}_T\mathbf{c}^\mathrm{T})\otimes I^{\otimes n_P}\right] \notag \\
&&\times \prod_{u=\hat{z}_q^{(t-1)}\in V_{M,S}}Z_q^{\mathbf{c}_u} \prod_{v=v^M_l\in V_{M,P}}Z_{n+l}^{\mathbf{c}_v}.~~~~~
\end{eqnarray}
\label{lem:operators}
\end{lemma}

\begin{proof}
The equation of $\sigma(\overline{\mathbf{P}}_1\overline{\mathbf{c}}^\mathrm{T})$ is justified by the following observations: 
\begin{itemize}
\item[i)] If $q\notin Q_I$, then $\hat{x}_q^{(1)},\hat{z}_q^{(1)}\in V_B$ are bits on the original Tanner graph; 
\item[ii)] If $q\in Q_I$, then $\hat{x}_q^{(1)},\hat{z}_q^{(1)}\notin V_B$, $\hat{z}_q^{(1)}$ is on the bit-splitting path connecting to $u=\hat{z}_q^{(t)}\in V_{I,S}$ or $v=v^I_l\in V_{I,P}$, and $\hat{x}_q^{(1)}$ is on the bit-splitting path connecting to the corresponding $x$ bit; 
\item[iii)] If $q\notin Q_I$, then $x_q^{(1)} = x_q^{(0)}$ and $x_z^{(1)} = x_z^{(0)}$ in $\mathbf{c}$, because there is not any non-identity Clifford gate applied in the first layer; 
\item[iv)] If $q\in Q_I$ and $q\leq n$, then $x_q^{(0)} = z_q^{(0)} = 0$ in $\mathbf{P}_0\mathbf{c}^\mathrm{T} = (\mathbf{x}^{(0)},\mathbf{z}^{(0)})^\mathrm{T}$ (because $\hat{x}_q^{(0)},\hat{z}_q^{(0)}\notin V_B$); 
\item[v)] If $q\in Q_I$, $x_q^{(1)} = 0$ in $\overline{\mathbf{c}}$. 
\end{itemize}

Because of observations i) and iii), $\mathbf{P}_0\mathbf{c}^\mathrm{T}$ determines the local Pauli operators acting on qubits not in $Q_I$. Because of the observation iv), $\sigma(\mathbf{P}_0\mathbf{c}^\mathrm{T})\otimes I^{\otimes n_P}$ acts trivially on qubits in $Q_I$. Because of observations ii) and v), there are only $I$ and $Z$ operators acting on qubits in $Q_I$ determined by $\mathbf{c}_u$ and $\mathbf{c}_v$, respectively. 

For the equation of $\sigma(\overline{\mathbf{P}}_{T-1}\overline{\mathbf{c}}^\mathrm{T})$, the proof is similar. 
\end{proof}

\begin{lemma}
For all codewords $\mathbf{c}$ of the original circuit and their corresponding codewords $\overline{\mathbf{c}} = \varphi(\mathbf{c})$ of the extended circuit, $[\overline{U}]\sigma(\overline{\mathbf{P}}_1\overline{\mathbf{c}}^\mathrm{T}) = \nu\left(\mathbf{c}\right) \sigma(\overline{\mathbf{P}}_{T-1}\overline{\mathbf{c}}^\mathrm{T})$, where $\nu\left(\mathbf{c}\right) = \eta\left([\overline{U}]\sigma(\overline{\mathbf{P}}_1\overline{\mathbf{c}}^\mathrm{T})\right)$. 
\label{lem:evolution}
\end{lemma}

\begin{proof}
Let $\overline{U}_t$ be the $(n+n_P)$-qubit Clifford gate formed of Clifford gates in layer-$t$ in the extended circuit. Then, $\overline{U} = \overline{U}_{T-1}\cdots\overline{U}_3\overline{U}_2$ and $\mathbf{M}_{\overline{U}} = \mathbf{M}_{\overline{U}_{T-1}}\cdots\mathbf{M}_{\overline{U}_3}\mathbf{M}_{\overline{U}_2}$. 

The codeword is in the form $\overline{\mathbf{c}} = \left(\overline{\mathbf{b}}_0,\overline{\mathbf{b}}_1,\ldots,\overline{\mathbf{b}}_T\right)$, where $\overline{\mathbf{b}}_t$ represents values of bits in layer-$t$. When $t = 1,2,\ldots,T-1$, $\overline{\mathbf{P}}_t\overline{\mathbf{c}}^\mathrm{T} = \overline{\mathbf{b}}_t^\mathrm{T}$. 

Because $\overline{\mathbf{c}}$ is a codeword, $\overline{\mathbf{b}}_t^\mathrm{T} = \mathbf{M}_{\overline{U}_t}\overline{\mathbf{b}}_{t-1}^\mathrm{T}$ when $t = 2,3,\ldots,T-1$. Notice that all initialisation and measurement operations are in the first and last layers of the extended circuit. Then, $\overline{\mathbf{b}}_{T-1}^\mathrm{T} = \mathbf{M}_{\overline{U}}\overline{\mathbf{b}}_{1}^\mathrm{T}$, i.e. $\mathbf{A}_{\overline{U}}\left(\begin{matrix}
\overline{\mathbf{b}}_1 & \overline{\mathbf{b}}_{T-1}
\end{matrix}\right)^\mathrm{T} = 0$. According to Proposition \ref{prop:gate}, the lemma is proved. 
\end{proof}

Let $\rho_0$ be the initial state of the original circuit. If we replace the initial state with $\tilde{\rho}_0 = \mathrm{Tr}_{V_{I,S}}(\rho_0)\otimes \bigotimes_{\hat{z}_q^{(t)}\in V_{I,S}}\vert 0\rangle\langle 0\vert_q$, the output state is the same. Here, $\mathrm{Tr}_{V_{I,S}}$ denotes the partial trace on qubits $\{k\vert \hat{z}_q^{(t)}\in V_{I,S}\}$, and $\vert 0\rangle\langle 0\vert_q$ is the state of qubit-$q$. Similarly, without loss of generality, we can assume that the initial state of the extended circuit is $\overline{\rho}_0 = \tilde{\rho}_0\otimes \bigotimes_{l=1}^{n_P}\vert 0\rangle\langle 0\vert_{n+l}$. 

In the extended circuit, the measurement $v=\hat{z}_q^{(t-1)}\in V_{M,S}$ is described by the projection $\pi_v(\mu_v) = \frac{1+\mu_v Z_q}{2}$, and the measurement $v = v^M_l\in V_{M,P}$ is described by the projection $\pi_v(\mu_v) = \frac{1+\mu_v Z_{n+l}}{2}$. 

Given the initial state $\overline{\rho}_0$ of the extended circuit, the final state is $\overline{\rho}_T(\boldsymbol{\mu}) = [E_{\boldsymbol{\mu}}][\overline{U}]\overline{\rho}_0$. Here, $E_{\boldsymbol{\mu}} = \prod_{v\in V_M}\pi_v(\mu_v)$ and $\sum_{\boldsymbol{\mu}}E_{\boldsymbol{\mu}}^\dag E_{\boldsymbol{\mu}} = \openone$. Then, the final state of the original circuit is $\rho_T(\boldsymbol{\mu}) = \mathrm{Tr}_{n_P}\left(\overline{\rho}_T(\boldsymbol{\mu})\right)$, where $\mathrm{Tr}_{n_P}$ denotes the partial trace on qubits $n+1,n+2,\ldots,n+n_P$. 

Now, we can prove the codeword equation. According to Lemma \ref{lem:operators}, we have $\mathrm{Tr}\sigma_{in}(\mathbf{c})\rho_0
= \mathrm{Tr}\sigma(\overline{\mathbf{P}}_1\overline{\mathbf{c}}^\mathrm{T})\overline{\rho}_0$ and $\sum_{\boldsymbol{\mu}} \mu_R(\mathbf{c},\boldsymbol{\mu})\mathrm{Tr}\sigma_{out}(\mathbf{c})\rho_T(\boldsymbol{\mu}) = \sum_{\boldsymbol{\mu}} \mathrm{Tr}\sigma(\overline{\mathbf{P}}_{T-1}\overline{\mathbf{c}}^\mathrm{T})\overline{\rho}_T(\boldsymbol{\mu})$. Because $\sigma(\overline{\mathbf{P}}_{T-1}\overline{\mathbf{c}}^\mathrm{T})$ and $E_{\boldsymbol{\mu}}$ are commutative, we have $\sum_{\boldsymbol{\mu}} \mathrm{Tr}\sigma(\overline{\mathbf{P}}_{T-1}\overline{\mathbf{c}}^\mathrm{T})\overline{\rho}_T(\boldsymbol{\mu}) = \mathrm{Tr}\sigma(\overline{\mathbf{P}}_{T-1}\overline{\mathbf{c}}^\mathrm{T})[\overline{U}]\overline{\rho}_0$. According to Lemma \ref{lem:evolution}, $\mathrm{Tr}\sigma(\overline{\mathbf{P}}_1\overline{\mathbf{c}}^\mathrm{T})\overline{\rho}_0 = \nu\left(\mathbf{c}\right) \mathrm{Tr}\sigma(\overline{\mathbf{P}}_{T-1}\overline{\mathbf{c}}^\mathrm{T})[\overline{U}]\overline{\rho}_0$. The codeword equation is proved. 

\section{Proofs of Theorem \ref{the:propagator} and Corollary \ref{coro:propagator}}
\label{app:propagator}

%\begin{lemma}
%For all codewords $\mathbf{c}$ of the original circuit and their corresponding codewords $\overline{\mathbf{c}} = \varphi(\mathbf{c})$ of the extended circuit, $[\sigma(\overline{\mathbf{P}}_1\overline{\mathbf{c}}^\mathrm{T}),Z_q] = 0$ holds for all $q\in Q_I$, and $[\sigma(\overline{\mathbf{P}}_{T-1}\overline{\mathbf{c}}^\mathrm{T}),Z_q] = 0$ for all $q\in Q_M$. 
%\label{lem:commutative1}
%\end{lemma}

%\begin{proof}
%According to the proof of Lemma \ref{lem:operators}, in $\sigma(\overline{\mathbf{P}}_1\overline{\mathbf{c}}^\mathrm{T})$, there are only $Z$ and $I$ operators acting on qubits in $Q_I$. It is similar for $\sigma(\overline{\mathbf{P}}_{T-1}\overline{\mathbf{c}}^\mathrm{T})$. 
%\end{proof}

\begin{lemma}
Let $S\in P_{n+n_P}$ be a Pauli operator on the extended qubit set. Suppose it satisfies the following conditions: i) $[S,Z_q] = 0$ holds for all $q\in Q_I$; and ii) $[[\overline{U}]S,Z_q] = 0$ holds for all $q\in Q_M$. Then, there exits a codeword $\mathbf{c}$ of the original circuit such that $\sigma(\overline{\mathbf{P}}_1\overline{\mathbf{c}}^\mathrm{T}) = S$ and $\sigma(\overline{\mathbf{P}}_{T-1}\overline{\mathbf{c}}^\mathrm{T}) = \eta\left([\overline{U}]S\right)\times[\overline{U}]S$, where $\overline{\mathbf{c}} = \varphi(\mathbf{c})$ is the corresponding codeword of the extended circuit. 
\label{lem:commutative2}
\end{lemma}

\begin{proof}
There exists $\overline{\mathbf{b}}_1 = (\mathbf{x}^{(1)},\mathbf{z}^{(1)}) \in \mathbb{F}_2^{2(n+n_P)}$ satisfying $\sigma(\overline{\mathbf{b}}_1) = S$. Because of the condition i), $x_q^{(1)} = 0$ for all $q\in Q_I$. Let $\overline{\mathbf{b}}_t = \overline{\mathbf{b}}_{t-1}\mathbf{M}_{\overline{U}_t}^\mathrm{T}$ for $t = 2,3,\ldots,T-1$; see proof of Lemma \ref{lem:evolution}. Then, $\sigma(\overline{\mathbf{b}}_{T-1}) = \eta\left([\overline{U}]S\right)\times[\overline{U}]S$ according to Proposition \ref{prop:gate}. Because of the condition ii), $\overline{\mathbf{b}}_{T-1}$ satisfies $x_q^{(T-1)} = 0$ for all $q\in Q_M$. 

On the Tanner graph of the extended circuit, bit numbers of layers $t = 0,T$ are smaller than other layers. In layer-$0$, the bit number is $2(n+n_P-\vert Q_I\vert)$; and in layer-$T$, the bit number is $2(n+n_P-\vert Q_M\vert)$. Let's introduce the following two vectors $\overline{\mathbf{b}}_0\in \mathbb{F}_2^{2(n+n_P-\vert Q_I\vert)}$ and $\overline{\mathbf{b}}_T\in \mathbb{F}_2^{2(n+n_P-\vert Q_M\vert)}$: Entries $x_q^{(0)},z_q^{(0)}$ in $\overline{\mathbf{b}}_0$ take values the same as entries $x_q^{(1)},z_q^{(1)}$ in $\overline{\mathbf{b}}_1$, respectively; and entries $x_q^{(T)},z_q^{(T)}$ in $\overline{\mathbf{b}}_T$ take values the same as entries $x_q^{(T-1)},z_q^{(T-1)}$ in $\overline{\mathbf{b}}_{T-1}$, respectively. Then, $\overline{\mathbf{c}} = \left(\overline{\mathbf{b}}_0,\overline{\mathbf{b}}_1,\ldots,\overline{\mathbf{b}}_T\right)$ is a codeword of the extended circuit, and $\overline{\mathbf{b}}_t^\mathrm{T}=\overline{\mathbf{P}}_t\overline{\mathbf{c}}^\mathrm{T}$ when $t = 1,2,\ldots,T-1$. 

Notice that $\varphi$ is a linear bijection between two codes. The lemma is proved. 
\end{proof}

{\bf `If' part of the theorem.} Let $\mathbf{c}\in C_{c,d,e}$. It can be decomposed as $\mathbf{c} = \mathbf{c}_c+\mathbf{c}_d+\mathbf{c}_e$, where $\mathbf{c}_c\in C_c$, $\mathbf{c}_d\in C_{c,d} - C_c$ and $\mathbf{c}_e\in C_{c,e} - C_c$ are checker, detector and emitter, respectively. Then, $\sigma_{in}(\mathbf{c}) = \sigma_{in}(\mathbf{c}_d)$ and $\sigma_{out}(\mathbf{c}) = \sigma_{out}(\mathbf{c}_e)$. 

Let $\overline{\mathbf{c}}_d = \varphi(\mathbf{c}_d)$ and $\overline{\mathbf{c}}' = \varphi(\mathbf{c}')$ be codewords of the extended circuit corresponding to $\mathbf{c}_d$ and $\mathbf{c}'$, respectively. Because $\sigma_{out}(\mathbf{c}_d) = \openone$, $\sigma(\overline{\mathbf{P}}_{T-1}\overline{\mathbf{c}}_d^\mathrm{T})$ acts trivially on all qubits not in $Q_M$. In $\sigma(\overline{\mathbf{P}}_{T-1}\overline{\mathbf{c}}_d^\mathrm{T})$ and $\sigma(\overline{\mathbf{P}}_{T-1}{\overline{\mathbf{c}}'}^\mathrm{T})$, single-qubit Pauli operators are all commutative on qubits in $Q_M$. Then, $[\sigma(\overline{\mathbf{P}}_{T-1}\overline{\mathbf{c}}_d^\mathrm{T}),\sigma(\overline{\mathbf{P}}_{T-1}{\overline{\mathbf{c}}'}^\mathrm{T})] = 0$. 

Pauli operators at $t = 1$ and $t = T-1$ are related by the superoperator $[\overline{U}]$. Then, we have $[\sigma(\overline{\mathbf{P}}_1\overline{\mathbf{c}}_d^\mathrm{T}),\sigma(\overline{\mathbf{P}}_1{\overline{\mathbf{c}}'}^\mathrm{T})] = 0$. In $\sigma(\overline{\mathbf{P}}_1\overline{\mathbf{c}}_d^\mathrm{T})$ and $\sigma(\overline{\mathbf{P}}_1{\overline{\mathbf{c}}'}^\mathrm{T})$, single-qubit Pauli operators are all commutative on qubits in $Q_I$. Therefore, $[\sigma_{in}(\mathbf{c}_d),\sigma_{in}(\mathbf{c}')] = 0$. 

Similarly, $[\sigma_{out}(\mathbf{c}_e),\sigma_{out}(\mathbf{c}')] = 0$. The `if’ part has been proved. 

{\bf `Only if' part of the theorem.} Let's define a matrix for describing the commutative and anti-commutative relations between Pauli operators. The matrix reads 
\begin{eqnarray}
\mathbf{X} = \left(\begin{matrix}
0 & \openone_n \\
\openone_n & 0
\end{matrix}\right).
\end{eqnarray}
We can find that $[\sigma(\mathbf{b}),\sigma(\mathbf{b}')] = 0$ if and only if $\mathbf{b}'\mathbf{X}\mathbf{b}^\mathrm{T} = 0$, and $\{\sigma(\mathbf{b}),\sigma(\mathbf{b}')\} = 0$ if and only if $\mathbf{b}'\mathbf{X}\mathbf{b}^\mathrm{T} = 1$, for all $\mathbf{b} = \left(\mathbf{x},\mathbf{y}\right)$ and $\mathbf{b}' = \left(\mathbf{x}',\mathbf{y}'\right)$. Here, $\mathbf{x},\mathbf{y},\mathbf{x}',\mathbf{y}'\in \mathbb{F}_2^n$. Similarly, $\overline{\mathbf{X}}$ denotes the matrix on the extended qubit set. 

Let $\mathbf{c}\in C - C_{c,d,e}$ be a genuine propagator. We can prove the `only if' part by constructing a codeword $\mathbf{c}'\in C$ satisfying $\{\sigma_{in}(\mathbf{c}), \sigma_{in}(\mathbf{c}')\} = 0$. 

First, there exists $S\in P_{n+n_P}$ such that i) $\{\sigma(\overline{\mathbf{P}}_1\overline{\mathbf{c}}^\mathrm{T}),S\} = 0$, ii) $[\sigma(\overline{\mathbf{P}}_1\overline{\mathbf{a}}^\mathrm{T}),S] = 0$ for all $\mathbf{a}\in C_{c,d,e}$, and iii) $[Z_p,S] = 0$ for all $p\in Q_I$. We can work out $S$ as follows. Because $\mathbf{c}\notin C_{c,d,e}$, $\mathbf{P}_0\mathbf{c}^\mathrm{T}\notin \{\mathbf{P}_0\mathbf{a}^\mathrm{T}\vert \mathbf{a}\in C_{c,d,e}\}$ (otherwise $\mathbf{P}_0(\mathbf{c}+\mathbf{a})^\mathrm{T} = 0$, i.e. $\mathbf{c}+\mathbf{a}\in C_{c,e}$). Because $\mathbf{X}$ is full rank, $\mathbf{X}\mathbf{P}_0\mathbf{c}^\mathrm{T}\notin \{\mathbf{X}\mathbf{P}_0\mathbf{a}^\mathrm{T}\vert \mathbf{a}\in C_{c,d,e}\}$. Then, there exists $\mathbf{b}\in \mathbb{F}_2^{2n}$ such that $\mathbf{b}\mathbf{X}\mathbf{P}_0\mathbf{c}^\mathrm{T} = 1$ and $\mathbf{b}\mathbf{X}\mathbf{P}_0\mathbf{a}^\mathrm{T} = 0$ for all $\mathbf{a}\in C_{c,d,e}$; notice that $\{\mathbf{X}\mathbf{P}_0\mathbf{a}^\mathrm{T}\vert \mathbf{a}\in C_{c,d,e}\}$ is a subspace. In vectors $\mathbf{P}_0\mathbf{c}^\mathrm{T}$ and $\mathbf{P}_0\mathbf{a}^\mathrm{T}$, entries $x_p^{(0)},z_p^{(0)}$ take the value of zero for all $p\in Q_I$. Let's define the subspace $V = \{\left(\mathbf{x},\mathbf{y}\right)\in\mathbb{F}_2^{2n}\vert x_p^{(0)},z_p^{(0)}=0\forall p\in Q_I\}$. Project $\mathbf{b}$ onto the subspace $V$, we obtain $\mathbf{b}'$. Then, $\mathbf{b}'\mathbf{X}\mathbf{P}_0\mathbf{c}^\mathrm{T} = 1$ and $\mathbf{b}'\mathbf{X}\mathbf{P}_0\mathbf{a}^\mathrm{T} = 0$ for all $\mathbf{a}\in C_{c,d,e}$. We take $S = \sigma(\mathbf{b}')\otimes I^{\otimes n_P}$, which satisfies all the three conditions.  

Let $\overline{\mathbf{b}} = \sigma^{-1}(S)$. We can construct two matrices $\mathbf{Z}_{in}\in \mathbb{F}_2^{\vert Q_I\vert\times(n+n_P)}$ and $\mathbf{Z}_{out}\in \mathbb{F}_2^{\vert Q_M\vert\times(n+n_P)}$: Rows of $\mathbf{Z}_{in}$ and $\mathbf{Z}_{out}$ are $\sigma^{-1}(Z_p)$ and $\sigma^{-1}(Z_q)$, respectively, where $p\in Q_I$ and $q\in Q_M$. Then, let's define
$\mathbf{V} = \mathbf{Z}_{in}\mathbf{M}_{\overline{U}}^\mathrm{T}\overline{\mathbf{X}}\mathbf{Z}_{out}^\mathrm{T}$ and $\mathbf{u} = \overline{\mathbf{b}}\mathbf{M}_{\overline{U}}^\mathrm{T}\overline{\mathbf{X}}\mathbf{Z}_{out}^\mathrm{T}$. Here, $\mathbf{V}_{p,q} = 0$ if and only if $[[\overline{U}]Z_p,Z_q] = 0$, $\mathbf{V}_{p,q} = 1$ if and only if $\{[\overline{U}]Z_p,Z_q\} = 0$; $\mathbf{u}_q = 0$ if and only if $[[\overline{U}]S,Z_q] = 0$, and $\mathbf{u}_q = 1$ if and only if $\{[\overline{U}]S,Z_q\} = 0$. 

Now, we can prove that $\mathbf{u}\in \mathrm{rowsp}(\mathbf{V})$. Suppose $\mathbf{u}\notin \mathrm{rowsp}(\mathbf{V})$. There exists $\mathbf{w}$ such that $\mathbf{u}\mathbf{w}^\mathrm{T} = 1$ and $\mathbf{V}\mathbf{w}^\mathrm{T} = 0$. Then, $\mathbf{Z}_{in}\mathbf{M}_{\overline{U}}^\mathrm{T}\overline{\mathbf{X}}(\mathbf{w}\mathbf{Z}_{out})^\mathrm{T} = 0$, i.e. $[Z_p,[\overline{U}^{-1}]\sigma(\mathbf{w}\mathbf{Z}_{out})] = 0$ holds for all $p\in Q_I$. Because $\sigma(\mathbf{w}\mathbf{Z}_{out})$ is a product of $Z$ operators on $Q_M$, $[Z_q,\sigma(\mathbf{w}\mathbf{Z}_{out})] = 0$ holds for all $q\in Q_M$. Then, according to Lemma \ref{lem:commutative2}, we can construct a codeword $\overline{\mathbf{a}}$ such that $\sigma(\overline{\mathbf{P}}_1\overline{\mathbf{a}}^\mathrm{T}) = \eta\left([\overline{U}^{-1}]\sigma(\mathbf{w}\mathbf{Z}_{out})\right)\times[\overline{U}^{-1}]\sigma(\mathbf{w}\mathbf{Z}_{out})$ and $\sigma(\overline{\mathbf{P}}_{T-1}\overline{\mathbf{a}}^\mathrm{T}) = \sigma(\mathbf{w}\mathbf{Z}_{out})$. Because $\sigma(\mathbf{w}\mathbf{Z}_{out})$ acts trivially on qubits not in $Q_M$, we have $\mathbf{a}=\varphi^{-1}(\overline{\mathbf{a}})\in C_{c,d}$. Then, $[\sigma(\overline{\mathbf{P}}_1\overline{\mathbf{a}}^\mathrm{T}),S] = 0$, i.e. $[[\overline{U}]S,\sigma(\mathbf{w}\mathbf{Z}_{out})] = 0$, which is in contradiction to $\mathbf{u}\mathbf{w}^\mathrm{T} = 1$ (i.e. $\overline{\mathbf{b}}\mathbf{M}_{\overline{U}}^\mathrm{T}\overline{\mathbf{X}}(\mathbf{w}\mathbf{Z}_{out})^\mathrm{T} = 1$). 

Because $\mathbf{u}\in \mathrm{rowsp}(\mathbf{V})$, there exists $\mathbf{w}$ such that $\mathbf{u} = \mathbf{w}\mathbf{V}$ (i.e. $(\overline{\mathbf{b}}+\mathbf{w}\mathbf{Z}_{in})\mathbf{M}_{\overline{U}}^\mathrm{T}\overline{\mathbf{X}}\mathbf{Z}_{out}^\mathrm{T} = 0$). We can define a new Pauli operator $S' = \sigma(\overline{\mathbf{b}}+\mathbf{w}\mathbf{Z}_{in}) = \eta\left(S\sigma(\mathbf{w}\mathbf{Z}_{in})\right)\times S\sigma(\mathbf{w}\mathbf{Z}_{in})$. It satisfies conditions: i) $[S',Z_p] = 0$ for all $p\in Q_I$; and ii) $[[\overline{U}]S',Z_q] = 0$ for all $q\in Q_M$. According to Lemma \ref{lem:commutative2}, we can construct a codeword $\overline{\mathbf{c}}'$ such that $\sigma(\overline{\mathbf{P}}_1{\overline{\mathbf{c}}'}^\mathrm{T}) = S'$ and $\sigma(\overline{\mathbf{P}}_{T-1}{\overline{\mathbf{c}}'}^\mathrm{T}) = \eta\left([\overline{U}]S'\right)\times[\overline{U}]S'$. Then, $\{\sigma(\overline{\mathbf{P}}_1\overline{\mathbf{c}}^\mathrm{T}),\sigma(\overline{\mathbf{P}}_1{\overline{\mathbf{c}}'}^\mathrm{T})\} = \{\sigma(\overline{\mathbf{P}}_{T-1}\overline{\mathbf{c}}^\mathrm{T}),\sigma(\overline{\mathbf{P}}_{T-1}{\overline{\mathbf{c}}'}^\mathrm{T})\} = 0$. According to expressions in Lemma \ref{lem:operators}, $\{\sigma_{in}(\mathbf{c}), \sigma_{in}(\mathbf{c}')\} = \{\sigma_{out}(\mathbf{c}), \sigma_{out}(\mathbf{c}')\} = 0$, where $\mathbf{c}' = \varphi^{-1}(\overline{\mathbf{c}}')$. The `only if’ part of has been proved. 

{\bf Corollary.} Because $\{\sigma_{in}(\mathbf{c}), \sigma_{in}(\mathbf{c}')\} = 0$, we have $\mathbf{c}'\notin C_{c,d,e}$ according to the `if’ part. Therefore, $\mathbf{c}'$ is a genuine propagator. The corollary has been proved. 

\section{Proofs of Theorem \ref{the:generalised}}
\label{app:generalised}

The proof of Theorem \ref{the:generalised} uses the extended circuit introduced in Appendix \ref{app:errorfree}. The spacetime error $\mathbf{e}$ is equivalent to adding Pauli gates $\sigma(\mathbf{X}\mathbf{P}_t\mathbf{e}^\mathrm{T})$ after layer-$t$ operations. With these erroneous Pauli gates, we have an erroneous Clifford circuit, which can be mapped to an erroneous extended circuit. Compared with the error-free extended circuit, $\overline{U} = \overline{U}_{T-1}\cdots\overline{U}_3\overline{U}_2$ is replaced by $\overline{U}'(\mathbf{e}) = \tau_T(\mathbf{e})\tau_{T-1}(\mathbf{e})\overline{U}_{T-1}\cdots\tau_3(\mathbf{e})\overline{U}_3\tau_2(\mathbf{e})\overline{U}_2\tau_1(\mathbf{e})\tau_0(\mathbf{e})$, where $\tau_t(\mathbf{e}) = \sigma(\mathbf{X}\mathbf{P}_t\mathbf{e}^\mathrm{T})\otimes I^{n_P}$. Accordingly, the final state reads $\rho_T(\mathbf{e},\boldsymbol{\mu}) = \mathrm{Tr}_{n_P}\left(\overline{\rho}_T(\mathbf{e},\boldsymbol{\mu})\right)$, where $\overline{\rho}_T(\mathbf{e},\boldsymbol{\mu}) = [E_{\boldsymbol{\mu}}][\overline{U}'(\mathbf{e})]\overline{\rho}_0$. 

Now, we consider the effect of errors on the Pauli operator $\sigma(\overline{\mathbf{P}}_{T-1}\overline{\mathbf{c}}^\mathrm{T})$. Let $\tilde{U}_t = U_{T+1}\cdots U_{t+2}U_{t+1}$, where $U_1 = U_T = U_{T+1} = \openone$. We define $\tau_t'(\mathbf{e}) = [\tilde{U}_t]\tau_t(\mathbf{e})$, which is the error at the time $T$ equivalent to the error $\tau_t(\mathbf{e})$ at the time $t$. The overall equivalent error is $\tilde{\tau}_{tot}(\mathbf{e}) = \tau'_T(\mathbf{e})\cdots\tau'_1(\mathbf{e})\tau'_0(\mathbf{e})$. Then, $\overline{U}'(\mathbf{e}) = \tilde{\tau}_{tot}(\mathbf{e})\overline{U}$. The effect of errors at the time $t$ is 
\begin{eqnarray}
[\tau'_t(\mathbf{e})^\dag]\sigma(\overline{\mathbf{P}}_{T-1}\overline{\mathbf{c}}^\mathrm{T}) &=& [\tilde{U}_t][\tau_t(\mathbf{e})^\dag][\tilde{U}_t^\dag]\sigma(\overline{\mathbf{P}}_{T-1}\overline{\mathbf{c}}^\mathrm{T}) \notag \\
&=& \eta\left([\tilde{U}_t^\dag]\sigma(\overline{\mathbf{P}}_{T-1}\overline{\mathbf{c}}^\mathrm{T})\right) \notag \\
&&\times [\tilde{U}_t][\tau_t(\mathbf{e})^\dag]\sigma(\overline{\mathbf{P}}_{t}\overline{\mathbf{c}}^\mathrm{T}).
\end{eqnarray}
The error $\tau_t(\mathbf{e})$ and the operator $\sigma(\overline{\mathbf{P}}_{t}\overline{\mathbf{c}}^\mathrm{T})$ are anti-commutative if and only if $(\mathbf{P}_t\mathbf{c}^\mathrm{T})^\mathrm{T}\mathbf{P}_t\mathbf{e}^\mathrm{T} = 1$. Therefore, 
\begin{eqnarray}
[\tau'_t(\mathbf{e})^\dag]\sigma(\overline{\mathbf{P}}_{T-1}\overline{\mathbf{c}}^\mathrm{T}) &=& (-1)^{(\mathbf{P}_t\mathbf{c}^\mathrm{T})^\mathrm{T}\mathbf{P}_t\mathbf{e}^\mathrm{T}} \notag \\
&&\times \eta\left([\tilde{U}_t^\dag]\sigma(\overline{\mathbf{P}}_{T-1}\overline{\mathbf{c}}^\mathrm{T})\right) [\tilde{U}_t]\sigma(\overline{\mathbf{P}}_{t}\overline{\mathbf{c}}^\mathrm{T}) \notag \\
&=& (-1)^{(\mathbf{P}_t\mathbf{c}^\mathrm{T})^\mathrm{T}\mathbf{P}_t\mathbf{e}^\mathrm{T}}\sigma(\overline{\mathbf{P}}_{T-1}\overline{\mathbf{c}}^\mathrm{T}).
\end{eqnarray}
For the overall error, 
\begin{eqnarray}
[\tilde{\tau}_{tot}(\mathbf{e})^\dag]\sigma(\overline{\mathbf{P}}_{T-1}\overline{\mathbf{c}}^\mathrm{T}) &=& (-1)^{\sum_{t=0}^T(\mathbf{P}_t\mathbf{c}^\mathrm{T})^\mathrm{T}\mathbf{P}_t\mathbf{e}^\mathrm{T}}\sigma(\overline{\mathbf{P}}_{T-1}\overline{\mathbf{c}}^\mathrm{T}) \notag \\
&=& (-1)^{\mathbf{c}\mathbf{e}^\mathrm{T}} \sigma(\overline{\mathbf{P}}_{T-1}\overline{\mathbf{c}}^\mathrm{T}).
\end{eqnarray}

Similar to the proof of Theorem \ref{the:errorfree}, we have 
\begin{eqnarray}
&& \sum_{\boldsymbol{\mu}} \mu_R(\mathbf{c},\boldsymbol{\mu})\mathrm{Tr}\sigma_{out}(\mathbf{c})\rho_T(\mathbf{e},\boldsymbol{\mu}) \notag \\
&=& \mathrm{Tr}\sigma(\overline{\mathbf{P}}_{T-1}\overline{\mathbf{c}}^\mathrm{T})[\overline{U}'(\mathbf{e})]\overline{\rho}_0.
\end{eqnarray}
We can rewrite the right-hand side as 
\begin{eqnarray}
&& \mathrm{Tr}\sigma(\overline{\mathbf{P}}_{T-1}\overline{\mathbf{c}}^\mathrm{T})[\overline{U}'(\mathbf{e})]\overline{\rho}_0 \notag \\
&=& \mathrm{Tr}\sigma(\overline{\mathbf{P}}_{T-1}\overline{\mathbf{c}}^\mathrm{T})[\tilde{\tau}(\mathbf{e})][\overline{U}]\overline{\rho}_0 \notag \\
&=& \mathrm{Tr}\left([\tilde{\tau}(\mathbf{e})^\dag]\sigma(\overline{\mathbf{P}}_{T-1}\overline{\mathbf{c}}^\mathrm{T})\right)\left([\overline{U}]\overline{\rho}_0\right) \notag \\
&=& (-1)^{\mathbf{c}\mathbf{e}^\mathrm{T}} \mathrm{Tr}\sigma(\overline{\mathbf{P}}_{T-1}\overline{\mathbf{c}}^\mathrm{T})[\overline{U}]\overline{\rho}_0.
\end{eqnarray}
As shown in Appendix \ref{app:errorfree}, $\mathrm{Tr}\sigma_{in}(\mathbf{c})\rho_0 = \nu\left(\mathbf{c}\right) \mathrm{Tr}\sigma(\overline{\mathbf{P}}_{T-1}\overline{\mathbf{c}}^\mathrm{T})[\overline{U}]\overline{\rho}_0$. Theorem \ref{the:generalised} has been proved. 

\section{Proof of Lemma \ref{lem:bit_splitting}}
\label{app:bit_splitting}

The isomorphism is obvious, and we note the details down because similar techniques will be used when discussing another Tanner-graph operation called symmetric splitting. Let $\mathbf{c}\in \mathrm{ker}\mathbf{A}$ and $\mathbf{c}'\in \mathrm{ker}\mathbf{A}'$. Because of the check $a$, $\mathbf{c}'$ always satisfies $\mathbf{c}'_v = \mathbf{c}'_{v'}$. Therefore, we take the map as follows: $\phi(\mathbf{c})_v = \phi(\mathbf{c})_{v'} = \mathbf{c}_v$ and $\phi(\mathbf{c})_u = \mathbf{c}_u$ for all other bits $u$. Then, $\phi$ is a linear bijection. 

Because $\phi$ is a linear bijection, $\mathbf{B}' = \phi(\mathbf{B})$ and $\mathbf{L}' = \phi(\mathbf{L})$ are compatible with $\mathbf{A}'$. 

Notice that $\mathbf{A}$ has $\vert V_B\vert$ columns, and $\mathbf{A}'$ has $\vert V_B\vert + 1$ columns. The error space of $\mathbf{A}$ is $\mathbb{F}_2^{\vert V_B\vert}$, and the error space of $\mathbf{A}'$ is $\mathbb{F}_2^{\vert V_B\vert + 1}$. 

We define a linear injection $\phi_{err}:\mathbb{F}_2^{\vert V_B\vert}\rightarrow \mathbb{F}_2^{\vert V_B\vert + 1}$ as follows. Let $\mathbf{e}\in \mathbb{F}_2^{\vert V_B\vert}$ be a spacetime error of $\mathbf{A}$. Then, $\phi_{err}(\mathbf{e})_v = \mathbf{e}_v$, $\phi_{err}(\mathbf{e})_{v'} = 0$, and $\phi_{err}(\mathbf{e})_u = \mathbf{e}_u$ for all other bits $u$. The map $\phi_{err}$ has the following properties: i) For all codewords $\mathbf{c} = \mathrm{ker}\mathbf{A}$, $\mathbf{c}\mathbf{e}^\mathrm{T} = \phi(\mathbf{c})\phi_{err}(\mathbf{e})^\mathrm{T}$; and ii) $\vert \mathbf{e}\vert = \vert \phi_{err}(\mathbf{e})\vert$. Therefore, 
\begin{eqnarray}
d(\mathbf{A},\mathbf{B},\mathbf{L}) = \min_{\mathbf{e}'\in \mathrm{im}\phi_{err}\cap (\mathrm{ker}\mathbf{B}' - \mathrm{ker}\mathbf{B}'\cap \mathrm{ker}\mathbf{L}')} \vert\mathbf{e}'\vert.
\end{eqnarray}

All errors in $\mathbb{F}_2^{\vert V_B\vert + 1}$ are equivalent to errors in $\mathrm{im}\phi_{err}$. Suppose $\mathbf{e}''\in \mathbb{F}_2^{\vert V_B\vert + 1}-\mathrm{im}\phi_{err}$ (i.e. $\mathbf{e}''_{v'} = 1$). Its equivalent error in $\mathrm{im}\phi_{err}$ is $\mathbf{e}'$: $\mathbf{e}'_v = \mathbf{e}''_v + 1$, $\mathbf{e}'_{v'} = 0$, and $\mathbf{e}'_u = \mathbf{e}''_u$ for all other bits $u$. We can find that $\vert \mathbf{e}'\vert\leq \vert \mathbf{e}''\vert$. Then, 
\begin{eqnarray}
d(\mathbf{A}',\mathbf{B}',\mathbf{L}') = \min_{\mathbf{e}'\in \mathrm{im}\phi_{err}\cap (\mathrm{ker}\mathbf{B}' - \mathrm{ker}\mathbf{B}'\cap \mathrm{ker}\mathbf{L}')} \vert\mathbf{e}'\vert.
\end{eqnarray}
The lemma is proved. 

\section{Proof of Theorem \ref{the:symmetric_splitting}}
\label{app:symmetric_splitting}

The linear bijection is defined in Sec. \ref{sec:preservation}. Here, we only need to prove the bound of the circuit code distance. 

The proof is similar to the proof of Lemma \ref{lem:bit_splitting}. First, 
\begin{eqnarray}
d(\mathbf{A},\mathbf{B},\mathbf{C}) = \min_{\mathbf{e}'\in \mathrm{im}\psi_{err}\cap (\mathrm{ker}\mathbf{B}' - \mathrm{ker}\mathbf{B}'\cap \mathrm{ker}\mathbf{C}')} \vert\mathbf{e}'\vert,
\end{eqnarray}
due to properties of $\psi_{err}$. 

Let $\mathbf{e}''\in \mathbb{F}_2^{\vert V_B'\vert}$ be a general spacetime error. Let $\{\mathbf{t}^{(u)}\in \mathbb{F}_2^{\vert V_B'\vert}\vert u\in V_B'\}$ be the set of single-bit errors: $\mathbf{t}^{(u)}$ is the single-bit error on $u$, i.e. $\mathbf{t}^{(u)}_v = \delta_{u,v}$. Then, we can decompose the general error in the form $\mathbf{e}'' = \sum_{u\in V_B'} \mathbf{e}''_u\mathbf{t}^{(u)}$. 

For each $\mathbf{t}^{(u)}$, there is an equivalent error $\mathbf{s}^{(u)}\in \mathrm{im}\psi_{err}$. When $u$ is a long terminal, $\mathbf{s}^{(u)} = \mathbf{t}^{(u)}$. When $u$ is on the bit tree of $v$, $\mathbf{s}^{(u)} = \mathbf{t}^{(\hat{N}_{v,1})}$. When $u$ is on the check tree of $a$, $\mathbf{s}^{(u)} = \sum_{v\in D_{a,u}}\mathbf{s}^{(v)}$ if $\vert D_{a,u}\vert \leq \lfloor g_a/2\rfloor$, and $\mathbf{s}^{(u)} = \sum_{v\in D_a-D_{a,u}}\mathbf{s}^{(v)}$ if $\vert D_{a,u}\vert > \lfloor g_a/2\rfloor$. Here, $D_a$ is the set of all bit leaves on the dressed tree of $a$. Therefore, $\vert \mathbf{s}^{(u)}\vert \leq \lfloor g_{max}/2\rfloor$ for all $u$. 

The error $\mathbf{e}' = \sum_{u\in V_B'} \mathbf{e}''_u\mathbf{s}^{(u)} \in \mathrm{im}\psi_{err}$ is equivalent to $\mathbf{e}''$, and $\mathbf{e}''$ is a logical error if and only if $\mathbf{e}'$ is a logical error. When $\mathbf{e}'$ is a logical error, $d(\mathbf{A},\mathbf{B},\mathbf{C})\leq \vert\mathbf{e}'\vert 
\leq \lfloor g_{max}/2\rfloor\vert\mathbf{e}''\vert$. The theorem has been proved. 

\section{Proof of Theorem \ref{the:circuit_construction}}
\label{app:circuit_construction1}

There are two Tanner graphs in the theorem. One is the given Tanner graph, and the other is the symmetric Tanner graph of the constructed circuit. We will prove that the second is generated from the first through symmetric splitting, in which bit and check trees are paths. For clarity, we call vertices on the given Tanner graph {\it macroscopic vertices}, and we call paths used in the symmetric splitting microscopic paths. 

\begin{figure}[tbp]
\centering
\includegraphics[width=\linewidth]{\figpath/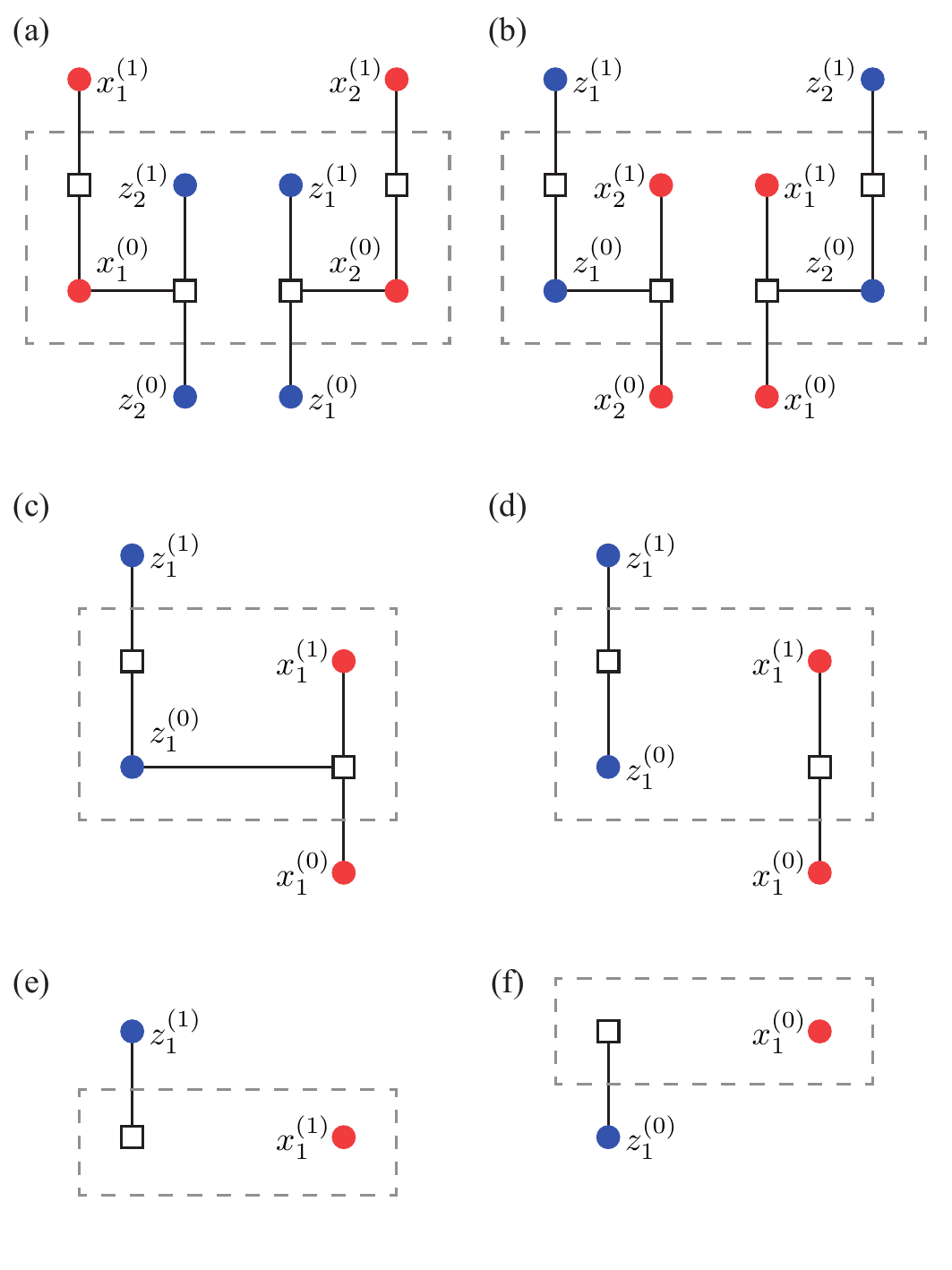}
\caption{
Tanner graphs of gates (a) $H_{2}\Lambda_{1,2}H_2$, (b) $H_1\Lambda_{1,2}H_1$, (c) $H_1S_1H_1$, (d) identity gate, (e) initialisation in the $X$ basis and (f) measurement in the $X$ basis. 
}
\label{fig:more_operations}
\end{figure}

First, we give the Tanner graph of each operation in the constructed circuit. Tanner graphs of the controlled-NOT gate, phase gate, initialisation and measurement in the $Z$ basis are in Figs. \ref{fig:gates_symmetry}(a), \ref{fig:gates_symmetry}(d), \ref{fig:IandM_symmetry}(a) and \ref{fig:IandM_symmetry}(b), respectively. Tanner graphs of other operations are in Fig. \ref{fig:more_operations}, including two-qubit gates $H_2\Lambda_{1,2}H_2$ and $H_1\Lambda_{1,2}H_1$ (taking $q = 1$ and $q' = 2$), the single-qubit gate $H_1S_1H_1$, initialisation in the $X$ basis and measurement in the $X$ basis. In addition to these non-trivial operations, Fig. \ref{fig:gates_symmetry}(b) illustrates the Tanner graph of the identity gate, and Fig. \ref{fig:more_operations}(d) illustrates an alternative way of drawing the same Tanner graph. 

With these Tanner graphs of primitive operations, we can generate the symmetric Tanner graph of the circuit. Let's focus on one qubit. Operations on qubit-$q$ are determined by the $X_q$ and $Z_q$ paths. We only needs to look at the $X_q$ path because of the symmetry. 

\begin{figure}[tbp]
\centering
\includegraphics[width=\linewidth]{\figpath/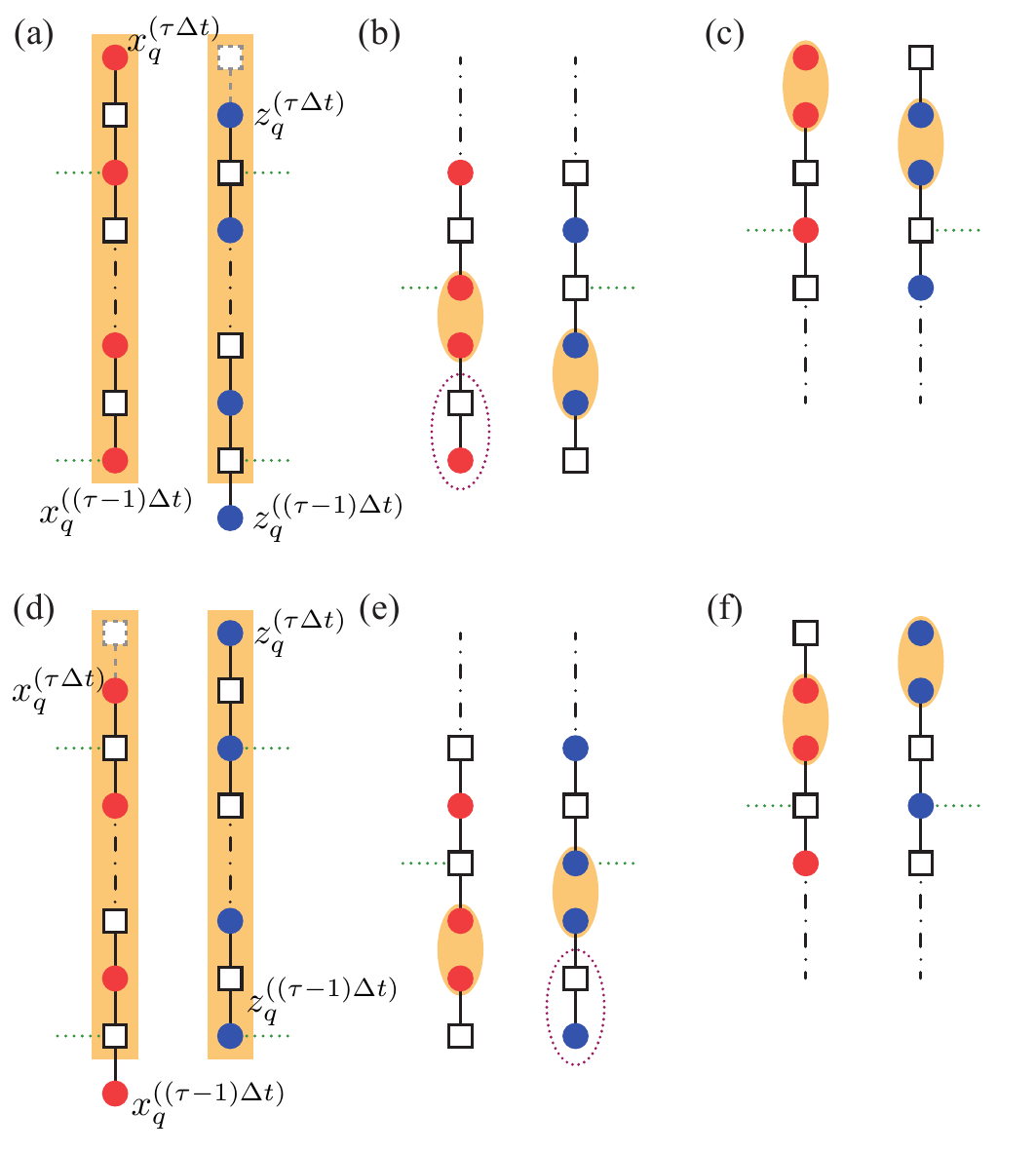}
\caption{
(a) Subgraph corresponding to a macroscopic bit on the $X_q$ path with the time label $\tau$. Two microscopic paths are marked by the orange boxes. The grey dashed square is a check provided by an operation in the next time window. Green dotted lines denote inter-path edges. (b) Merging the subgraph in (a) with the graph of the initialisation in the $X$ basis. Vertices in the purple dotted oval are introduced by bit splitting. (c) Merging the subgraph in (a) with the graph of the measurement in the $X$ basis. (d) Subgraph corresponding to a macroscopic check on the $X_q$ path with the time label $\tau$. The grey dashed square is provided by an operation in the next time window. (e) Merging the subgraph in (d) with the graph of the initialisation in the $Z$ basis. (f) Merging the subgraph in (d) with the graph of the measurement in the $Z$ basis. 
}
\label{fig:microscopic1}
\end{figure}

Consider a macroscopic bit on the $X_q$ path. Suppose its time label is $\tau$. Then, according to the universal approach of circuit construction, it is mapped to gates in the time window from $t = (\tau-1)\Delta t + 1$ to $t = \tau\Delta t$. According to Table \ref{table:gates}, the gates are $S_q$, $\Lambda_{q,q'}$, $H_{q'}\Lambda_{q,q'}H_{q'}$ and the identity gate. For all these gates, the input $x_q$ bit is always a short terminal, and the output $x_q$ bit is always a long terminal (notice that there are two ways of drawing the Tanner graph of the identity gate). Therefore, merging between their Tanner graphs is always symmetric without using bit splitting. The subgraph of qubit-$q$ in the time window $\tau$ is shown in Fig. \ref{fig:microscopic1}(a), which illustrates two microscopic paths. There are $\Delta t + 1$ bits (checks) on the $X_q$ ($Z_q$) microscopic path, and all inter-path edges are incident on bits (checks) on the $X_q$ ($Z_q$) microscopic path. Each inter-path edge corresponds to a non-identity gate on qubit-$q$ in the time window. We can find that the $X_q$ ($Z_q$) microscopic path is a valid bit (check) tree: All checks (bits) have a degree of two, and all inter-tree edges are incident on bits (checks). 

If the macroscopic bit is an end of the $X_q$ path, it may also correspond to initialisation and measurement operations. With the initialisation and measurement operations, the $X_q$ ($Z_q$) microscopic path is still a valid check (bit) tree, as shown in Figs. \ref{fig:microscopic1}(b) and (c). Similarly, for a macroscopic check on the $X_q$ path, the corresponding $X_q$ ($Z_q$) microscopic path is a valid check (bit) tree, as shown in Figs. \ref{fig:microscopic1}(d), (e) and (f). 

\begin{figure}[tbp]
\centering
\includegraphics[width=\linewidth]{\figpath/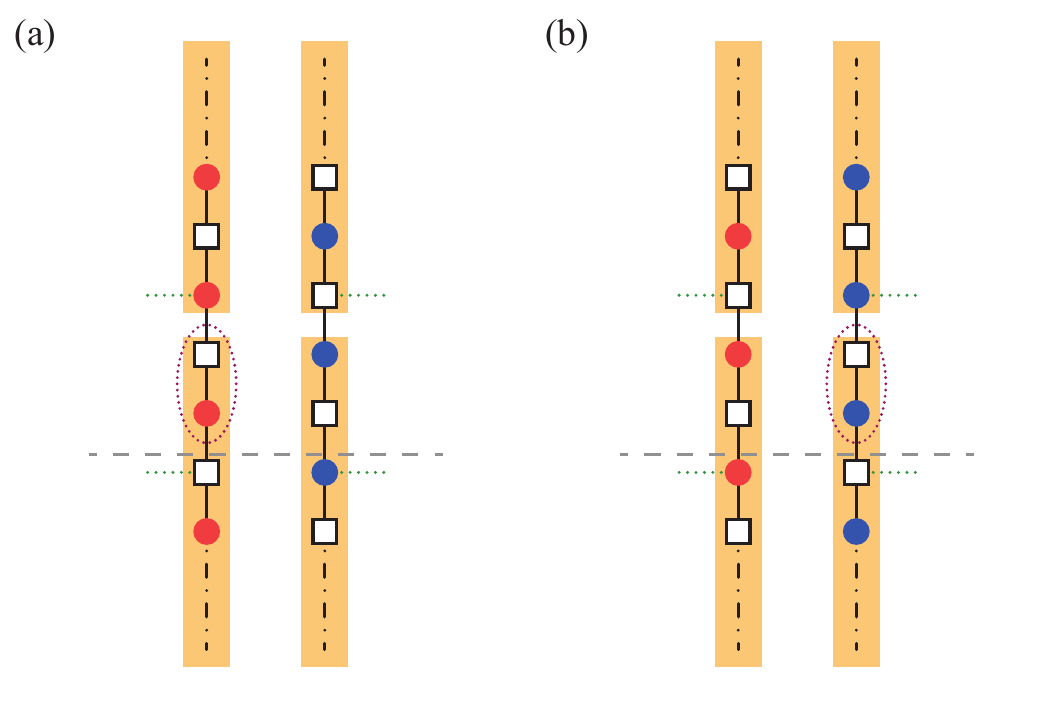}
\caption{
Merging subgraphs of two adjacent macroscopic vertices. Microscopic paths are marked by the orange boxes. Vertices in the purple dotted ovals are introduced by bit splitting. In (a) [(b)], the upper (lower) microscopic paths correspond to the macroscopic bit, and the lower (upper) microscopic paths correspond to the macroscopic check. If we cut the graph along the grey dashed line and rechoose microscopic paths accordingly, they are still valid bit and check trees. 
}
\label{fig:microscopic2}
\end{figure}

Now, let's consider two adjacent macroscopic vertices on the $X_q$ path. One of them is a macroscopic bit with the time label $\tau$, and the other is a macroscopic check with the time label $\tau'$. When $\tau = \tau'+1$, their subgraphs are merged as shown in Fig. \ref{fig:microscopic2}(a). Because two long (short) terminals are merged into one bit, merging is asymmetric, and bit splitting is used. We can find that all vertices belong to certain microscopic paths. It is similar when $\tau = \tau'-1$, as shown in Fig. \ref{fig:microscopic2}(b). 

In Fig. \ref{fig:microscopic2}(a), there is not any inter-path edge on the two vertices in the purple dotted oval. Therefore, we can choose microscopic paths in an alternative way, and they are still valid bit and check trees: Move the two vertices in the purple dotted oval and their dual vertices from lower microscopic paths to upper microscopic paths. Therefore, with such segments without any inter-path edge, there is a degree of freedom for choosing microscopic paths. 

In some time windows, only the identity gate is applied on qubit-$q$. If the first operation on qubit-$q$ is not initialisation, the identity gate is applied from the beginning to $t = (\tau_{min}-1)\Delta t$. If the last operation on qubit-$q$ is not measurement, the identity gate is applied from $t = \tau_{max}\Delta t 
+ 1$ to the end. For two adjacent macroscopic vertices on the $X_q$ path, if $\vert\tau-\tau'\vert > 1$, the identity gate is applied between the two time windows. The subgraph for a sequence of identity gates is a pair of paths without any inter-path edge. Therefore, they can be taken into account in adjacent microscopic paths. 

In summary, each pair of dual macroscopic vertices corresponds to a pair of dual microscopic paths. Microscopic paths are valid bit and check trees. Except for long terminals, all vertices belong to microscopic paths on the symmetric Tanner graph of the constructed circuit. With the observation that on the two Tanner graphs inter-path edges are consistent, the theorem is proved. 

\RED{\section{Block matrices in the resource-efficient protocol}
\label{app:matrices}

\begin{widetext}
Matrices $\boldsymbol{\beta}_{in/out}$ and $\boldsymbol{\gamma}_{\alpha,in/out}$ ($\alpha = X,Z$) are defined as 
\begin{eqnarray}
\boldsymbol{\beta}_{in} = \left(\begin{matrix}
\openone_{n_{in}n_2} & 0_{n_{in}n_2\times n_{out}n_2} & 0_{n_{in}n_2\times r_{in}r_2} & 0_{n_{in}n_2\times r_{out}r_2} & 0_{n_{in}n_2\times m_Gr_2} \\
0_{r_{in}r_2\times n_{in}n_2} & 0_{r_{in}r_2\times n_{out}n_2} & \openone_{r_{in}r_2} & 0_{r_{in}r_2\times r_{out}r_2} & 0_{r_{in}r_2\times m_Gr_2}
\end{matrix}\right),
\end{eqnarray} 
\begin{eqnarray}
\boldsymbol{\beta}_{out} = \left(\begin{matrix}
0_{n_{out}n_2\times n_{in}n_2} & \openone_{n_{out}n_2} & 0_{n_{out}n_2\times r_{in}r_2} & 0_{n_{out}n_2\times r_{out}r_2} & 0_{n_{out}n_2\times m_Gr_2} \\
0_{r_{out}r_2\times n_{in}n_2} & 0_{r_{out}r_2\times n_{out}n_2} & 0_{r_{out}r_2\times r_{in}r_2} & \openone_{r_{out}r_2} & 0_{r_{out}r_2\times m_Gr_2}
\end{matrix}\right),
\end{eqnarray}
and 
\begin{eqnarray}
\boldsymbol{\gamma}_{X,in} &=& \left(\begin{matrix}
\openone_{n_{in}r_2} & 0_{n_{in}r_2\times n_{out}r_2}
\end{matrix}\right), \\
\boldsymbol{\gamma}_{X,out} &=& \left(\begin{matrix}
0_{n_{out}r_2\times n_{in}r_2} & \openone_{n_{out}r_2}
\end{matrix}\right), \\
\boldsymbol{\gamma}_{Z,in} &=& \left(\begin{matrix}
\openone_{r_{in}n_2} & 0_{r_{in}n_2\times r_{out}n_2} & 0_{r_{in}n_2\times m_Gn_2}
\end{matrix}\right), \\
\boldsymbol{\gamma}_{Z,out} &=& \left(\begin{matrix}
0_{r_{out}n_2\times r_{in}n_2} & \openone_{r_{out}n_2} & 0_{r_{out}n_2\times m_Gn_2}
\end{matrix}\right).
\end{eqnarray}
When $d_T = 3$, matrices $\mathbf{A}_\alpha$ and $\mathbf{B}_\alpha$ are in the form 
\begin{eqnarray}
\mathbf{A}_X = \left(\begin{matrix}
\openone & \boldsymbol{\beta}_{in} & 0 & 0 & 0 & \mathbf{G}_{X,in}^\mathrm{T} & 0 & 0 & 0 \\
0 & \openone & \openone & 0 & 0 & 0 & \overline{\mathbf{G}}_X^\mathrm{T} & 0 & 0 \\
0 & 0 & \openone & \openone & 0 & 0 & 0 & \overline{\mathbf{G}}_X^\mathrm{T} & 0 \\
0 & 0 & 0 & \boldsymbol{\beta}_{out} & \openone & 0 & 0 & 0 & \mathbf{G}_{X,out}^\mathrm{T} \\
\mathbf{G}_{Z,in} & 0 & 0 & 0 & 0 & 0 & 0 & 0 & 0 \\
0 & \overline{\mathbf{G}}_Z & 0 & 0 & 0 & 0 & 0 & 0 & 0 \\
0 & 0 & \overline{\mathbf{G}}_Z & 0 & 0 & 0 & 0 & 0 & 0 \\
0 & 0 & 0 & \overline{\mathbf{G}}_Z & 0 & 0 & 0 & 0 & 0 \\
0 & 0 & 0 & 0 & \mathbf{G}_{Z,out} & 0 & 0 & 0 & 0
\end{matrix}\right),
\end{eqnarray}
\begin{eqnarray}
\mathbf{A}_Z = \left(\begin{matrix}
\openone & \openone & 0 & 0 & 0 & 0 & \mathbf{G}_{Z,in}^\mathrm{T} & 0 & 0 & 0 & 0 \\
0 & \boldsymbol{\beta}_{in}^\mathrm{T} & \openone & 0 & 0 & 0 & 0 & \overline{\mathbf{G}}_Z^\mathrm{T} & 0 & 0 & 0 \\
0 & 0 & \openone & \openone & 0 & 0 & 0 & 0 & \overline{\mathbf{G}}_Z^\mathrm{T} & 0 & 0 \\
0 & 0 & 0 & \openone & \boldsymbol{\beta}_{out}^\mathrm{T} & 0 & 0 & 0 & 0 & \overline{\mathbf{G}}_Z^\mathrm{T} & 0 \\
0 & 0 & 0 & 0 & \openone & \openone & 0 & 0 & 0 & 0 & \mathbf{G}_{Z,out}^\mathrm{T} \\
0 & \mathbf{G}_{X,in} & 0 & 0 & 0 & 0 & 0 & 0 & 0  & 0 & 0\\
0 & 0 & \overline{\mathbf{G}}_X & 0 & 0 & 0 & 0 & 0 & 0 & 0 & 0 \\
0 & 0 & 0 & \overline{\mathbf{G}}_X & 0 & 0 & 0 & 0 & 0 & 0 & 0 \\
0 & 0 & 0 & 0 & \mathbf{G}_{X,out} & 0 & 0 & 0 & 0 & 0 & 0
\end{matrix}\right),
\end{eqnarray}
\begin{eqnarray}
\mathbf{B}_X = \left(\begin{matrix}
\mathbf{G}_{X,in} & 0 & 0 & 0 & 0 & \openone & 0 & 0 & 0 \\
0 & \overline{\mathbf{G}}_X & 0 & 0 & 0 & \boldsymbol{\gamma}_{X,in}^\mathrm{T} & \openone & 0 & 0 \\
0 & 0 & \overline{\mathbf{G}}_X & 0 & 0 & 0 & \openone & \openone & 0 \\
0 & 0 & 0 & \overline{\mathbf{G}}_X & 0 & 0 & 0 & \openone & \boldsymbol{\gamma}_{X,out}^\mathrm{T} \\
0 & 0 & 0 & 0 & \mathbf{G}_{X,out} & 0 & 0 & 0 & \openone
\end{matrix}\right),
\end{eqnarray}
and 
\begin{eqnarray}
\mathbf{B}_Z = \left(\begin{matrix}
\mathbf{G}_{Z,in} & 0 & 0 & 0 & 0 & 0 & \openone & 0 & 0 & 0 & 0 \\
0 & \mathbf{G}_{Z,in} & 0 & 0 & 0 & 0 & \openone & \boldsymbol{\gamma}_{Z,in} & 0 & 0 & 0 \\
0 & 0 & \overline{\mathbf{G}}_Z & 0 & 0 & 0 & 0 & \openone & \openone & 0 & 0 \\
0 & 0 & 0 & \overline{\mathbf{G}}_Z & 0 & 0 & 0 & 0 & \openone & \openone & 0 \\
0 & 0 & 0 & 0 & \mathbf{G}_{Z,out} & 0 & 0 & 0 & 0 & \boldsymbol{\gamma}_{Z,out} & \openone \\
0 & 0 & 0 & 0 & 0 & \mathbf{G}_{Z,out} & 0 & 0 & 0 & 0 & \openone
\end{matrix}\right).
\end{eqnarray}
When $d_T = 3$ and $l = 2$, matrices $\mathbf{L}_X$ and $\mathbf{L}_Z^{(l)}$ are in the form 
\begin{eqnarray}
\mathbf{L}_X = \left(\begin{matrix}
(\mathbf{g}_{X,in}\otimes\openone)\mathbf{J}_{X,in} & \overline{\mathbf{J}}_X & \overline{\mathbf{J}}_X & \overline{\mathbf{J}}_X & (\mathbf{g}_{X,out}\otimes\openone)\mathbf{J}_{X,out} & 0 & 0 & 0 & 0
\end{matrix}\right)
\end{eqnarray}
and 
\begin{eqnarray}
\mathbf{L}_Z^{(l)} = \Biggl(\begin{matrix}
(\mathbf{g}_{Z,in}\otimes\openone)\left(\begin{matrix}\mathbf{J}_{Z,in} & \mathbf{J}_{Z,in} & \mathbf{J}_{Z,in}\boldsymbol{\beta}_{in}\end{matrix}\right) & (\mathbf{g}_{Z,out}\otimes\openone)\left(\begin{matrix}\mathbf{J}_{Z,out}\boldsymbol{\beta}_{out} & \mathbf{J}_{Z,out} & \mathbf{J}_{Z,out}\end{matrix}\right) & 0 & 0 & \boldsymbol{\delta} & 0 & 0
\end{matrix}\Biggr).
\end{eqnarray}
\end{widetext}

\bibliography{LDPC_rep}

\end{document}